\documentclass[twoside]{article}

\usepackage[accepted]{aistats2024}
\usepackage[round]{natbib}

\usepackage{amsmath}
\usepackage{amsthm}
\usepackage{amssymb}
\usepackage{bbm}
\usepackage{stmaryrd}
\usepackage{algorithm}
\usepackage{algorithmic}
\usepackage{titletoc}
\usepackage{hyperref}
\usepackage[noabbrev,nameinlink,capitalise]{cleveref}
\usepackage{nicefrac}
\usepackage{graphicx}
\usepackage{booktabs}
\usepackage{adjustbox}
\usepackage{eqparbox}

\usepackage{xcolor}

\newtheorem{theorem}{Theorem}[section]

\newtheorem{lemma}[theorem]{Lemma}
\newtheorem{corollary}[theorem]{Corollary}
\newtheorem{definition}[theorem]{Definition}

\usepackage{xcolor}
\definecolor{shapiq}{rgb}{0.937,0.153,0.651}
\definecolor{svarmiq}{rgb}{0,0.706,0.847}
\definecolor{baseline}{rgb}{0.49,0.325,0.871}

\begin{document}

%
\runningtitle{SVARM-IQ: Efficient Approximation of Any-order Shapley Interactions through Stratification}

%
\runningauthor{Kolpaczki, Muschalik, Fumagalli, Hammer, Hüllermeier}

\twocolumn[

\aistatstitle{SVARM-IQ: Efficient Approximation of Any-order \\ Shapley Interactions through Stratification}

\aistatsauthor{Patrick Kolpaczki\\Paderborn University \And Maximilian Muschalik\\University of Munich (LMU)\\Munich Center for Machine Learning\And Fabian Fumagalli\\CITEC\\Bielefeld University\AND Barbara Hammer\\CITEC\\Bielefeld University\And Eyke Hüllermeier\\University of Munich (LMU)\\Munich Center for Machine Learning}
\aistatsaddress{ } 

]

\begin{abstract}    
    Addressing the limitations of individual attribution scores via the Shapley value (SV), the field of explainable AI (XAI) has recently explored intricate interactions of features or data points.
    In particular, \mbox{extensions}~of~the SV, such as the Shapley Interaction Index (SII), have been proposed as a measure
    to still benefit from the axiomatic basis of the SV.
    However, similar to the SV, their exact computation remains computationally prohibitive.
    Hence, we propose with SVARM-IQ a sampling-based approach to efficiently approximate Shapley-based interaction indices of any order.
    SVARM-IQ can be applied to a broad class of interaction indices, including the SII, by leveraging a novel stratified representation.
    We provide non-asymptotic theoretical guarantees on its approximation quality and empirically demonstrate that SVARM-IQ achieves state-of-the-art estimation results in practical XAI scenarios on different model classes and application domains.
\end{abstract}

\section{INTRODUCTION} \label{sec:Introduction}

Interpreting black box machine learning (ML) models via feature attribution scores is a widely applied technique in the field of explainable AI (XAI) \citep{Adadi.2022,Covert_Lundberg_Lee_2021,chen2023algorithms}.
However, in real-world applications, such as genomics \citep{Wright.2016} or tasks involving natural language \citep{Tsang.2020b}, isolated  features are less meaningful.
In fact, it was shown that, in the presence of strong feature correlation or higher order interactions, feature attribution scores are not sufficient to capture the reasoning of a trained ML model
\citep{Wright.2016,Slack.2020,Sundararajan.2020b,Kumar.2020,Kumar.2021}.
As a remedy, \emph{feature interactions} extend feature attributions to arbitrary groups of features (see \cref{fig_intro_example}).

\begin{figure}[t]
    \centering
    \includegraphics[width=0.82\columnwidth]{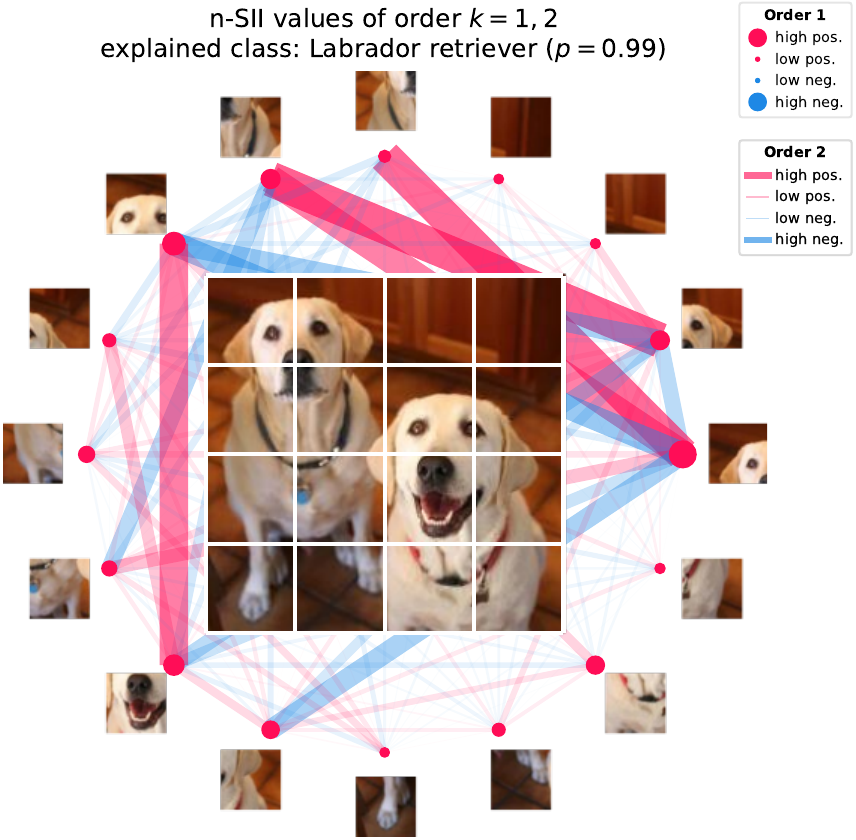}
    \caption{By dividing an ImageNet picture into multiple patches, attribution scores for single patches and interactions scores for pairs aid explaining a vision transformer.
    }
    \label{fig_intro_example}
\end{figure}

A prevalent approach to define feature attributions is based on the Shapley value (SV) \citep{Shapley.1953}, an axiomatic concept from cooperative game theory that fairly distributes the payout achieved by a group among its members.
Extensions of the SV to Shapley-based interaction indices, i.e., interaction indices that reduce to the SV for single players, have been proposed \citep{Grabisch_Roubens_1999,pmlr-v206-bordt23a,Sundararajan_Dhamdhere_Agarwal_2020,Tsai_Yeh_Ravikumar_2022}.
Yet, the exact computation of the SV and Shapley-based interactions without further assumptions on the ML model quickly becomes infeasible due to its exponential complexity \citep{Deng.1994}.

In this work, we present \emph{SVARM Interaction Quantification} (SVARM-IQ), a novel approximation technique for a broad class of interaction indices, including Shapley-based interactions, which is applicable to any cooperative game.
SVARM-IQ extends Stratified SVARM \citep{Kolpaczki.2023} to any-order interactions by introducing a novel representation of interaction indices through stratification.

\paragraph{Contribution.}
Our core contributions include:
\begin{enumerate}
    \item \emph{SVARM-IQ} (\cref{sec_method}): A model-agnostic approximation algorithm for estimating Shapley-based interaction scores of any order through leveraging a \emph{stratified} representation.
    \item \emph{Theoretical Analysis} (\cref{sec:TheoreticalResults}): 
    We prove, under mild assumptions, that SVARM-IQ is unbiased and provide bounds on the approximation error.
    \item \emph{Application} (\cref{sec_experiments}):
    An open-source implementation\footnote{\url{https://github.com/kolpaczki/svarm-iq}} and empirical evaluation demonstrating SVARM-IQ's superior approximation quality over state-of-the-art techniques.
\end{enumerate}

\paragraph{Related work.}
In cooperative game theory, Shapley-based interactions, as an extension to the SV, were first proposed with the Shapley-Interaction index (SII) \citep{Grabisch_Roubens_1999}.
Besides the SII, the Shapley-Taylor Interaction index (STI) \citep{Sundararajan_Dhamdhere_Agarwal_2020} and Faithful Shapley-Interaction index (FSI) \citep{Tsai_Yeh_Ravikumar_2022} were introduced, which, in contrast to the SII, directly require the efficiency axiom.
Beyond Shapley-based interaction indices, extensions of the Banzhaf value were studied by \cite{Hammer_Holzman_1992}.
In ML, limitations of feature attribution scores have been discussed in \cite{Wright.2016}, \cite{Sundararajan.2020b}, and \cite{Kumar.2020,Kumar.2021} among others. 
Model-specific interaction measures have been proposed for neural networks \citep{Tsang.2018,Singh.2019,Janizek.2021}.
Model-agnostic measures were introduced via functional decomposition \citep{Hooker.2004,Hooker.2007} in \citep{Lou.2013,Molnar.2019,Lengerich.2020,Hiabu.2023}.
Applications include complex language \citep{Murdoch.2018} and image classification \citep{Tsang.2020b} models, as well as application domains, such as gene interactions \citep{Wright.2016}.
Besides pure explanation purposes, e.g.\ understanding sentiment predictions from NLP models \citep{Fumagalli.2023}, \cite{Chu.2020} leveraged the SII to improve feature selection for tree classifiers. 

Approximation techniques for the SV have been proposed via permutation sampling \citep{Castro.2009}, which has been extended to the SII and STI \citep{Sundararajan_Dhamdhere_Agarwal_2020,Tsai_Yeh_Ravikumar_2022}.
For the SV, \cite{Castro.2017} demonstrated the impact of stratification on approximation performance.
Alternatively, the SV can be represented as a solution to a least squares problem \citep{Charnes_Golany_Keane_Rousseau_1988}, which was exploited for approximation \citep{Lundberg.2017,Covert.2021} and extended to FSI \citep{Tsai_Yeh_Ravikumar_2022}.
Recent work proposed a model-agnostic sampling-based approach \citep{Fumagalli.2023} for Shapley-based interactions, which was further linked to \cite{Covert.2021}.
On the model-specific side \cite{Muschalik.2024} extended the polynomial-time exact computation of the SV for local feature importance in decision trees \citep{Lundberg.2020} to the SII.
While permutation-based approaches are restricted to update single estimates, \cite{Kolpaczki.2023} proposed wit Stratified SVARM a novel approach for the SV that is capable of updating all estimates using only a single value function call.

\section{SHAPLEY-BASED INTERACTION INDICES}
\label{sec:InteractionIndices}
In the following, we are interested in properties of a \emph{cooperative game}, that is a 
tuple $(\mathcal{N}, \nu)$ containing a \emph{player set} $\mathcal{N} = \{1,\ldots,n\}$ with $n \in \mathbb{N}$ players and a \textit{value function} $\nu : 2^{\mathcal{N}} \to \mathbb{R}$ mapping each subset $S \subseteq \mathcal{N}$ of players, also called coalition, to a real-valued number $\nu(S)$.
In the field of XAI, the value function typically represents a specific \emph{model behavior} \citep{Covert_Lundberg_Lee_2021}, such as the prediction of an instance or the dataset loss.
The player set represents the entities whose attribution will be determined, e.g., the contribution of features to a prediction or the dataset loss.
To determine the worth of individual players, the Shapley value (SV) \citep{Shapley.1953} can be expressed as a weighted average over marginal contributions.

\begin{definition}[Shapley Value \citep{Shapley.1953}]
The SV is 
\begin{equation*}
    \phi_i = \sum\limits_{S \subseteq \mathcal{N} \setminus \{i\}} \frac{1}{n \binom{n-1}{|S|}} \Delta_i(S),
\end{equation*}
where $i \in \mathcal N$ and $\Delta_i(S) := \nu(S \cup \{i\}) - \nu(S)$.
\end{definition}

The SV is provably the unique attribution measure that fulfills the following axioms: linearity (linear combinations of value functions yield linear combinations of attribution), dummy (players that do not impact the worth of any coalition receive zero attribution), symmetry (two players contributing equally to all coalitions receive the same attribution), and efficiency (the sum of of all players' attributions equals the worth of all players) \citep{Shapley.1953}.
In many ML related applications, however, the attribution via the SV is limited 
in the presence of strong feature correlation or higher order interaction \citep{Slack.2020,Sundararajan.2020b,Kumar.2020,Kumar.2021}.
It is therefore necessary to study \emph{interactions between players} in cooperative games.
The SV is a weighted average of marginal contributions $\Delta_i$ of single players, and a natural extension to pairs of players is 
\begin{equation*}
    \Delta_{i,j}(S) :=\nu(S \cup \{i,j\}) - \nu(S) - \Delta_i(S) - \Delta_j(S)
\end{equation*}
for $S \subseteq \mathcal{N} \setminus \{i,j\}$.
Generalizing this recursion to higher order interactions yields the following definition.

\begin{definition}[Discrete Derivative \citep{DBLP:journals/geb/FujimotoKM06}]
    For $K \subseteq \mathcal N$, the \emph{K-derivative} of $\nu$ at $S \subseteq \mathcal N \setminus K$ is 
    \begin{equation*}
    \Delta_K(S) := \sum\limits_{W \subseteq K} (-1)^{\vert K \vert-\vert W \vert} \cdot \nu(S \cup W).
\end{equation*}
\end{definition}

The Shapley interaction index (SII) was the first axiomatic extension of the SV to higher order interaction \citep{Grabisch_Roubens_1999}.
It can be represented as a weighted average of discrete derivatives.

\begin{definition}[Shapley Interaction Index \citep{Grabisch_Roubens_1999}]
    The SII of $K\subseteq \mathcal N$ is defined as
    \begin{equation*}
        I^{\text{SII}}_K =  \sum\limits_{S \subseteq \mathcal{N} \setminus K} \frac{1}{(n-\vert K \vert+1)\binom{n-\vert K \vert}{\vert S \vert}}\Delta_{K}(S).
    \end{equation*}
\end{definition}

\paragraph{Cardinal Interaction Indices.}
Besides the SII, the Shapley-Taylor interaction index (STI) \citep{Sundararajan_Dhamdhere_Agarwal_2020} and Faithful Shapley interaction index (FSI) \citep{Tsai_Yeh_Ravikumar_2022} have been proposed as extensions of the SV to interactions.
More general, the SII can be viewed as a particular instance of a broad class of interaction indices, known as cardinal interaction indices (CIIs) \citep{DBLP:journals/geb/FujimotoKM06}, which are defined as a weighted average over discrete derivatives:
\begin{equation*}
    I_K = \sum\limits_{S \subseteq \mathcal{N} \setminus K}\lambda_{k,\vert S \vert} \Delta_{K}(S)
\end{equation*}
with weights $\lambda_{k,|S|}$.
In particular, 
every interaction index satisfying the (generalized) linearity, symmetry and dummy axiom, e.g., SII, STI and FSI, can be represented as a CII \citep{Grabisch_Roubens_1999}.
Beyond Shapley-based interaction indices, CIIs also include other interaction indices, such as a generalized Banzhaf value \citep{Hammer_Holzman_1992}.
In Section~\ref{sec_method}, we propose a unified approximation that applies to any CII.
For details about other CIIs and their specific weights, we refer to Appendix~\ref{app:CII}.

The SII is the provably unique interaction index that fulfills the (generalized) linearity, symmetry and dummy axiom, as well as a novel recursive axiom that links higher order interactions to lower order interactions \citep{Grabisch_Roubens_1999}.
For interaction indices it is also possible to define a generalized efficiency condition, i.e. that $\sum_{K \subseteq \mathcal N, \vert K \vert \leq k_{\max}} I_K = \nu(\mathcal N)$ for a maximum interaction order $k_{\max}$.
In ML applications, this condition ensures that the sum of contributions equals the model behavior of $\mathcal N$, such as the prediction of an instance.
The SII scores can be aggregated to fulfill efficiency, which yield the n-Shapley values (n-SII) \citep{pmlr-v206-bordt23a}.
Furthermore, other variants, such as STI and FSI, extend the SV to interactions by directly requiring an efficiency axiom.
In contrast to the SV, however, a unique index is only obtained by imposing further conditions.
Similar to the SV, whose computation is NP-hard \citep{Deng.1994}, the weighted sum of discrete derivatives requires $2^n$ model evaluations, necessitating approximation techniques.

\subsection{Approximations of Shapley-based Interaction Scores}
Different approximation techniques have been proposed to overcome the computational complexity of Shapley-based interaction indices, which extend on existing techniques for the SV.

\paragraph{Permutation Sampling.}
For the SV, permutation sampling \citep{Castro.2009} was proposed, where the SV is represented as an average over randomly drawn permutations of the player set.
For each drawn permutation, 
the algorithm successively adds players to the subset, starting from the empty set using the given order.
By comparing the evaluations successively, the marginal contributions are used to update the estimates.
Extensions of permutation sampling have been proposed for the SII \citep{Tsai_Yeh_Ravikumar_2022} and STI \citep{Sundararajan_Dhamdhere_Agarwal_2020}.
For the SII, only interactions that appear in a consecutive order in the permutation can be updated, resulting in very few updates per permutation. 
For the STI, all interaction scores can be updated with a single permutation, however, the computational complexity increases, as the discrete derivatives have to be computed for every subset, resulting in an increase by a factor of $2^k$ per interaction.

\paragraph{Kernel-based Approximation.}
Besides the weighted average, the SV also admits a representation as a solution to a constrained weighted least square problem \citep{Charnes_Golany_Keane_Rousseau_1988}.
This optimization problem requires again $2^n$ model evaluations.
However, it was proposed to approximate the optimization problem through sampling 
and solve the resulting optimization problem explicitly, which is known as KernelSHAP \citep{Lundberg.2017}.
An extension of kernel-based approximation was proposed for FSI \citep{Tsai_Yeh_Ravikumar_2022}, but it remains open, whether this approach can be generalized to other indices, while its theoretical properties are unknown.

\paragraph{Unbiased KernelSHAP and SHAP-IQ.}
Unbiased KernelSHAP \citep{Covert.2021} constitutes a variant of KernelSHAP to approximate the SV, which yields stronger theoretical results, including an unbiased estimate.
While this approach is motivated through a kernel-based approximation, it was shown that it is possible to simplify the calculation to a sampling-based approach \citep{Fumagalli.2023}.
Using the sampling-based approach, SHAP-IQ \citep{Fumagalli.2023} extends Unbiased KernelSHAP to general interaction indices.

\subsection{Stratified Approximation for the SV}
Stratification partitions a population into distinct sub-populations, known as strata, where sampling is then separately executed for each stratum.
If the strata~are chosen as homogeneous groups with lower variability, stratified sampling yields a better approximation.
First proposed for the SV by \cite{Maleki.2013}, it was shown empirically that stratification by coalition size can 
improve the approximation 
\citep{Castro.2017}, while recent work extended it by more sophisticated techniques \citep{Burgress.2021}.
With Stratified SVARM, \cite{Kolpaczki.2023} proposed an approach that abstains from sampling marginal contributions.
Instead, it samples coalitions to leverage its novel representation of the SV, which splits the marginal contributions into two coalitions and stratifies them by size.
This allows one to assign each sampled coalition to one stratum per player, thus efficiently computing SV estimates for all players simultaneously.
Hence in contrast to permutation sampling, Stratified SVARM reaches a new level of efficiency as all estimates are updated using a single model evaluation.
In comparison to KernelSHAP, it is well understood theoretically and shows significant performance improvements compared to Unbiased KernelSHAP \citep{Kolpaczki.2023}.
In the following, we extend Stratified SVARM to Shapley-based interaction indices, and even general CIIs.

\section{SVARM-IQ: A STRATIFIED APPROACH}
\label{sec_method}

Since the practical infeasibility of computing the CII incentivizes its approximation as a remedy, we formally state our considered approximation problem under the fixed-budget setting in \cref{subsec:Problem}.
We continue by introducing our stratified representation of the CII in \cref{subsec:Representation}, which stands at the core of our new method \emph{SVARM-IQ} presented in \cref{subsec:SVARM-IQ}.

\subsection{Approximation Problem} \label{subsec:Problem}

Given a cooperative game $(\mathcal{N}, \nu)$, an order $k \geq 2$, a budget $B \in \mathbb{N}$, and the weights $(\lambda_{k,\ell})_{\ell \in \mathcal{L}_k}$, with $\mathcal{L}_k := \{0,\ldots,n-k\}$ specifying the desired CII, the goal is to approximate all the latent but unknown CII $I_K$ with $K \in \mathcal{N}_k := \{S \subseteq \mathcal{N} \mid |S| = k\}$ precisely.
The budget $B$ is the number of coalition evaluations or in other words accesses to $\nu$ that the approximation algorithm is allowed to perform.
It captures a time or resource constraint on the computation and is justified by the fact that the access to $\nu$ frequently imposes a bottleneck on the runtime due to costly inference, manipulation of data, or even retraining of models. 
We denote by $\hat{I}_K$ the algorithm's estimate of $I_K$.
Since we consider randomized algorithms, returning stochastic estimates, the approximation quality of an estimate $\hat{I}_K$ is judged by the following two commonly used measures that are to be minimized:
First, the mean squared error (MSE) of any set $K$: $\mathbb{E} \left[ ( \hat{I}_K - I_K )^2 \right]$, and second, a bound on the probability $\mathbb{P} (|\hat{I}_K - I_K| \geq \varepsilon) \leq \delta$ to exceed a threshold $\varepsilon > 0$, commonly known as a $(\epsilon, \delta)$-approximation.

\subsection{Stratified Representation} \label{subsec:Representation}

Our sampling-based approximation algorithm SVARM-IQ leverages a novel stratified representation of the CII.
For the remainder, we stick to the general notion of the CII of any fixed order $k \geq 2$.
The concrete interaction type to be approximated can be specified by the weights $\lambda_{k,\ell}$.
We stratify the CII $I_K$ by coalition size and split the discrete derivatives $\Delta_K(S)$ into multiple strata to obtain: 
\begin{equation*}
    I_K = \sum\limits_{\ell=0}^{n-k} \binom{n-k}{\ell} \lambda_{k,\ell} \sum\limits_{W \subseteq K} (-1)^{k-|W|} \cdot I_{K,\ell}^W.
\end{equation*}
with strata terms for all $W \subseteq K$ and $\ell \in \mathcal{L}_k$:
\begin{equation}\label{eq_def_stratum}
    I_{K,\ell}^W := \frac{1}{\binom{n-k}{\ell}} \sum\limits_{\substack{S \subseteq \mathcal{N} \setminus K \\ |S| = \ell}} \nu(S \cup W).
\end{equation}
This representation is a generalization of the SV representation utilized by Stratified SVARM \citep{Kolpaczki.2023} as it extends from the SV to the CII.
Since each stratum contains $\binom{n-k}{\ell}$ many coalitions, $I_{K,\ell}^W$ is a uniform average of all eligible coalition worths and hence we obtain its estimate $\hat{I}_{K,\ell}^W$ by taking the sample-mean of evaluated coalitions belonging to that particular stratum.
Further, we can express any CII by means of the strata $I_{K,\ell}^W$ trough manipulating their weighting according to the weights $\lambda_{k,\ell}$.
Subsequently, the aggregation of the strata estimates, mimicking our representation, yields the desired CII estimate:
\begin{equation*}
     \hat{I}_K = \sum\limits_{\ell=0}^{n-k} \binom{n-k}{\ell} \lambda_{k,\ell} \sum\limits_{W \subseteq K} (-1)^{k-|W|} \cdot \hat{I}_{K,\ell}^W.
\end{equation*}
Further, we demonstrate the popular special case of SII between pairs, i.e., $k=2$, in Appendix~\ref{app:Pairs}.

\subsection{SVARM-IQ} \label{subsec:SVARM-IQ}

Instead of naively sampling coalitions separately from each of the $2^k \binom{n}{k} (n-k+1)$ many strata, we propose with SVARM-IQ a more sophisticated mechanism, similar to \cite{Kolpaczki.2023}, which leverages the stratified representation of the CII.

\paragraph{Update Mechanism.}
SVARM-IQ (given in \cref{alg:SVARM-IQ} and \cref{fig_SVARMIQ}) updates for a single sampled coalition $A \subseteq \mathcal{N}$ one strata estimate of each of the $\binom{n}{k}$ many considered subsets $K$.
This is made feasible by the observation that any coalition $A$ belongs into exactly one stratum associated with $I_K$.
This is in the spirit of the \emph{maximum sample reuse} principle, employed previously for the Banzhaf value \citep{Wang.2023} and the SV \citep{Kolpaczki.2023} with the underlying motivation that each seen observation should be utilized to update all interaction estimates.
To be more precise, for each $K \in \mathcal{N}_k$ we update
\begin{equation}\label{eq_W_ell}
    \hat{I}_{K,\ell}^W \text{ with } W = A \cap K \text{ and } \ell = |A| - |W|.
\end{equation}
Notably, our sampling ensures that for \emph{every interaction} $A \setminus K \sim \text{unif}(\{S \subseteq \mathcal N \setminus K \mid \vert S \vert = \ell\})$, as the probability of $A \setminus K$ conditioned on $W = A \cap K$ and $\ell = \vert A \vert - \vert A \cap K \vert$ is uniform.
This is required as $\hat I_{K,\ell}^W$ is a uniform average, cf. Eq.~(\ref{eq_def_stratum}), and allows to update estimates for every interaction by sampling a single subset $A$.
Considering the limited budget $B$, this update rule elicits information from $\nu$ in a more ``budget-efficient'' manner, since it contributes to $\binom{n}{k}$ many estimates with only a single evaluation.
To guide the sampling, we first draw in each time step $b$ a coalition size $a_b$ from a probability distribution $P_k$ over the eligible sizes, 
and draw then $A_b$ uniformly at random among all coalitions of size $a_b$.
We store the evaluated worth $\nu(A_b)$ in order to reuse it for all estimate updates, one for each $K$.
This is done by calling \textsc{\texttt{UpdateMean}} (see Appendix~\ref{app:Pseudocode}), which sets the associated estimate $\hat{I}_{K,\ell}^W$ to the new average, taking the sampled worth $\nu(A_b)$ and the number of so far observed samples $c_{K,\ell}^W$ of that particular estimate into account.
We set $P_k$ to be the uniform distribution over all sizes, i.e., $P_k=\text{unif}(0,n)$.
A specifically tailored distribution for $k=2$ allows us to express sharper theoretical results in \cref{sec:TheoreticalResults}.

\begin{figure}[t]
    \centering
    \includegraphics[width=0.95\columnwidth]{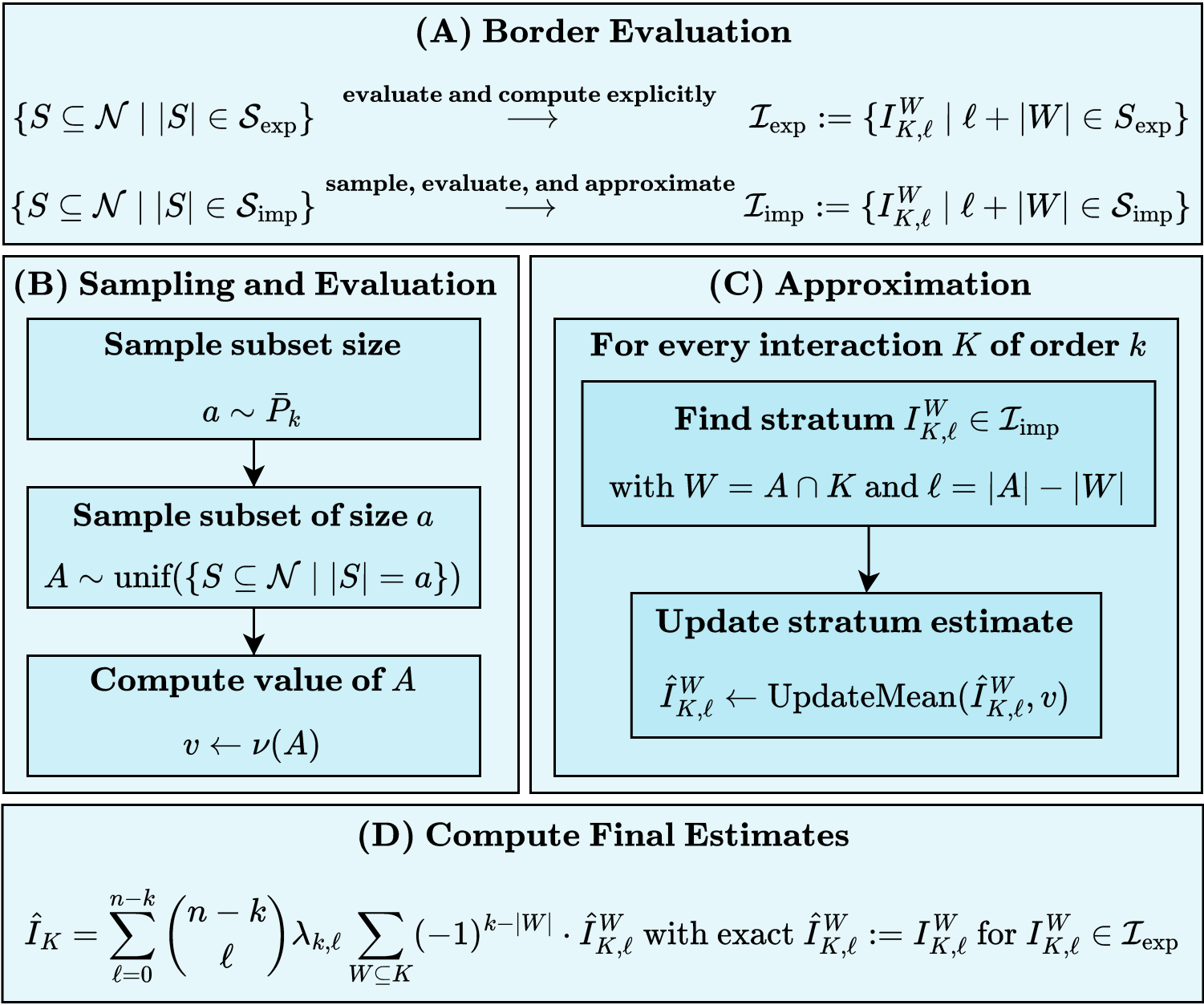}
    \caption{Schematic overview of SVARM-IQ.}
    \label{fig_SVARMIQ}
\end{figure}

\paragraph{Border Sizes.}
Further, we enhance our approach by transferring a technique, introduced by \cite{Fumagalli.2023}.
We observe that for very low and very high $s$ only a few coalitions of size $s$ exist, $\binom{n}{s}$ many to be precise.
Thus, evaluating all these coalitions and calculating the associated strata $I_{K,\ell}^W$ explicitly upfront saves budget, as it avoids duplicates, i.e., coalitions sampled multiple times.
Given the budget $B$ and the probability distribution over sizes $P_k$, we determine a set of subset sizes $\mathcal{S}_{\text{exp}} = \{0,\ldots,s_{\text{exp}},n-s_{\text{exp}},\ldots,n\}$, for which the expected number of samples exceeds the number of coalitions of each subset size.
Consequently, we evaluate \emph{all coalitions} of sizes in $\mathcal{S}_{\text{exp}}$, i.e., $S \subseteq \mathcal N$ with $\vert S \vert \in \mathcal S_{\text{exp}}$.
From the remaining sizes $\mathcal S_{\text{imp}} := \{s_{\text{exp}}+1,\ldots,n-s_{\text{exp}}-1\}$, we sample coalitions.
This split allows to compute all strata
\begin{equation*}
    \mathcal I_{\text{exp}} := \{ I_{K,\ell}^W \mid \ell + \vert W \vert \in S_\text{exp}\}
\end{equation*}
\emph{explicitly}, which follows from Eq.~(\ref{eq_W_ell}) and $\ell + \vert W \vert = \vert A \vert$.
The remaining strata 
\begin{align*}
     \mathcal I_{\text{imp}} &:= \{ I_{K,\ell}^W \mid \ell + \vert W \vert \in \mathcal S_{\text{imp}}\}
\end{align*}
are \emph{approximated} with $\hat I_{K,\ell}^W$ by sampling coalitions. 
The procedure to determine $\mathcal S_{\text{exp}}$ and $\mathcal S_{\text{imp}}$, named \textsc{\texttt{ComputeBorders}} (see Appendix~\ref{app:Pseudocode}), is applied before the sampling loop in \cref{alg:SVARM-IQ}. 
Hence, SVARM-IQ enters its sampling loop with a leftover budget of $\bar{B} := B - \sum\nolimits_{s \in \mathcal{S}_{\text{exp}}} \binom{n}{s}$, and repeatedly applies the update mechanism.
The distribution $P_k$ is altered to $\bar{P}_k$ by setting $\bar{P}_k(s) = 0$ for all $s \in \mathcal{S}_{\text{exp}}$ and upscaling all entries $s \in \mathcal{S}_{\text{imp}}$ such that they sum up to 1.
Note that this technique yields exact CII values for $B = 2^n$.

\begin{algorithm}[t]
\caption{SVARM-IQ}
\label{alg:SVARM-IQ}
\begin{algorithmic}[1]
    \STATE {\bfseries Input:} $(\mathcal{N}, \nu)$, $B \in \mathbb{N}$, $k \in \{1,\ldots,n\}$, $(\lambda_{k,\ell})_{\ell \in \mathcal{L}_k}$
    \STATE $\hat{I}_{K,\ell}^W, c_{K,\ell}^W \leftarrow 0$ $\forall K \in \mathcal{N}_k, \ell \in \mathcal{L}_k, W \subseteq K$      
    \STATE $\mathcal{S}_{\text{exp}}, \mathcal{S}_{\text{imp}} \leftarrow$ \textsc{\texttt{ComputeBorders}}
    \STATE $\bar{B} \leftarrow B - \sum\nolimits_{s \in \mathcal{S}_{\text{exp}}} \binom{n}{s}$
    \FOR{$b = 1, \ldots, \bar{B}$}
        \STATE Draw size $a_b \in \mathcal{S}_{\text{imp}} \sim \bar{P}_k$
        \STATE Draw $A_b$ from $\{S \subseteq \mathcal{N} \mid |S| = a_b\}$ u.a.r.
        \STATE $v_b \leftarrow \nu(A_b)$ \COMMENT{store coalition worth}
        \FOR{$K \in \mathcal{N}_k$}
            \STATE $W \leftarrow A_b \cap K$ \COMMENT{get stratum set}
            \STATE $\ell \leftarrow a_b - |W|$ \COMMENT{get stratum size}
            \STATE $\hat{I}_{K,\ell}^W \leftarrow$ \textsc{\texttt{UpdateMean}}$(\hat{I}_{K,\ell}^W, c_{K,\ell}^W, v_b)$
            \STATE $c_{K,\ell}^{W} \leftarrow c_{K,\ell}^{W} + 1$ \COMMENT{increment counter}
        \ENDFOR
    \ENDFOR
    \STATE $\hat{I}_k \leftarrow \sum\limits_{\ell=0}^{n-k} \binom{n-k}{\ell} \lambda_{k,\ell} \sum\limits_{W \subseteq K} (-1)^{k-|W|} \hat{I}_{K,\ell}^W$ $\forall K \in \mathcal{N}_k$
    \STATE {\bfseries Output:} $\hat{I}_K$ for all $K \in \mathcal{N}_k$
\end{algorithmic}
\end{algorithm}


\paragraph{Approximating Multiple Orders and Indices.}

SVARM-IQ is not restricted to approximate only one specific order $k$ at the time.
Quite to the contrary, it can be extended to maintain strata estimates $\hat{I}_{K,\ell}^W$ for multiple orders, which are then simultaneously updated within the sampling loop without imposing further budget costs.
The aggregation to interaction estimates $\hat{I}_K$ is then carried out for each considered subset $K$ separately.
Note that this also entails the SV, i.e., $k=1$, thus allowing one to approximate attribution and interaction simultaneously.
Since the stratification allows to combine the strata to any CII, SVARM-IQ can approximate multiple CII's at the same time, notably without even the need to specify them during sampling.
This can be realized by specifying multiple weighting sequences $(\lambda_{k,\ell})_{\ell \in \mathcal{L}_k}$, one for each CII of interest, and performing the final estimate computation $\hat{I}_K$ for each type.
Note that this comes without incurring any additional budget cost.

\section{THEORETICAL RESULTS} \label{sec:TheoreticalResults}

In the following, we present the results of our theoretical analysis for SVARM-IQ.
All proofs are deferred to Appendix~\ref{app:Proofs}.
In order to make the analysis feasible, a natural assumption is to observe at least one sample for each approximated stratum $I_{K,\ell}^W \in \mathcal{I}_{\text{imp}}$.
We realize this requirement algorithmically only for the remainder of this chapter by executing a \textsc{\texttt{WarmUp}} procedure (see Appendix~3) between \textsc{\texttt{ComputeBorders}} and SVARM-IQ's sampling loop.
For each $I_{K,\ell}^W \in \mathcal{I}_{\text{imp}}$ it samples a coalition $A \subseteq \mathcal{N} \setminus K$ of size $\ell$ and sets $\hat{I}_{K,\ell}^W$ to $\nu(A \cup W)$.
Hence, SVARM-IQ enters its sampling loop with a leftover budget of $\tilde{B} := B - \sum\nolimits_{s \in \mathcal{S}_{\text{exp}}} \binom{n}{s} - \vert \mathcal{I}_{\text{imp}} \vert$.
We automatically set $s_{\text{exp}} \geq 1$, which consumes only $2n+2$ evaluations. 
Hence for $n=3$, all strata are already explicitly calculated.
Since \textsc{\texttt{ComputeBorders}} evaluates then at least all coalitions of size $s \in \{0,1,n-1,n\}$, the initial distribution $P_k$ over sizes has support $\{2,\ldots,n-2\}$.
For $k \geq 3$, this allows us to specify $P_k$ to be the uniform distribution:
\begin{equation*}
    P_k(s) := \frac{1}{n-3} \text{ for all } s \in \{2,\ldots,n-2\}.
\end{equation*}
Further for the remainder of the analysis, we use a specifically tailored distribution in the case of $k=2$:
\begin{equation*}
    P_2(s) :=
    \begin{cases}
        \frac{\beta_n}{s(s-1)} & \text{if } s \leq \frac{n-1}{2} \\
        \frac{\beta_n}{(n-s)(n-s-1)} & \text{if } s \geq \frac{n}{2}
    \end{cases}
\end{equation*}
with $\beta_n = \frac{n^2-2n}{2(n^2-4n+2)}$ for even $n \geq 4$ and $\beta_n = \frac{n-1}{2(n-3)}$ for odd $n \geq 5$.
This allows us to express sharper bounds in comparison to the uniform distribution.

\paragraph{Notation and assumptions.}
We introduce some notation, coming in helpful in expressing our results legibly.
For any $w \in \{0,\ldots,k\}$ we denote by $\mathcal{L}_k^w := \{\ell \in \mathcal{L}_k \mid \ell + w \in \mathcal{S}_{\text{imp}}\}$.
For any $K \in \mathcal{N}_k$ and $\ell \in \mathcal{L}_k$ let $A_{K,\ell}$ be a random coalition with distribution $\mathbb{P}(A_{K,\ell} = S) = \binom{n-k}{\ell}^{-1}$ for all $S \subseteq \mathcal{N} \setminus K$ with $|S|=\ell$.
For any $W \subseteq K$ we denote the \emph{stratum variance} by $\sigma_{K,\ell,W}^2 := \mathbb{V}[\nu(A_{K,\ell} \cup W)]$
and the \emph{stratum range} by $r_{K,\ell,W} := \max\nolimits_{\substack{S \subseteq \mathcal{N} \setminus K \\ |S| = \ell}} \nu(S \cup W) - \min\nolimits_{\substack{S \subseteq \mathcal{N} \setminus K \\ |S| = \ell}} \nu(S \cup W)$.
For a comprehensive overview of the used notation, we refer to Appendix~\ref{app:symbols}.
As our only assumptions, we demand $n \geq 4$ and the budget to be large enough 
to execute \textsc{\texttt{ComputeBorders}}, \textsc{\texttt{WarmUp}} , and the sampling loop for one iteration, i.e., \ $\tilde{B} > 0$.

\paragraph{Unbiasedness, Variance, and MSE.}
We begin by showing that SVARM-IQ's estimates are unbiased, which is not only desirable but will also turn out useful shortly after in our analysis.
\begin{theorem} \label{the:Unbiased}
    SVARM-IQ's CII estimates are unbiased for all $K \in \mathcal{N}_k$, i.e., $\mathbb{E} [\hat{I}_K] = I_K$.
\end{theorem}
The unbiasedness enables us to reduce the MSE of any $\hat{I}_K$ to its variance.
In fact, the bias-variance decomposition states that $\mathbb{E}[(\hat{I}_K - I_K)^2] = (\mathbb{E}[\hat{I}_K] - I_K)^2 + \mathbb{V}[\hat{I}_K]$.
Hence, a variance analysis of the obtained estimates suffices to bound the MSE.
The variance of $\hat{I}_K$ is tightly linked to the number of samples SVARM-IQ collects for each stratum estimate $\hat{I}_{K,\ell}^W$.
At this point, we distinguish in our analysis between $k=2$ and $k \geq 3$ to obtain sharper bounds for the former case facilitated by our carefully designed probability distribution $P_2$ over coalition sizes.
To keep the presented results concise, we introduce $\gamma_k := 2(n-1)^2$ for $k=2$ and $\gamma_k := n^{k-1} (n-k+1)^2$ for all $k \geq 3$.
This stems from the aforementioned difference in precision on the lower bound of collected samples.
\begin{theorem} \label{the:Variance}
    For any $K \in \mathcal{N}_k$ the variance of the CII estimate $\hat{I}_K$ returned by SVARM-IQ is bounded by
    \begin{equation*}
        \mathbb{V} \left[ \hat{I}_K \right] \leq \frac{\gamma_k}{\tilde{B}} \sum\limits_{W \subseteq K} \sum\limits_{\ell \in \mathcal{L}_k^{|W|}} \binom{n-k}{\ell}^2 \lambda_{k,\ell}^2 \sigma_{K,\ell,W}^2.
    \end{equation*}
\end{theorem}
Note that our efforts in optimizing the analysis for $k=2$ reduced the bound by a factor of $\frac{n}{2}$ in comparison to substituting $k$ with 2 in our bound for the general case.
This is caused by the severe increase in complexity when trying to give a lower bound for the number of samples each stratum receives.
Although our approach allows one to obtain a sharper bound for special cases as $k=3$ or $k=4$ with a similarly dedicated analysis, we abstain from doing so as we prioritize a concise presentation of our results.
\begin{corollary} \label{cor:MSE}
    For any $K \in \mathcal{N}_k$, the MSE of $\hat{I}_K$ returned by SVARM-IQ is bounded by $\mathbb{E}[(\hat{I}_K - I_K)^2] \leq$
    \begin{equation*}
        \frac{\gamma_k}{\tilde{B}} \sum\limits_{W \subseteq K} \sum\limits_{\ell \in \mathcal{L}_k^{|W|}} \binom{n-k}{\ell}^2 \lambda_{k,\ell}^2 \sigma_{K,\ell,W}^2.
    \end{equation*}
\end{corollary}
We 
state this result more explicitly for the 
frequently considered interaction type: the SII for pairs of players $i$ and $j$.
In this case our bound boils down to
\begin{equation*}
    \mathbb{E} \left[ \left( \hat{I}_{i,j}^{\text{SII}} - I_{i,j}^{\text{SII}} \right)^2 \right] \leq \frac{2}{\tilde{B}} \sum\limits_{W \subseteq \{i,j\}} \sum\limits_{\ell \in \mathcal{L}_2^{|W|}} \sigma_{i,j,\ell,W}^2.
\end{equation*}
The simplicity achieved by this result supports a straightforward and natural interpretation.
The MSE bound of each SII estimate is inversely proportional to the available budget for the sampling loop and each stratum variance contributes equally to its growth.

We \emph{intentionally abstain} from expressing our bounds in asymptotic notation w.r.t.\ $B$ and $n$ only, as it would not do justice to the motivation behind employing stratification.
The performance of SVARM-IQ is based on lower strata variances (and also strata ranges) compared to the whole population of all coalition values within the powerset of $\mathcal{N}$.
This improvement can not be reflected adequately by the asymptotics in which the variances vanish to constants.

\paragraph{$(\epsilon, \delta)$-Approximation.}
Combining \cref{the:Variance} with Chebyshev's inequality immediately yields a bound on the probability that the absolute error of a fixed $\hat{I}_K$ exceeds some $\varepsilon > 0$ given the budget at~hand.

\begin{figure*}[t]
    \centering
    \begin{minipage}[c]{0.49\textwidth}
    \centering
    \includegraphics[width=0.49\textwidth]{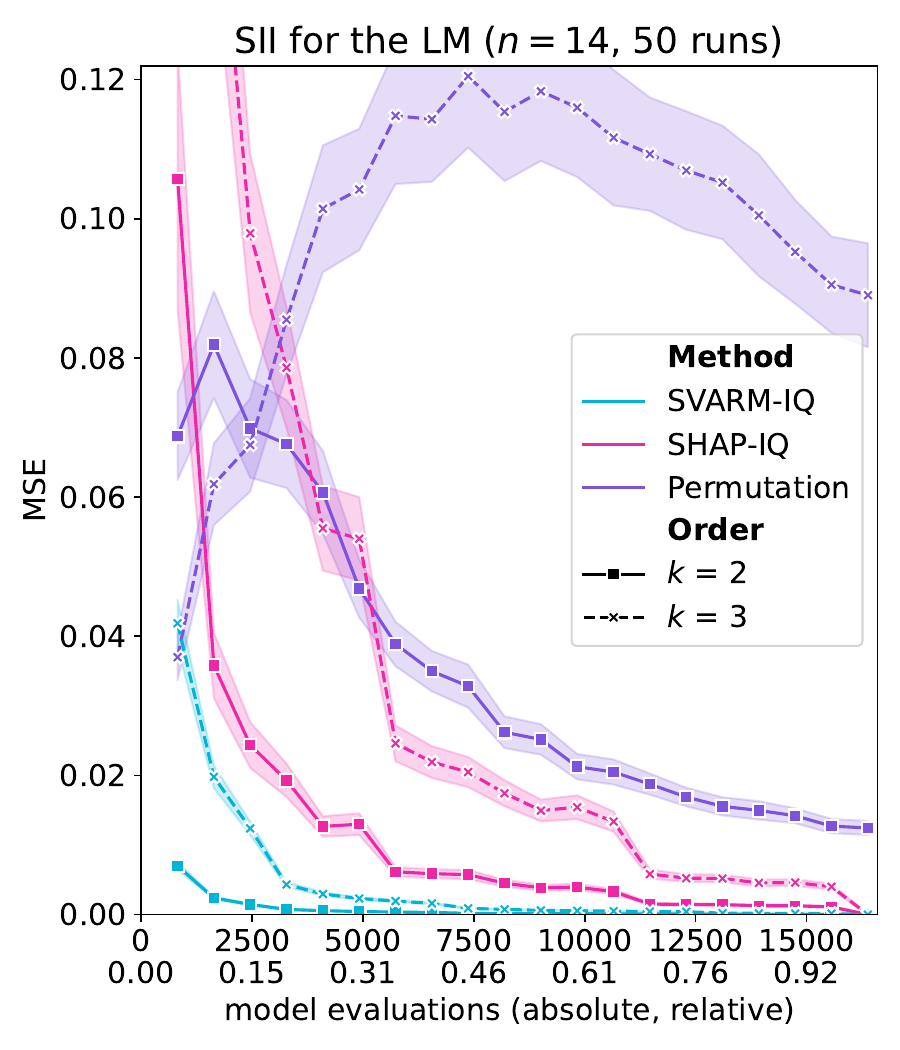}
    \includegraphics[width=0.49\textwidth]{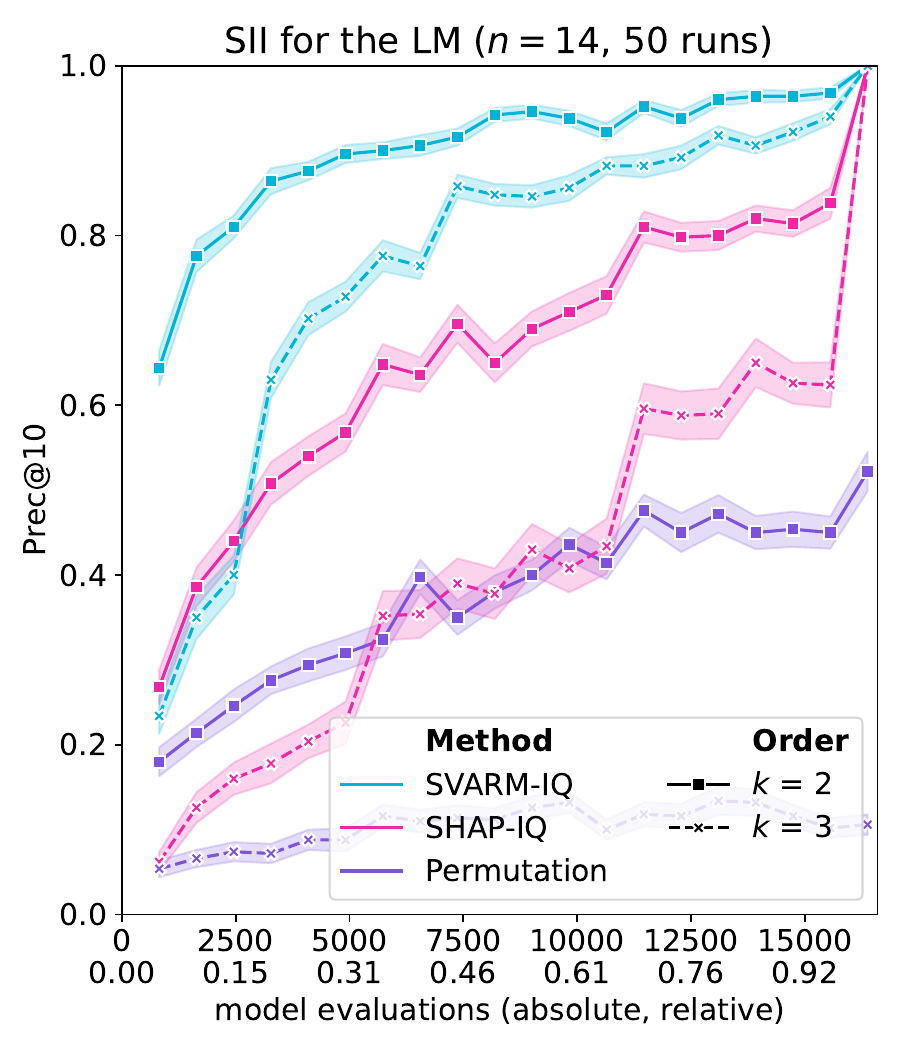}
    (a) MSE and Prec@10 for the LM
    \end{minipage}
    \hfill
    \begin{minipage}[c]{0.49\textwidth}
    \centering
    \includegraphics[width=0.49\textwidth]{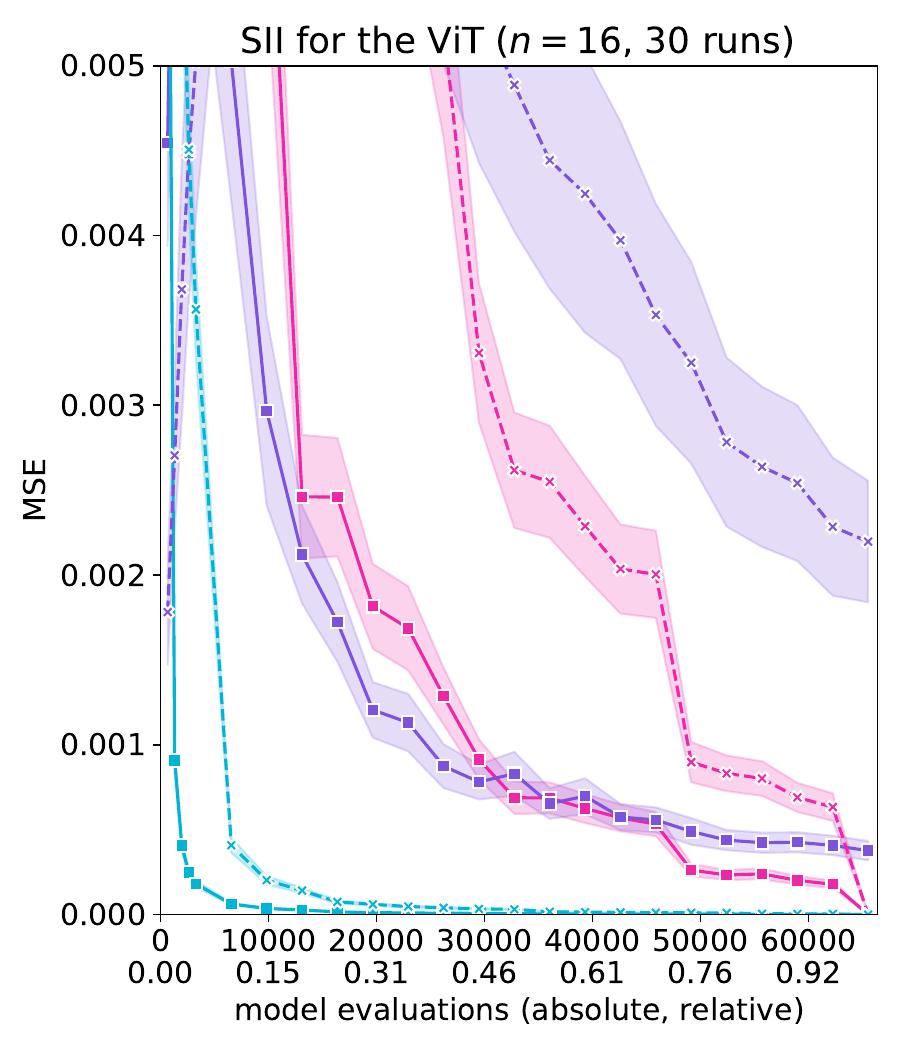}
    \includegraphics[width=0.49\textwidth]{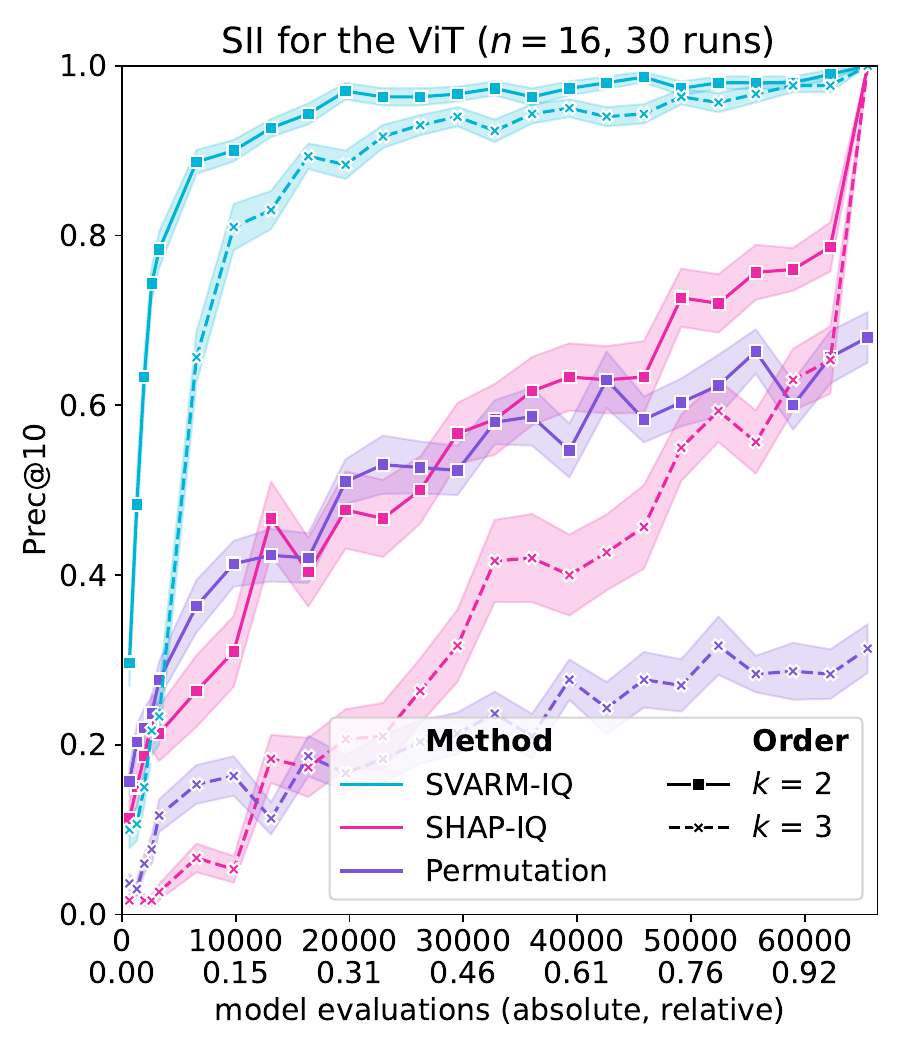}
    (b) MSE and Prec@10 for the ViT
     \end{minipage}
    \caption{Approximation quality of SVARM-IQ (\textcolor{svarmiq}{blue}) compared to SHAP-IQ (\textcolor{shapiq}{pink}) and permutation sampling (\textcolor{baseline}{purple}) baselines 
    for estimating order $k=2,3$ SII on the LM (a; $n=14$) and the ViT (b; $n=16$). Shaded bands represent the standard error over 50, respectively 30 runs.}
    \label{fig_exp_approx}
\end{figure*}

\begin{corollary} \label{cor:Chebyshev}
    For any $K \in \mathcal{N}_k$, the absolute error of $\hat{I}_K$ returned by SVARM-IQ exceeds some fixed $\varepsilon$ with probability of at most $\mathbb{P} (| \hat{I}_K - I_K | \geq \varepsilon) \leq$
    \begin{equation*}
        \frac{\gamma_k}{\varepsilon^2 \tilde{B}} \sum\limits_{W \subseteq K} \sum\limits_{\ell \in \mathcal{L}_k^{|W|}} \binom{n-k}{\ell}^2 \lambda_{k,\ell}^2 \sigma_{K,\ell,W}^2.
    \end{equation*}
\end{corollary}
One can easily rearrange the terms to find the minimum budget required to obtain $\mathbb{P}(|\hat{I}_K - I_K| \leq \varepsilon) \geq 1 - \delta$ for a given $\delta > 0$.
Note that this bound still depends on the unknown strata variances.
Further, we provide another bound in Theorem 4.5 (see Appendix~\ref{app:epsilon_delta}), resulting from a slightly more laborious usage of Hoeffding's inequality which takes the strata ranges into account.
To the best of our knowledge there exists no theoretical analysis for permutation sampling of CIIs.
SHAP-IQ is like wise unbiased, but its theoretical analysis (Theorem 4.3) \cite{Fumagalli.2023} does not provide such detail for fixed $n$ and $k$.

\section{EXPERIMENTS} \label{sec_experiments}

We empirically evaluate SVARM-IQ's approximation quality in different XAI application scenarios and compare it with current state-of-the-art baselines.

\paragraph{Baselines.} In the case of estimating SII and STI scores, we compare SVARM-IQ to SHAP-IQ \citep{Fumagalli.2023} and permutation sampling \citep{Sundararajan_Dhamdhere_Agarwal_2020,Tsai_Yeh_Ravikumar_2022}.
For FSI, we compare against the kernel-based regression approach \citep{Tsai_Yeh_Ravikumar_2022} instead of permutation sampling.

\begin{table}[h]
    \caption{Overview of the XAI tasks and models used}\label{tab_setup}
    \vspace{0.5em}
    \adjustbox{max width=\columnwidth}{%
    \begin{tabular}{@{}ccccc@{}}
    \toprule
    \textbf{Task} & \textbf{Model ID} & \textbf{Removal Strategy} & $n$ & $\mathcal{Y}$ \\ \midrule
    LM & DistilBert & Token Removal & 14 & $[-1, 1]$  \\[0.8em]
    ViT & ViT-32-384 &Token Removal & 9,16 & $[0, 1]$\\[0.8em]
    CNN & ResNet18 & \begin{tabular}{c}Superpixel\\Marginalization\end{tabular} & 14 & $[0, 1]$ \\
    \bottomrule
    \end{tabular}}
\end{table}

\paragraph{Explanation Tasks.} Similar to \cite{Fumagalli.2023} and \cite{Tsai_Yeh_Ravikumar_2022}, we evaluate the approximation algorithms based on different real-world ML models and classical XAI scenarios (cf.\ \cref{tab_setup}).
First, we compute interaction scores to explain a sentiment analysis \emph{language model} (LM), which is a fine-tuned version of \texttt{DistilBert} \citep{Sanh.2019} on the IMDB \citep{Maas.2011} dataset.
Second, we investigate two types of image classification models, which were pre-trained on ImageNet \citep{ImageNet}.
We explain a \emph{vision transformer} (ViT), \citep{DBLP:conf/iclr/DosovitskiyB0WZ21}, and a \texttt{ResNet18} \emph{convolutional neural network} (CNN) \citep{DBLP:conf/cvpr/HeZRS16}.
The ViT operates on patches of 32 times 32 pixels and is abbreviated with \texttt{ViT-32-384}.
The \texttt{torch} versions of the LM, ViT, and the CNN are retrieved from \cite{Wolf_Transformers_State-of-the-Art_Natural_2020} and \cite{torch.2017}. 
For further descriptions on the models and feature removal strategies aligned with \cite{Covert_Lundberg_Lee_2021}, we refer to Appendix~\ref{app:models_datasets}.

\paragraph{Measuring Performance.}
To assess the performance of the different approximation algorithms, we measure the mean squared error averaged over all $K \in \mathcal{N}_k$ (MSE; lower is better) and the precision at ten (Prec@10; higher is better) of the estimated interaction scores compared to pre-computed ground-truth values (GTV).
Prec@10 measures the ratio of correctly identifying the ten highest (absolute) interaction values.
The GTV for each run are computed exhaustively with $2^n$ queries to the black box models.
All results are averaged over multiple independent runs.

\paragraph{Approximation Quality for SII.}
We compare SVARM-IQ against permutation sampling and SHAP-IQ at the LM and ViT explanation tasks for approximating all SII values of order $k=2$ and $k=3$ in \cref{fig_exp_approx}.
Across both considered measures, MSE and Prec@10, SVARM-IQ demonstrates superior approximation quality.
Noteworthy is SVARM-IQ's steep increase in approximation quality in the earlier budget range allowing applications with limited computational resources.
Based on our theoretical findings, we assume the stratification by size in combination with the splitting of discrete derivatives to be the cause for the observed behavior.
Most plausibly coalitions of the same size and sharing a predetermined set, as encompassed by each stratum $I_{K,\ell}^W$, vary less in their worth than the whole population of coalitions.
Consequently, the associated variance $\sigma_{K,\ell,W}^2$ is considerably lower, leading to faster convergence of the estimate $\hat{I}_{K,\ell}^W$.

\begin{figure}[t]
    \centering
    \begin{minipage}[c]{0.49\textwidth}
    \includegraphics[width=0.49\textwidth]{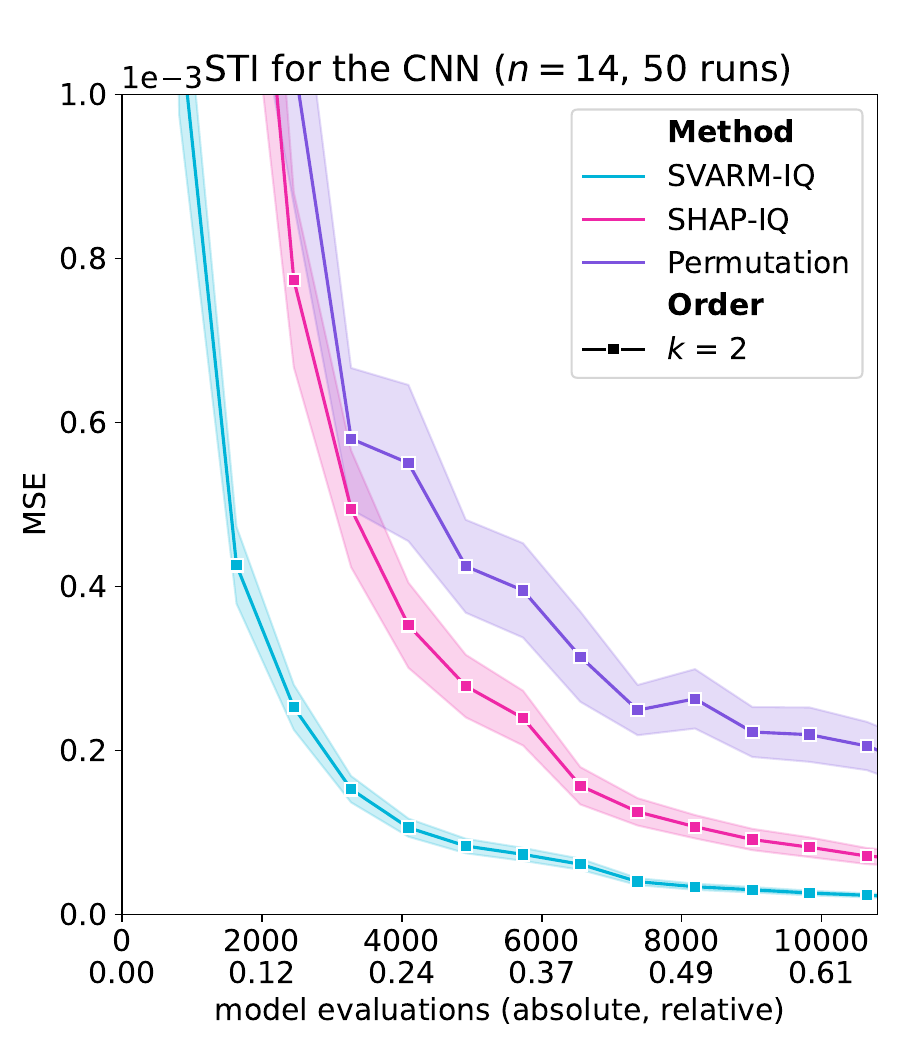}
    \includegraphics[width=0.49\textwidth]{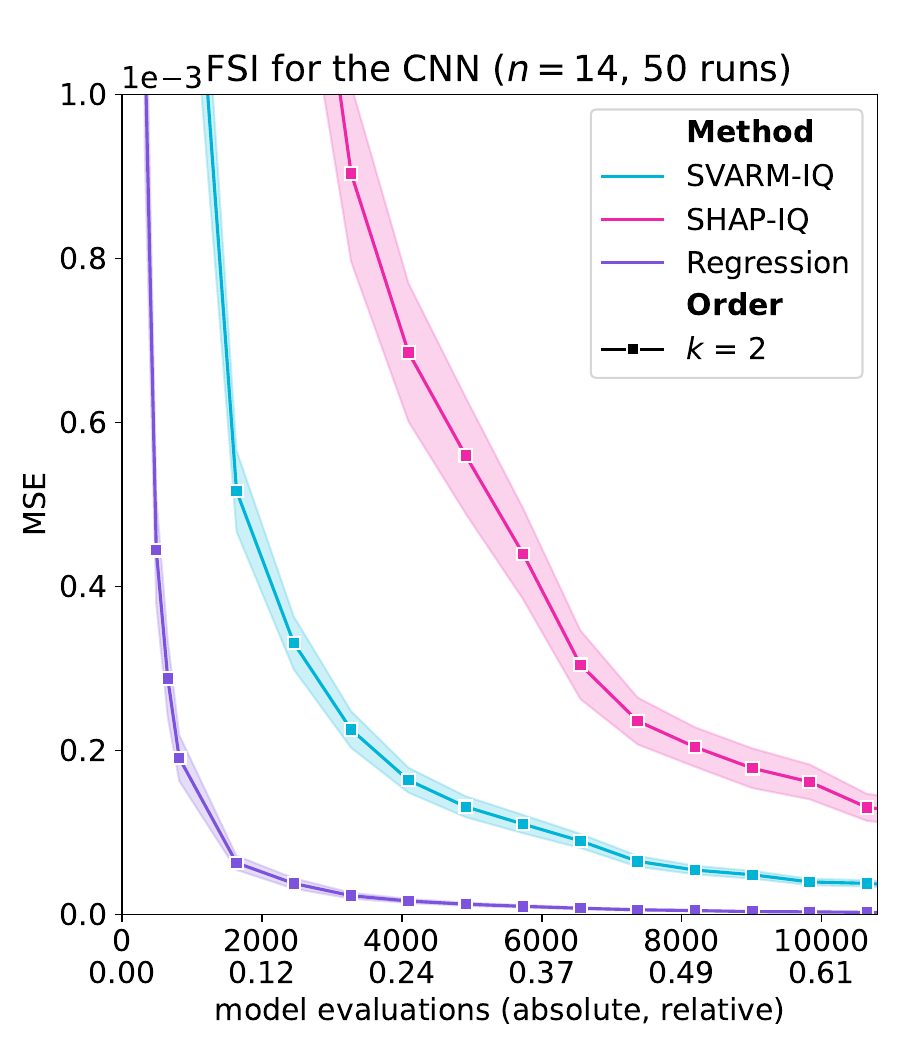}
    \end{minipage}
    \caption{
    Comparison of SVARM-IQ and 
    baselines for STI (left) and FSI (right) on the CNN. Shaded bands represent the standard error over 50 runs.}
    \label{fig_exp_approx_cii}
\end{figure}

\begin{figure*}[t]
    \includegraphics[width=0.29\textwidth]{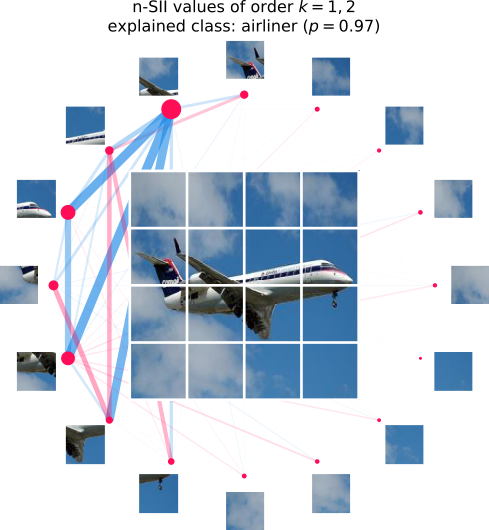}
    \hfill
     \includegraphics[width=0.29\textwidth]{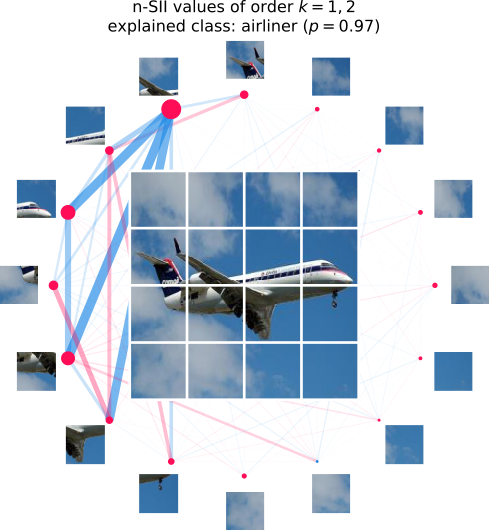}
     \hfill
     \includegraphics[width=0.32\textwidth]{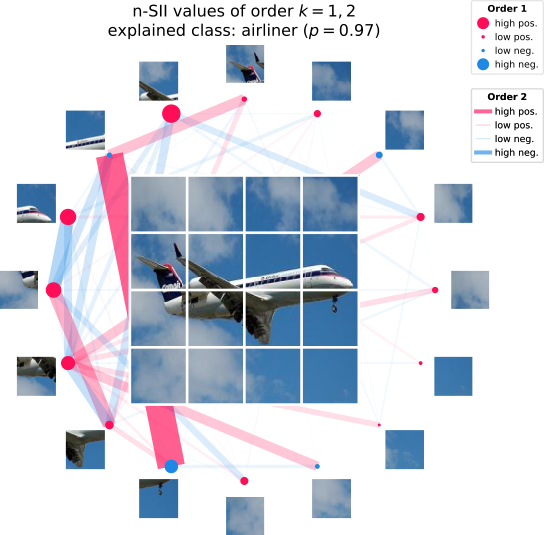}
    \caption{Comparison of ground-truth n-SII values of order $k=1$ and $k=2$ for the predicted class probability of a ViT for an ImageNet picture sliced into a grid of $n=16$ patches (left) against n-SII values estimated by \mbox{SVARM-IQ} (center) and permutation sampling (right). The exact computation requires 65,536 model~evaluations while the budget of both approximators is limited by 5000, making up only 7.6\% of the space to sample.}
    \label{fig_vit_comparison}
\end{figure*}

\paragraph{Example Use-Case of n-SII Values.}
Precise estimates allow to construct high-quality n-SII scores as proposed by \cite{pmlr-v206-bordt23a}.
\cref{fig_intro_example} illustrates how n-SII scores can be used to \mbox{explain} the ViT with 16 patches for an image of two correctly classified Labradors.
All individual patches receive positive attribution scores ($k=1$) of varying degree, leading practitioners to assume that patches with \mbox{similar} attribution are of equal importance.
However, enhancing the explanation with second order interactions ($k=2$), reveals how the interplay between patches containing complementing facial parts, like the eyes and the mouth, strongly influences the model's prediction towards the correct class label.
On the contrary, tiles depicting the same parts, e.g.\ those containing eyes, show negative interaction, allowing to conclude that the addition of one in the presence of the other is on average far less impactful than their individual contribution.
Solely observing the monotony of the~individual scores would have arguably led to overlook this insight.
We describe this further in Appendix~\ref{app:empirical_results:vit_examples}.

\paragraph{Estimating FSI and STI.}
Further, we compute different CIIs of a fixed order with SVARM-IQ and consistently achieve high approximation quality.
We summarize the results on the CNN in \cref{fig_exp_approx_cii}.
For STI, SVARM-IQ, again, outperforms both sampling-based baselines.
The kernel-based regression estimator, which is only applicable to the FSI index, yields lower approximation errors than SVARM-IQ.
Similar to SV estimation through KernelSHAP \citep{Lundberg.2017}, this highlights the expressive power of the least-squares representation available for FSI.

\paragraph{Instance-wise comparison.}
Lastly, we compare in \cref{fig_vit_comparison} SVARM-IQ's n-SII estimates of order $k=1$ and $k=2$ against those of permutation sampling and the ground truth for single a single instance.
The ground truth interaction is computed upfront for the predicted class probability of the ViT for a specified image sliced into a grid of 16 patches, and both approximation algorithms are executed for a single run with a budget of 5000 model evaluations, thus consuming only 7.6\% of the budget necessary to compute GTV exactly.
The estimates obtained by \mbox{SVARM-IQ} show barely any visible difference to the human eye.
In fact, SVARM-IQ's approximation replicates the ground truth with only a fraction of the number of model evaluations that are necessary for its exact computation.
Hence, it significantly lowers the computational burden for precise explanations.
On the contrary, permutation sampling yields estimated
importance and interaction scores which are
afflicted with evident imprecision.
This lack in~approximation quality has the potential to cause misguiding explanations.
More comparisons are shown in Appendix~\ref{app:empirical_results:vit_examples}.

\section{CONCLUSION}

We proposed SVARM-IQ, a new sampling-based approximation algorithm for interaction indices based on a stratified representation to maximize budget efficiency.
SVARM-IQ is capable of approximating all types and orders of cardinal interactions simultaneously, including the popular SII.
Consequently, as the special case of SVs is also entailed, this facilitates the approximation of feature importance and interaction simultaneously, thus offering an enriched explanation.
Besides proving theoretical results, we empirically demonstrated SVARM-IQ's advantage against current state-of-the-art baselines.
Its model-agnostic nature and domain-independence allow practitioners to obtain high-quality interaction scores for various entity types such as features or data points. 

\paragraph{Limitations and Future Work.}
Due to SVARM-IQ's stratification, the number of maintained strata estimates grows exponentially with the interaction order $k$.
This space complexity poses a challenge for large interaction orders.
As a pragmatic remedy, future work may consider the approximation of interaction scores for a smaller number of sets or a coarser stratification by size.
Lastly, it still remains unclear whether the performance of the kernel-based regression estimator available for FSI and the SV can be transferred to other types of CIIs like the SII or STI indices.

\subsection*{Acknowledgements}
This reserach was supported supported by the research training group Dataninja (Trustworthy AI for Seamless Problem Solving: Next Generation Intelligence Joins Robust Data Analysis) funded by the German federal state of North Rhine-Westphalia.
Maximilan Muschalik and Fabian Fumagalli gratefully acknowledge funding by the Deutsche Forschungsgemeinschaft (DFG, German Research Foundation): TRR 318/1 2021 – 438445824.

\bibliography{references}

\appendix
\onecolumn
\paragraph{Organization of the Appendix.}

Within the Appendix, we provide not only proofs for our theoretical analysis in \cref{app:Proofs} and further empirical results in \cref{app:empirical_results}, but also give a table of frequently used symbols throughout the paper in\cref{app:symbols}, provide Shapley-based interaction measures and other indices falling under the notion of cardinal interaction indices explicitly in \cref{app:CII}, provide further and more detailed pseudocode of our algorithmic approach SVARM-IQ in \cref{app:Pseudocode}, showcase our method at the popular special case of the Shapley Interaction index for pairs \cref{app:Pairs}, and describe the used models, datasets, and explanation tasks within our experimental setup in \cref{app:models_datasets}. \cref{app:hardware_details} contains the hardware details.

\startcontents[sections]
\printcontents[sections]{l}{1}{\setcounter{tocdepth}{2}}

\clearpage
\section{LIST OF SYMBOLS} \label{app:symbols}

\begin{table}[h]
\centering
	\begin{tabular}{l|l}	
		\hline
		\multicolumn{2}{c}{\textbf{Problem setting}} \\
		\hline
        $\mathcal{N}$ & Set of players \\
        $\mathcal{N}_k$ & Set of all subsets of the player set with cardinality $k$ \\
        $n$ & Number of players \\
        $\nu$ & Value function \\
        $B$ & Budget, number of allowed evaluations of $\nu$ \\
        $k$ & Considered interaction order \\
        $K$ & Interaction set \\
        $\Delta_K(S)$ & Discrete derivative of players $K$ at coalition $S$ \\
        $I_K$ & Cardinal interaction index of players $K$ \\
        $\hat{I}_K$ & Estimated cardinal interaction index of players $K$ \\
        $\mathcal{L}_k$ & Set of all coalition sizes at which interactions of order $k$ can occur \\
        $\lambda_{k,\ell}$ & Weight of each coalition of size $\ell$ for interaction order $k$ \\
        \hline
		\multicolumn{2}{c}{\textbf{SVARM-IQ}} \\
		\hline
        $I_{K,\ell}^W$ & Average worth of coalitions $S \cup W$ with $S \subseteq \mathcal{N} \setminus K$, $|S| = \ell$ and $W \subseteq K$ \\
        $\hat{I}_{K,\ell}^W$ & Estimate of $I_{K,\ell}^W$ \\
        $c_{K,\ell}^W$ & Number of samples observed for stratum $I_{K,\ell}^W$ \\
        $\mathcal{S}_{\text{exp}}$ & Sizes for which all coalitions are evaluated for explicit stratum computation \\
        $\mathcal{S}_{\text{imp}}$ & Sizes for which coalitions are sampled for implicit stratum estimation \\
        $\mathcal{I}_{\text{exp}}$ & Set of all explicitly computed strata \\
        $\mathcal{I}_{\text{imp}}$ & Set of all implicitly extimated strata \\
        $\mathcal{L}_k^w$ & Set of implicit coalitions sizes $\ell$ depending on $W$ of $I_{K,\ell}^W$ \\
        $P_k$ & Probability distribution over sizes $\{2,\ldots,n-2\}$ \\
        $\bar{P}_k$ & Altered probability distribution over sizes $\{s_{\text{exp}}+1,\ldots,n-s_{\text{exp}}-1\}$ \\
        $\bar{B}$ & Budget left for the sampling loop after \textsc{\texttt{ComputeBorders}} \\
        $\tilde{B}$ & Budget left for the sampling loop after \textsc{\texttt{ComputeBorders}} and \textsc{\texttt{WarmUp}} \\
        $\sigma_{K,W,\ell}^2$ & Variance of coalition worths in stratum $I_{K,\ell}^W$ \\
        $r_{K,\ell,W}$ & Range of coalition worths in stratum $I_{K,\ell}^W$ \\
        $\bar{m}_{K,\ell}^W$ & Number of samples with which $\hat{I}_{K,\ell}^W$ is updated after the warm-up \\
        $m_{K,\ell}^W$ & Total number of samples with which $\hat{I}_{K,\ell}^W$ is updated \\
        $A_{K,\ell,m}^W$ & $m$-th coalition used to update $\hat{I}_{K,\ell}^W$ \\
        \hline
	\end{tabular}
 \caption{List of symbols used frequently throughout the paper.} \label{tab:Symbols}
\end{table}

\clearpage
\section{CARDINAL INTERACTION INDICES AND THEIR WEIGHTS} \label{app:CII}

All Shapley-based interaction indices and a few other game-theoretic measures of interaction can be captured under the notion of cardinal interaction indices (CII).
We have stated this in Section~2 without presenting the aforementioned indices explicitly.
We catch up on this by providing the weights $(\lambda_{k,\ell})_{\ell \in \mathcal{L}_k}$ of each index that is contained within the CII
\begin{equation*}
    I_K = \sum\limits_{S \subseteq \mathcal{N} \setminus K} \lambda_{k,|S|} \cdot \Delta_K(S)
\end{equation*}

with discrete derivative
\begin{equation*}
    \Delta_K(S) = \sum\limits_{W \subseteq K} (-1)^{|K| - |W|} \cdot \nu(S \cup W).
\end{equation*}

\begin{itemize}
    \item Shapley Interaction index (SII) \citep{Grabisch_Roubens_1999}:
    \begin{equation*}
        \lambda_{k,\ell}^{\text{SII}} = \frac{1}{(n-k+1) \binom{n-k}{\ell}}
    \end{equation*}
    \item Shapley-Taylor Interaction index (STI) \citep{Sundararajan_Dhamdhere_Agarwal_2020}:
    \begin{equation*}
        \lambda_{k,\ell}^{\text{STI}} = \frac{k}{n \binom{n-1}{\ell}}
    \end{equation*}
    \item Faithful-Shapley Interaction index (FSI) \citep{Tsai_Yeh_Ravikumar_2022}:
    \begin{equation*}
        \lambda_{k,\ell}^{\text{FSI}} = \frac{(2k-1)!}{((k-1)!)^2} \cdot \frac{(n-\ell-1)!(\ell+k-1)!}{(n+k-1)!}
    \end{equation*}
    \item Banzhaf Interaction index (BII) \citep{Grabisch_Roubens_1999}:
    \begin{equation*}
        \lambda_{k,\ell}^{\text{BII}} = \frac{1}{2^{n-k}}
    \end{equation*}
\end{itemize}

For $k=1$, the SII, STI, and FSI are identical and equal to the Shapley value:

\begin{equation*}
    \phi_i = \sum\limits_{S \subseteq \mathcal{N} \setminus \{i\}} \frac{1}{n \binom{n-1}{\ell}} \cdot \left[ \nu(S \cup \{i\}) - \nu(S) \right],
\end{equation*}

which is why these are also called Shapley-based interactions.
For a comprehensive overview of the axiomatic background justifying these indices, we refer to \citep{Tsai_Yeh_Ravikumar_2022} and \citep{Fumagalli.2023}.

\paragraph{n-SII Values.}

The n-Shapley Values (n-SII) $\Phi^n$ were introduced by \citet{pmlr-v206-bordt23a} as an extension of the Shapley interactions \cite{Lundberg.2020} to higher orders.
The n-SII constructs an interaction index for interactions up to size $n$, which is efficient, i.e.\ the sum of all interactions equals the full model $\nu(\mathcal N)$.
The n-SII are based on SII, $I^{\text{SII}}$, and aggregate SII up to order n.
The highest interaction order of n-SII is always equal to SII.
For every lower order, the n-SII values are constructed recursively, as

\begin{align*}             
    \Phi^n_K := \begin{cases}
            I_K^{\text{SII}}(S) &\text{if } \vert K \vert = n
            \\
            \Phi^{n -1}_K + B_{n-\vert K \vert} \sum_{\substack{\tilde K \subseteq \mathcal N \setminus K \\ \vert K \vert + \vert \tilde K \vert = n}} I^{\text{SII}}_{K \cup \tilde K} &\text{if } \vert K \vert < n,
        \end{cases}
    \end{align*}
where the initial values are the SV $\Phi^1 \equiv \phi$ and $B_n$ are the Bernoulli numbers.
It was shown \cite{pmlr-v206-bordt23a} that n-SII yield an efficient index, i.e.
\begin{equation*}
    \sum_{\substack{K \subseteq \mathcal N \\ \vert K \vert \leq n}} \Phi^n_K = \nu(\mathcal N).
\end{equation*}

\clearpage
\section{ADDITIONAL PSEUDOCODE} \label{app:Pseudocode}

\subsection{Computing Border Sizes}
We have only sketched the \textsc{\texttt{ComputeBorders}} procedure and will provide it now in full detail (see \cref{alg:ComputeBorders}).
Its purpose is to determine the coalition sizes $\mathcal{S}_{\text{exp}}$ for which all coalitions are to be evaluated such that the corresponding strata are computed explicitly.
We construct this set symmetrically, in the sense that a size $s_{\text{exp}}$ is determined such that all $\mathcal{S}_{\text{exp}} = \{0,\ldots,s_{\text{exp}},n-s_{\text{exp}},\ldots,n\}$, in other words: the smallest and the largest $s_{\text{exp}}$ many set sizes are included.
Hence, we assume for simplicity that the initial probability distribution over sizes $P_k$ is symmetric, i.e., $P_k(s) = P_k(n-s)$, although it does not pose a challenge to extend this to any $P_k$ of arbitrary shape.

We start with $s_{\text{exp}} = 1$ and adjust the remaining budget $\bar{B}$ and the altered probability distribution over sizes $\bar{P}_k$.
For each size $s$ being included into $\mathcal{S}_{\text{exp}}$, we set its probability mass to zero and upscale the remaining entries, effectively transferring probability mass from the border sizes to the middle.
According to this procedure, \textsc{\texttt{ComputeBorders}} constructs $\bar{P}_k$ with
\begin{equation*}
    \bar{P}_k(s) = \frac{P_k(s)}{\sum\limits_{s' \in \mathcal{S}_{\text{imp}}} P_k(s')} \text{ for all } s \in \mathcal{S}_{\text{imp}} \hspace{0.3cm} \text{ and } \hspace{0.3cm} \bar{P}_k(s) = 0 \text{ for all } s \in \mathcal{S}_{\text{exp}}.
\end{equation*}

\begin{algorithm}[H]
\caption{\textsc{\texttt{ComputeBorders}}}
\label{alg:ComputeBorders}
\begin{algorithmic}[1]
    \STATE $s_{\text{exp}} \leftarrow 1$
    \STATE $\bar{B} \leftarrow B - 2n - 2$
    \STATE $\bar{P}_k(0),\bar{P}_k(1),\bar{P}_k(n-1),\bar{P}_k(n) \leftarrow 0$
    \STATE $\bar{P}_k(s) \leftarrow \frac{P_k(s)}{1 - P_k(0) - P_k(1) - P_k(n-1) - P_k(n)}$ for all $s \in \{2,\ldots,n-2\}$
    \WHILE{$s_{\text{exp}} + 1 \leq \frac{n}{2}$ \textbf{and} $\binom{n}{s_{\text{exp}}+1} \leq \bar{P}_k(s_{\text{exp}} + 1) \cdot \bar{B}$}
         \STATE $s_{\text{exp}} \leftarrow s_{\text{exp}} + 1$
        \IF{$s_{\text{exp}} = \frac{n}{2}$}
            \STATE $\bar{B} \leftarrow \bar{B} - \binom{n}{s_{\text{exp}}}$
            \STATE $\bar{P}_k \leftarrow \text{Unif}(0,n)$
        \ELSIF{$s_{\text{exp}} = \frac{n-1}{2}$}
            \STATE $\bar{B} \leftarrow \bar{B} - 2\binom{n}{s_{\text{exp}}}$
            \STATE $\bar{P}_k \leftarrow \text{Unif}(0,n)$
        \ELSE
            \STATE $\bar{B} \leftarrow \bar{B} - 2\binom{n}{s_{\text{exp}}}$
            \STATE $\bar{P}_k(s_{\text{exp}}) \leftarrow 0$
            \STATE $\bar{P}_k(n - s_{\text{exp}}) \leftarrow 0$
            \STATE $\bar{P}_k(s) \leftarrow \frac{\bar{P}_k(s)}{1 - 2\bar{P}(s_{\text{exp}})}$ for all $s \in \{s_{\exp}+1,\ldots,n-s_{\exp}-1\}$
        \ENDIF
    \ENDWHILE
    \STATE $\mathcal{S}_{\text{exp}} \leftarrow \{0,\ldots,s_{\text{exp}},n-s_{\text{exp}},\ldots,n\}$
    \STATE $\mathcal{S}_{\text{imp}} \leftarrow \{s_{\text{exp}}+1,\ldots,n-s_{\text{exp}}-1\}$
    \FOR{$s \in \mathcal{S}_{\text{exp}}$} 
        \FOR{$A \in \mathcal{N}_s$}       
            \STATE $v \leftarrow \nu(A)$
            \FOR{$K \in \mathcal{N}_k$}
                \STATE $W \leftarrow A \cap K$
                \STATE $\ell \leftarrow s - |W|$
                \STATE $\hat{I}_{K,\ell}^W \leftarrow \hat{I}_{K,\ell}^W + \frac{v}{\binom{n-k}{\ell}}$
            \ENDFOR
        \ENDFOR
    \ENDFOR
    \STATE {\bfseries Output:} $\mathcal{S}_{\text{exp}}, \mathcal{S}_{\text{imp}}$
\end{algorithmic}
\end{algorithm}

\textsc{\texttt{ComputeBorders}} iterates over sizes in increasing manner, checking whether the reminaing budget $\bar{B}$ is large enough such that the number of coalitions of the next size $s_{\text{exp}} + 1$ considered is covered by the expected number of drawn coalitions with that size.
As long as this holds true, $s_{\text{exp}}$ is incremented and $\bar{B}$ as well as $\bar{P}_k$ are adjusted.
Note that thus not only $s_{\text{exp}} + 1$ is added to $\mathcal{S}_{\text{exp}}$ but also $n-s_{\text{exp}}-1$.
We distinguish between different cases, depending on whether the incremented $s_{\text{exp}}$ has reached the middle of the range of coalition sizes.
In case of even $n$ this is $\frac{n}{2}$, otherwise $\frac{n-1}{2}$.
As soon as $s_{\text{exp}}$ reaches that number, $\bar{P}_k(s)$ becomes irrelevant because then all coalitions of all sizes are being evaluated, leaving no strata to be estimated.
In this case we simply set $\bar{P}_k$ to the uniform distribution such that it is well-defined.

After the computation of $\mathcal{S}_{\text{exp}}$, we evaluate all coalitions with cardinality $s \in \mathcal{S}_{\text{exp}}$.
For each such coalition $A$ we update the estimate $\hat{I}_{K,\ell}^W$ with $W = A \cap K$ and $\ell = s - |W|$ according to our update mechanism.
Since each stratum contains only coalitions of the same size, this leads to exactly computed strata representing the average of the contained coalitions' worths.

\subsection{Warm-up}
The \textsc{\texttt{WarmUp}} (see \cref{alg:WarmUp}) procedure guarantees that each stratum estimate $\hat{I}_{K,\ell}^W$ with $I_{K,\ell}^W \in \mathcal{I}_{\text{imp}}$ is initialized with the worth of one sampled coalition.
This is a natural requirement to facilitate our theoretical analysis in \cref{app:Proofs}.
We achieve this algorithmically by iterating over all combinations of $K \in \mathcal{N}_k$, $W \subseteq K$, and $\ell \in \mathcal{L}_k^{|W|}$.
Each such combination specifies a stratum that is implicitly to be estimated.
\textsc{\texttt{WarmUp}} draws for each stratum a coalition $A$ uniformly at random from the set of all coalitions of size $\ell$ and not containing any player of $K$.
The estimate $\hat{I}_{K,\ell}^W$ is then set to the evaluated worth $\nu(A \cup W)$ and the counter of observed samples is set to one.
The spent budget is:
\begin{align*}
    |\mathcal{I}_{\text{imp}}| = & \ \binom{n}{k} \cdot \sum\limits_{w=0}^k \binom{k}{w} \vert \mathcal{L}_k^{|w|} \vert \\
    = & \ \binom{n}{k} \cdot \sum\limits_{w=0}^k \binom{k}{w} \vert \{\max\{0,s_{\text{exp}}+1-w\},\ldots,\min\{n-k,n-s_{\text{exp}}-1-w\}\}\vert \\
    = & \ \binom{n}{k} \cdot \sum\limits_{w=0}^k \binom{k}{w} \left( n - \max\{k,s_{\text{exp}}+1+w\} - \max\{0,s_{\text{exp}}+1-w\} +1 \right).
\end{align*}

\begin{algorithm}[H]
\caption{\textsc{\texttt{WarmUp}}}
\label{alg:WarmUp}
\begin{algorithmic}[1]
    \FOR{$K \subseteq \mathcal{N}_k$}
        \FOR{$W \subseteq K$}
            \FOR{$\ell \in \mathcal{L}_k^{|W|}$}
                \STATE Draw $A$ from $\{S \subseteq \mathcal{N} \setminus K \mid |S| = \ell\}$ uniformly at random
                \STATE $\hat{I}_{K,\ell}^W \leftarrow \nu(A \cup W)$
                \STATE $c_{K,\ell}^W \leftarrow 1$
            \ENDFOR
        \ENDFOR
    \ENDFOR
\end{algorithmic}
\end{algorithm}

\subsection{Updating Strata Mean Estimates}

In order to update the mean estimates $\hat{I}_{K,\ell}^W$ of the estimated strata incrementally with a single pass, thus not requiring to iterate over all previous samples, we use \textsc{\texttt{UpdateMean}} (see \cref{alg:UpdateMean}).
Besides the old estimate and the newly observed coalition worth $v_b$, this requires the number of observations made so far given by $c_{K,\ell}^W$.

\begin{algorithm}[H]
\caption{\textsc{\texttt{UpdateMean}}}
\label{alg:UpdateMean}
\begin{algorithmic}[1]
    \STATE {\bfseries Input:} $\hat{I}_{K,\ell}^W$, $c_{K,\ell}^W$, $v_b$
    \STATE {\bfseries Output:} $\frac{\hat{I}_{K,\ell}^W \cdot c_{K,\ell}^W + v_b}{c_{K,\ell}^W + 1}$
\end{algorithmic}
\end{algorithm}

\clearpage
\section{THE SPECIAL CASE OF PAIRWISE SHAPLEY-INTERACTIONS} \label{app:Pairs}

We stated our approximation algorithm SVARM-IQ for all CII and any order $k$.
Since the Shapley Interaction index (SII) for pairs, i.e., $k=2$, is the most popular among them, we provide a description of SVARM-IQ and the pseudocode (see \cref{alg:SVARM-IQPairs}) for that specific case, leading to a simpler presentation of our approach.

The SII of a pair of players $\{i,j\} \in \mathcal{N}_2$ is given by
\begin{equation*}
    I_{i,j}^{\text{SII}} = \sum\limits_{S \subseteq \mathcal{N} \setminus \{i,j\}} \frac{1}{(n-1)\binom{n-2}{|S|}} \left[ \nu(S \cup \{i,j\}) - \nu(S \cup \{i\}) - \nu(S \cup \{j\}) + \nu(S) \right].
\end{equation*}

Now, our approach stratifies the discrete derivatives $\Delta_{i,j}(S)$ by size and splits them into multiple strata, yielding the following representation of the SII:
\begin{equation*}
    I_{i,j}^{\text{SII}} = \frac{1}{n-1} \sum\limits_{\ell=0}^{n-2} I_{i,j,\ell}^{\{i,j\}} - I_{i,j,\ell}^{\{i\}} - I_{i,j,\ell}^{\{i\}} + I_{i,j,\ell}^{\emptyset},
\end{equation*}

with strata terms for all $W \subseteq \{i,j\}$ and $\ell \in \mathcal{L}_2 := \{0,\ldots,n-2\}$:
\begin{equation*}
     I_{i,j,\ell}^W := \frac{1}{\binom{n-2}{\ell}} \sum\limits_{\substack{S \subseteq \mathcal{N} \setminus \{i,j\} \\ |S| = \ell}} \nu(S \cup W).
\end{equation*}

We keep a stratum estimate $\hat{I}_{i,j,\ell}^W$ for each pair $i$ and $j$, size $\ell \in \mathcal{L}_2$, and subset $W \subseteq \{i,j\}$.
Subsequently, the aggregation of the strata estimates, which we obtain during sampling, provides the desired SII estimate:
\begin{equation*}
    \hat{I}_{i,j}^{\text{SII}} := \frac{1}{n-1} \sum\limits_{\ell=0}^{n-2} \hat{I}_{i,j,\ell}^{\{i,j\}} - \hat{I}_{i,j,\ell}^{\{i\}} - \hat{I}_{i,j,\ell}^{\{i\}} + \hat{I}_{i,j,\ell}^{\emptyset}.
\end{equation*}

For each sampled coalition $A$ of size $|A|=a$, the update mechanism needs to distinguish between only 4 cases.
For each pair $i$ and $j$ it updates:
\begin{itemize}
    \item $\hat{I}_{i,j,a-2}^{\{i,j\}}$ if $i,j \in A$,
    \item $\hat{I}_{i,j,a-1}^{\{i\}}$ if $i \in A$ but $j \notin A$,
    \item $\hat{I}_{i,j,a-1}^{\{j\}}$ if $j \in A$ but $i \notin A$, or
    \item $\hat{I}_{i,j,a}^{\emptyset}$ if $i,j \notin A$.
\end{itemize}

This case distinction is still captured by computing $W = A \cap K$, $\ell = a - |W|$, and updating $\hat{I}_{i,j,\ell}^W$.

\begin{algorithm}[H]
\caption{SVARM-IQ (for the Shapley Interaction index of order $k=2$)}
\label{alg:SVARM-IQPairs}
\begin{algorithmic}[1]
    \STATE {\bfseries Input:} $(\mathcal{N}, \nu)$, $B \in \mathbb{N}$
    \STATE $\hat{I}_{i,j,\ell}^{\emptyset}, \hat{I}_{i,j,\ell}^{\{i\}}, \hat{I}_{i,j,\ell}^{\{j\}}, \hat{I}_{i,j,\ell}^{\{i,j\}} \leftarrow 0$ \hspace{0.5cm} for all $\{i,j\} \in \mathcal{N}_2, \ell \in \mathcal{L}_2$
    \STATE $c_{i,j,\ell}^{\emptyset}, c_{i,j,\ell}^{\{i\}}, c_{i,j,\ell}^{\{j\}}, c_{i,j,\ell}^{\{i,j\}} \leftarrow 0$ \hspace{0.5cm} for all $\{i,j\} \in \mathcal{N}_2, \ell \in \mathcal{L}_2$
    \STATE \textsc{\texttt{ComputeBorders}}
    \STATE $\bar{B} \leftarrow B - \sum\limits_{s \in \mathcal{S}_{\text{exp}}} \binom{n}{s}$
    \FOR{$b = 1, \ldots, \bar{B}$}
        \STATE Draw size $a_b \in \mathcal{S}_{\text{imp}} \sim \bar{P}_k$
        \STATE Draw $A_b$ from $\{S \subseteq \mathcal{N} \mid |S| = a_b\}$ uniformly at random
        \STATE $v_b \leftarrow \nu(A_b)$
        \FOR{$\{i,j\} \in \mathcal{N}_2$}
            \STATE $W \leftarrow A_b \cap \{i,j\}$
            \STATE $\ell \leftarrow a_b - |W|$
            \STATE $\hat{I}_{i,j,\ell}^W \leftarrow \frac{\hat{I}_{i,j,\ell}^W \cdot c_{i,j,\ell}^W + v_b }{c_{i,j,\ell}^W+1}$
            \STATE $c_{i,j,\ell}^W \leftarrow c_{i,j,\ell}^W + 1$
        \ENDFOR
    \ENDFOR
    \STATE $\hat{I}_{i,j} \leftarrow \frac{1}{n-1} \sum\limits_{\ell=0}^{n-2} \hat{I}_{i,j,\ell}^{\{i,j\}} - \hat{I}_{i,j,\ell}^{\{i\}} - \hat{I}_{i,j,\ell}^{\{j\}} + \hat{I}_{i,j,\ell}^{\emptyset}$  \hspace{0.5cm} for all $\{i,j\} \in \mathcal{N}_2$
    \STATE {\bfseries Output:} $\hat{I}_{i,j}$ for all $\{i,j\} \in \mathcal{N}_2$
\end{algorithmic}
\end{algorithm}

\clearpage
\section{PROOFS} \label{app:Proofs}

In the following, we give the proofs to our theoretical results in Section~4.
We start by defining and revisiting some helpful notation and stating our assumptions.

\textbf{Notation:}
\begin{itemize}
\item Let $\mathcal{L}_k := \{0,\ldots,n-k\}$.
\item Let $\mathcal{L}_k^{|W|} := \{\ell \in \mathcal{L}_k \mid \ell + |W| \in \mathcal{S}_{\text{imp}} \} = [\max\{0,s_{\text{exp}}+1-w\}, \min\{n-k,n-s_{\text{exp}}-1\}]$ for any $W \subseteq K \in \mathcal{N}_k$.
\item Let $\tilde{B} = B - \sum\nolimits_{s \in \mathcal{S}_{\text{exp}}} \binom{n}{s} - \vert \mathcal{I}_{\text{imp}} \vert$ be the available budget left for the sampling loop after the completion of \textsc{\texttt{ComputeBorders}} and \textsc{\texttt{WarmUp}}.
\item For all $K \in \mathcal{N}_k$ with $\ell \in \mathcal{L}_k$, let $A_{K,\ell}$ be a random set with $\mathbb{P}(A_{K,\ell} = S) = \frac{1}{\binom{n-k}{\ell}}$ for all $S \subseteq \mathcal{N} \setminus K$ with $|S| = \ell$.
\item For all $K \in \mathcal{N}_k$ with $W \subseteq K$ and $\ell \in \mathcal{L}_k^w$:
\begin{itemize}
    \item Let $\sigma_{K,\ell,W}^2 := \mathbb{V}[\nu(A_{K,\ell} \cup W)]$ be the strata variance.
    \item Let $r_{K,\ell,W} := \max\limits_{\substack{S \subseteq \mathcal{N} \setminus K \\ |S| = \ell}} \nu(S \cup W) - \min\limits_{\substack{S \subseteq \mathcal{N} \setminus K \\ |S| = \ell}} \nu(S \cup W)$ be the strata range.
    \item Let $\bar{m}_{K,\ell}^W :=  \# \{b \mid |A_b| = \ell + |W|, A_b \cap K = W \}$ be the number of samples with which $\hat{I}_{K,\ell}^W$ is updated during the sampling loop.
    \item Let $m_{K,\ell}^W := \bar{m}_{K,\ell}^W + 1$ be the total number of samples with which $\hat{I}_{K,\ell}^W$ is updated.
    \item Let $A_{K,\ell,m}^W$ be the $m$-th coalition used to update $I_{K,\ell}^W$.
\end{itemize}
\item For all $K \in \mathcal{N}_k$ let $R_K := \sum\limits_{W \subseteq K} \sum\limits_{\ell \in \mathcal{L}_k^{|W|}} r_{K,\ell,W}$.
\item Let $\gamma_k$ be $\gamma_2 := 2(n-1)^2$ for $k=2$ and $\gamma_k := n^{k-1} (n-k+1)^2$ for $k \geq 3$.
\end{itemize}

\textbf{Assumptions:}
\begin{itemize}
    \item $\tilde{B} > 0$
    \item $n \geq 4$
    \item $B < 2^n$
\end{itemize}
The lower bound on the leftover budget $\tilde{B}$ is necessary to ensure the completion of \textsc{\texttt{ComputeBorders}} and \textsc{\texttt{WarmUp}}, and that at least one coalition is sampled during the sampling loop.
The assumption on $n$ arises from the fact that \textsc{\texttt{ComputeBorders}} automatically evaluates the worth of all coalitions having size $0,1,n-1$ or $n$.
Hence, all CII values are computed exactly for $n=3$.
Our considered problem statement becomes trivial for $n \leq 2$.
Likewise, in order to avoid triviality, we demand the budget to be lower than the total number of coalitions $2^n$.
Otherwise, all CII values will be computed exactly by \textsc{\texttt{ComputeBorders}} and the approximation problem vanishes.
This allows us to state $\mathcal{S}_{\text{imp}} \neq \emptyset$ and $\mathcal{I}_{\text{imp}} \neq \emptyset$.

\subsection{Unbiasedness}

\begin{lemma} \label{lem:UnbiasedStrata}
    All strata estimates $\hat{I}_{K,\ell}^W$ are unbiased, i.e., for all $K \in \mathcal{N}_k$, $W \subseteq K$, $\ell \in \mathcal{L}_k$:
    \begin{equation*}
        \mathbb{E} \left[ \hat{I}_{K,\ell}^W \right] = I_{K,\ell}^W .
    \end{equation*}
\end{lemma}

\begin{proof}
The statement trivially holds for all strata explicitly computed by \textsc{\texttt{ComputeBorders}}.
Thus, we consider the remaining strata which are estimated via sampling.
Fix any $K \in\mathcal{N}_k$, $W \subseteq K$, and $\ell \in \mathcal{L}_k^{|W|}$. 
Due to the uniform sampling of eligible coalitions once the size is fixed, we have:

\begin{align*}
    & \ \mathbb{E} \left[ \hat{I}_{K,\ell}^W \mid m_{K,\ell}^W \right] \\
    = & \ \frac{1}{m_{K,\ell}^W} \sum\limits_{m=1}^{m_{K,\ell}^W} \mathbb{E} \left[ \nu(A_{K,\ell,m}^W) \mid m_{K,\ell}^W \right] \\
    = & \ \frac{1}{m_{K,\ell}^W} \sum\limits_{m=1}^{m_{K,\ell}^W} \sum\limits_{\substack{S \subseteq \mathcal{N} \setminus K \\ |S| = \ell}} \mathbb{P}(A_{K,\ell,m}^W = S \cup W \mid |A_{K,\ell,m}^W| = \ell + |W|, A_{K,\ell,m}^W \cap K = W) \cdot \nu(S \cup W) \\
    = & \ \frac{1}{m_{K,\ell}^W} \sum\limits_{m=1}^{m_{K,\ell}^W} \sum\limits_{\substack{S \subseteq \mathcal{N} \setminus K \\ |S| = \ell}} \frac{1}{\binom{n-k}{\ell}} \cdot \nu(S \cup W) \\
    = & \ \frac{1}{m_{K,\ell}^W} \sum\limits_{m=1}^{m_{K,\ell}^W} I_{K,\ell}^W \\
    = & \ I_{K,\ell}^W .
\end{align*}

Note that the set $A_{K,\ell,m}^W$ has cardinality $\ell+|W|$ and fulfills $A_{K,\ell,m}^W \cap K = W$ by definition.
Otherwise, it would not be used to update $\hat{I}_{K,\ell}^W$.
Since \textsc{\texttt{WarmUp}} gathers one sample for each stratum estimate, it guarantees $m_{K,\ell}^W \geq 1$.
Thus the above terms are well defined.
Finally, we obtain:
\begin{align*}
    \mathbb{E} \left[ \hat{I}_{K,\ell}^W \right] = & \ \sum\limits_{m=1}^{\bar{B}+1} \mathbb{E} \left[ \hat{I}_{K,\ell}^W \mid m_{K,\ell}^W = m \right] \cdot \mathbb{P}(m_{K,\ell}^W = m) \\
    = & \ \sum\limits_{m=1}^{\bar{B}+1} I_{K,\ell}^W \cdot \mathbb{P}(m_{K,\ell}^W = m) \\
    = & I_{K,\ell}^W .
\end{align*}
\end{proof}

\textbf{Theorem 4.1.}
\textit{
    The CII estimates returned by SVARM-IQ are unbiased for all $K \in \mathcal{N}_k$, i.e.,
    \begin{equation*}
        \mathbb{E} \left[ \hat{I}_K \right] = I_K.
    \end{equation*}
}

\begin{proof}
We have already proven the unbiasedness of all strata estimates with \cref{lem:UnbiasedStrata}.
Thus, we obtain for all $K \in \mathcal{N}_k$:
\begin{align*}
     \mathbb{E} \left[ \hat{I}_K \right] = & \mathbb{E} \left[ \sum\limits_{\ell=0}^{n-k} \binom{n-k}{\ell} \lambda_{k,\ell} \sum\limits_{W \subseteq K} (-1)^{k-|W|} \cdot \hat{I}_{K,\ell}^W \right] \\
     = & \ \sum\limits_{\ell=0}^{n-k} \binom{n-k}{\ell} \lambda_{k,\ell} \sum\limits_{W \subseteq K} (-1)^{k-|W|} \cdot \mathbb{E} \left[ \hat{I}_{K,\ell}^W \right] \\
    = & \ \sum\limits_{\ell=0}^{n-k} \binom{n-k}{\ell} \lambda_{k,\ell} \sum\limits_{W \subseteq K} (-1)^{k-|W|} \cdot I_{K,\ell}^W \\
    = & \ I_K.
\end{align*}
\end{proof}

\subsection{Sample Numbers}
Form now on, we distinguish between the special case of order $k = 2$ and all others $k \geq 3$, allowing us to give tighter bounds for the former.
Hence, we introduce $\gamma_k$ for all $k \geq 2$ with
\begin{equation*}
    \gamma_2 = 2(n-1)^2 \text{ and } \gamma_k = n^{k-1} (n-k+1)^2 \text{ for all } k \geq 3.
\end{equation*}

\begin{lemma} \label{lem:ExpectedSamples}
    The number of samples $\bar{m}_{K,\ell}^W$ collected for the strata estimate $\hat{I}_{K,\ell}^W$ of any fixed player set $K \in \mathcal{N}_k$, $W \subseteq K$, and $\ell \in \mathcal{L}_k^{|W|}$ collected during the sampling loop is binomially distributed with an expected value of at least
    \begin{equation*}
        \mathbb{E} \left[ \bar{m}_{K,\ell}^W \right] \geq \frac{\tilde{B}}{\gamma_k}.
    \end{equation*}
\end{lemma}

\begin{proof}
The number of collected samples during the sample loop, i.e. $\bar{m}_{K,\ell}^W$, is binomially distributed because in each iteration the stratum $\hat{I}_{K,\ell}^W$ has the same probability to be updated and the sampled coalitions are independent of each other across the iterations.
The number of iterations is $\tilde{B}$ and the condition for an update of $\hat{I}_{K,\ell}^W$ is that the sampled set $A_b$ fulfills $|A_b| = a_b = \ell+|W|$ and $A_b \cap K = W$.
This happens with a probability of:
\begin{align*}
    & \ \mathbb{P}(a_b = \ell + |W|, A_b \cap K = W) \\
    = & \ \mathbb{P}(A_b \cap K = W \mid a_b = \ell + |W|) \cdot \mathbb{P}(a_b = \ell + |W|) \\
    = & \ \frac{\binom{n-k}{\ell}}{\binom{n}{\ell+|W|}} \cdot \bar{P}_k(\ell+|W|).
\end{align*}

Hence, we obtain $\bar{m}_{K,\ell}^W \sim Bin \left(\tilde{B}, \frac{\binom{n-k}{\ell}}{\binom{n}{\ell+|W|}} \cdot \bar{P}_k(\ell+|W|) \right)$.
This yields
\begin{align*}
    \mathbb{E} \left[ \bar{m}_{K,\ell}^W \right] = & \ \tilde{B} \cdot \frac{\binom{n-k}{\ell}}{\binom{n}{\ell+|W|}} \cdot \bar{P}_k(\ell+|W|) \\
    \geq & \ \tilde{B} \cdot \frac{\binom{n-k}{\ell}}{\binom{n}{\ell+|W|}} \cdot P_k(\ell+|W|) .
\end{align*}

Note that $\bar{P}_k(\ell + |W|) \geq P_k(\ell + |W|)$ holds true for all $\ell$ and $W \subseteq K$ with $\ell + |W| \in \mathcal{S}_{\text{imp}}$ because for these sizes, from which coalitions are sampled, $\bar{P}_k$ can only gain probability mass in comparison to $P_k$.
More precisely, for all $s \in \mathcal{S}_{\text{imp}}$ we have
\begin{equation*}
    \bar{P}_k(s) = \frac{P_k(s)}{\sum\limits_{s' \in \mathcal{S}_{\text{imp}}} P_k(s')} \geq \frac{P_k(s)}{\sum\limits_{s' \in \mathcal{S}_{\text{exp}}} P_k(s') + \sum\limits_{s' \in \mathcal{S}_{\text{imp}}} P_k(s')} = P_k(s).
\end{equation*}

We continue to prove our statement for the case of $k=2$ and any fixed $K = \{i,j\}$ by giving a lower bound for the expected value of $\bar{m}_{i,j,\ell}^W$.
Inserting $k=2$, we can further write
\begin{align*}
    \mathbb{E} \left[ \bar{m}_{i,j,\ell}^W \right] \geq & \ \tilde{B} \cdot \frac{\binom{n-2}{\ell}}{\binom{n}{\ell+|W|}} \cdot P_2(\ell+|W|) \\
    = & \ \frac{\tilde{B}}{n (n-1)} \cdot \frac{(\ell+|W|)!}{\ell!} \cdot \frac{(n-\ell-|W|)!}{(n-\ell-2)!} \cdot P_2(\ell+|W|).
\end{align*}

Let
\begin{equation*}
    f(\ell,w) := \frac{(\ell+w)!}{\ell!} \cdot \frac{(n-\ell-w)!}{(n-\ell-2)!} =
    \begin{cases}
        (n-\ell) (n-\ell-1) & \text{if } w=0 \\
        (\ell+1) (n-\ell-1) & \text{if } w=1 \\
        (\ell+1) (\ell+2) & \text{if } w=2
    \end{cases}.
\end{equation*}

In the following, we derive the lower bound $f(\ell,|W|) \cdot P_2(\ell+|W|) \geq \frac{n}{2(n-1)}$ for all $|W| \in \{0,1,2\}$ and $\ell \in \mathcal{L}_2^{|W|}$ by distinguishing over different cases of $n$, $\ell$, and $|W|$ and exploiting our tailored distribution $P_2$.

For odd $n$, $\ell+|W| \leq \frac{n-1}{2}$, and $|W| = 0$:
\begin{equation*}
    f(\ell, |W|) \cdot P_2(\ell+|W|)
    = \frac{(n-\ell) (n-\ell-1)}{\ell (\ell-1)} \cdot \frac{n-1}{2(n-3)}
    \geq \frac{(n-1)^2}{2(n-3)^2}
    \geq \frac{n}{2(n-1)}
\end{equation*}

For odd $n$, $\ell+|W| \leq \frac{n-1}{2}$, and $|W| = 1$:
\begin{equation*}
    f(\ell, |W|) \cdot P_2(\ell+|W|)
    = \frac{(\ell+1) (n-\ell-1)}{(\ell+1) \ell} \cdot \frac{n-1}{2(n-3)}
    \geq \frac{(n-1)(n+1)}{2(n-3)^2}
    \geq \frac{n}{2(n-1)}
\end{equation*}

For odd $n$, $\ell+|W| \leq \frac{n-1}{2}$, and $|W| = 2$:
\begin{align*}
    f(\ell, |W|) \cdot P_2(\ell+|W|)
    = \frac{(\ell+1) (\ell+2)}{(\ell+2) (\ell+1)} \cdot \frac{n-1}{2(n-3)}
    = \frac{n-1}{2(n-3)}
    \geq \frac{n}{2(n-1)}
\end{align*}

For odd $n$, $\ell+|W| \geq \frac{n+1}{2}$, and $|W| = 0$:
\begin{equation*}
    f(\ell, |W|) \cdot P_2(\ell+|W|)
    = \frac{(n-\ell) (n-\ell-1)}{(n-\ell) (n-\ell-1)} \cdot \frac{n-1}{2(n-3)}
    = \frac{n-1}{2(n-3)}
    \geq \frac{n}{2(n-1)}
\end{equation*}

For odd $n$, $\ell+|W| \geq \frac{n+1}{2}$, and $|W| = 1$:
\begin{equation*}
    f(\ell, |W|) \cdot P_2(\ell+|W|)
    = \frac{(\ell+1) (n-\ell-1)}{(n-\ell-1) (n-\ell-2)} \cdot \frac{n-1}{2(n-3)}
    \geq \frac{(n-1)(n+1)}{2(n-3)^2}
    \geq \frac{n}{2(n-1)}
\end{equation*}

For odd $n$, $\ell+|W| \geq \frac{n+1}{2}$, and $|W| = 2$:
\begin{equation*}
    f(\ell, |W|) \cdot P_2(\ell+|W|)
    = \frac{(\ell+1) (\ell+2)}{(n-\ell-2) (n-\ell-3)} \cdot \frac{n-1}{2(n-3)}
    \geq \frac{(n-1)(n+1)}{2(n-3)^2}
    \geq \frac{n}{2(n-1)}
\end{equation*}

For even $n$, $\ell+|W| \leq \frac{n-2}{2}$, and $|W| = 0$:
\begin{equation*}
    f(\ell, |W|) \cdot P_2(\ell+|W|)
    = \frac{(n-\ell) (n-\ell-1)}{\ell (\ell-1)} \cdot \frac{n^2-2n}{2(n^2-4n+2)}
    \geq \frac{n^2(n+2)}{2(n-4)(n^2-4n+2)}
    \geq \frac{n}{2(n-1)}
\end{equation*}

For even $n$, $\ell+|W| \leq \frac{n-2}{2}$, and $|W| = 1$:
\begin{equation*}
    f(\ell, |W|) \cdot P_2(\ell+|W|)
    = \frac{(\ell+1) (n-\ell-1)}{(\ell+1) \ell} \cdot \frac{n^2-2n}{2(n^2-4n+2)}
    \geq \frac{n(n-2)(n+2)}{2(n-4)(n^2-4n+2)}
    \geq \frac{n}{2(n-1)}
\end{equation*}

For even $n$, $\ell+|W| \leq \frac{n-2}{2}$, and $|W| = 2$:
\begin{equation*}
    f(\ell, |W|) \cdot P_2(\ell+|W|)
    = \frac{(\ell+1) (\ell+2)}{(\ell+2) (\ell+1)} \cdot \frac{n^2-2n}{2(n^2-4n+2)}
    = \frac{n^2-2n}{2(n^2-4n+2)}
    \geq \frac{n}{2(n-1)}
\end{equation*}

For even $n$, $\ell+|W| \geq \frac{n}{2}$, and $|W| = 0$:
\begin{equation*}
    f(\ell, |W|) \cdot P_2(\ell+|W|)
    = \frac{(n-\ell) (n-\ell-1)}{(n-\ell) (n-\ell-1)} \cdot \frac{n^2-2n}{2(n^2-4n+2)}
    = \frac{n^2-2n}{2(n^2-4n+2)}
    \geq \frac{n}{2(n-1)}
\end{equation*}

For even $n$, $\ell+|W| \geq \frac{n}{2}$, and $|W| = 1$:
\begin{equation*}
    f(\ell, |W|) \cdot P_2(\ell+|W|)
     = \frac{(\ell+1) (n-\ell-1)}{(n-\ell-1) (n-\ell-2)} \cdot \frac{n^2-2n}{2(n^2-4n+2)}
     \geq \frac{n^2}{2(n^2-4n+2)}
     \geq \frac{n}{2(n-1)}
\end{equation*}

For even $n$, $\ell+|W| \geq \frac{n}{2}$, and $|W| = 2$:
\begin{equation*}
    f(\ell, |W|) \cdot P_2(\ell+|W|)
    = \frac{(\ell+1) (\ell+2)}{(n-\ell-2) (n-\ell-3)} \cdot \frac{n^2-2n}{2(n^2-4n+2)}
    \geq \frac{n^2-2n}{2(n^2-4n+2)}
    \geq \frac{n}{2(n-1)}
\end{equation*}

This allows us to conclude:
\begin{align*}
    \mathbb{E} \left[ \bar{m}_{i,j,\ell}^W \right] \geq & \ \tilde{B} \cdot \frac{\binom{n-2}{\ell}}{\binom{n}{\ell+|W|}} \cdot P_2(\ell+|W|) \\
    = & \ \frac{\tilde{B}}{n(n-1)} \cdot \frac{(\ell+|W|)!}{\ell!} \cdot \frac{(n-\ell-|W|)!}{(n-\ell-2)!} \cdot P_2(\ell+|W|) \\
    \geq & \ \frac{\tilde{B}}{n(n-1)} \cdot \frac{n}{2(n-1)} \\
    = & \ \frac{\tilde{B}}{\gamma_2}.
\end{align*}

Next, we turn our attention to the case of $k \geq 3$.
Inserting the uniform distribution for $P_k$, we can write for the expected number of samples:
\begin{align*}
    \mathbb{E} \left[ \bar{m}_{K,\ell}^W \right] \geq & \ \tilde{B} \cdot \frac{\binom{n-k}{\ell}}{\binom{n}{\ell+|W|}} \cdot P_k(\ell+|W|) \\
    = & \ \tilde{B} \cdot \frac{(n-k)!}{n!} \cdot \frac{(\ell+|W|)!}{\ell!} \cdot \frac{(n-\ell-|W|)!}{(n-\ell-k)!} \cdot \frac{1}{n-3} \\
    \geq & \ \tilde{B} \cdot \frac{(n-k)!}{n!} \cdot \frac{1}{n-3}.
\end{align*}

In the following we prove that $\frac{(n-k)!}{n!} \cdot \frac{1}{n-3} \geq \frac{1}{n^{k-1} (n-k+1)^2}$.
First, we obtain the equivalent inequality
\begin{equation*}
    n^{k-1} (n-k+1) \geq (n-3) \prod\limits_{i=n-k+2}^n i .
\end{equation*}

Note that we have $n \geq k$ at all times.
The inequality obviously holds true for all $k \leq 4$.
We prove its correctness for $k \geq 5$ by induction over $k$.
We start with the induction base at $k = 5$:
\begin{align*}
    & \ n^{k-1} (n-k+1) \geq (n-3) \prod\limits_{i=n-k+2}^n i \\
    \Leftrightarrow & \ n^3 (n-4) \geq (n-1)(n-2)(n-3)^2 \\
    \Leftrightarrow & \ 5n^3 + 39n \geq 29n^2 + 18 .
\end{align*}

The resulting equality is obviously fulfilled by all $n \geq 5$.
Next, we conduct the induction step by considering the inequality for $k + 1$ with $k \geq 5$:
\begin{align*}
    & \ n^k (n-k) \\
    = & \ \frac{n(n-k)}{n-k+1} \cdot n^{k-1} (n-k+1) \\
    \geq & \ \frac{n(n-k)}{n-k+1} \cdot (n-3) \prod\limits_{i=n-k+2}^n i \\
    \geq & \ (n-k+1) \cdot (n-3) \prod\limits_{i=n-k+2}^n i \\
    = & \ (n-3) \prod\limits_{i=n-k+1}^n i .
\end{align*}

With the inequality proven, we finally obtain the desired lower bound for the expectation of $\bar{m}_{K,\ell}^W$:
\begin{align*}
    \mathbb{E} \left[ \bar{m}_{K,\ell}^W \right] \geq & \ \tilde{B} \cdot \frac{(n-k)!}{n!} \cdot \frac{1}{n-3} \\
    \geq & \ \frac{\tilde{B}}{n^{k-1} (n-k+1)^2} \\
    = & \ \frac{\tilde{B}}{\gamma_k}.
\end{align*}
\end{proof}

\begin{lemma} \label{lem:InvertedSamples}
    The expected inverted total sample number of the strata estimate $\hat{I}_{K,\ell}^W$ for any fixed $K \in \mathcal{N}_k$, $W \subseteq K$, and $\ell \in \mathcal{L}_k^{|W|}$ is bounded by
    \begin{equation*}
        \mathbb{E} \left[ \frac{1}{m_{K,\ell}^W} \right] \leq \frac{\gamma_k}{\tilde{B}} .
    \end{equation*}
\end{lemma}

\begin{proof}
In the following, we apply equation (3.4) in \citep{Chao.1972}, stating 
\begin{equation*}
    \mathbb{E} \left[ \frac{1}{X + 1} \right] = \frac{1-(1-p)^{m+1}}{(m+1)p} \leq \frac{1}{mp} = \frac{1}{\mathbb{E}[X]} ,
\end{equation*}

for any binomially distributed random variable $X \sim Bin(m,p)$.
Due to \textsc{\texttt{WarmUp}} we have $m_{K,\ell}^W = \bar{m}_{K,\ell}^W + 1$, since it guarantees exactly one sample for each stratum.
Next, \cref{lem:ExpectedSamples} allows us to substitute $X$ with $\bar{m}_{K,\ell}^W$ and we obtain:

\begin{equation*}
    \mathbb{E} \left[ \frac{1}{m_{K,\ell}^W} \right] =  \left[ \frac{1}{\bar{m}_{K,\ell}^W + 1}\right] \leq \frac{1}{\mathbb{E}\left[ \bar{m}_{K,\ell}^W \right]} \leq \frac{\gamma_k}{\tilde{B}}.
\end{equation*}
\end{proof}

\subsection{Variance and Mean Squared Error}

\begin{lemma} \label{lem:VarianceGivenSampleNumbers}
    For any $K \in \mathcal{N}_k$, given the sample numbers $m_{K,\ell}^W$ for all $W \subseteq K$ and $\ell \in \mathcal{L}_k^{|W|}$, the variance of the estimate $\hat{I}_K$ is given by
    \begin{equation*}
        \mathbb{V} \left[ \hat{I}_K \mid \left( m_{K,\ell}^W \right)_{\ell \in \mathcal{L}, W \subseteq K} \right] =  \sum\limits_{W \subseteq K} \sum\limits_{\ell \in \mathcal{L}_k^{|W|}} \binom{n-k}{\ell}^2 \lambda_{k,\ell}^2 \cdot \frac{\sigma_{K,\ell,W}^2}{m_{K,\ell}^W} .
    \end{equation*}
\end{lemma}

\begin{proof}
First, we split the variance of $\hat{I}_K$ with the help of Bienaym\'e's identity into the variances of the strata estimates and their covariances.
Then we make use of the fact that each sample to update a stratum is effectively drawn uniformly:

\begin{align*}
    & \ \mathbb{V} \left[ \hat{I}_K \mid \left( m_{K,\ell}^W \right)_{\ell \in \mathcal{L}, W \subseteq K} \right] \\
    = & \ \mathbb{V} \left[ \sum\limits_{\ell=0}^{n-k} \binom{n-k}{\ell} \lambda_{k,\ell} \sum\limits_{W \subseteq K} (-1)^{k-|W|} \cdot \hat{I}_{K,\ell}^W \mid \left( m_{K,\ell}^W \right)_{\ell \in \mathcal{L}, W \subseteq K} \right] \\
     = & \ \mathbb{V} \left[ \sum\limits_{\ell=0}^{n-k} \sum\limits_{W \subseteq K} \binom{n-k}{\ell} \lambda_{k,\ell} \cdot (-1)^{k-|W|} \cdot \hat{I}_{K,\ell}^W \mid \left( m_{K,\ell}^W \right)_{\ell \in \mathcal{L}, W \subseteq K} \right] \\
    = & \ \sum\limits_{\ell=0}^{n-k} \sum\limits_{W \subseteq K} \binom{n-k}{\ell}^2 \lambda_{k,\ell}^2 \mathbb{V} \left[ \hat{I}_{K,\ell}^W \mid m_{K,\ell}^W \right] \\
    & \ + \sum\limits_{\substack{\ell \in \mathcal{L}_k \\ W \subseteq K }} \sum\limits_{\substack{\ell' \in \mathcal{L}_k \\ W' \subseteq K \\ \ell \neq \ell' \lor W \neq W'}} \binom{n-k}{\ell} \binom{n-k}{\ell'} \lambda_{k,\ell} \lambda_{k,\ell'} \cdot (-1)^{2k - |W| - |W'|} \cdot \text{Cov} \left( \hat{I}_{K,\ell}^W, \hat{I}_{K,\ell'}^{W'} \mid m_{K,\ell}^W, m_{K,\ell'}^{W'} \right) \\
    = & \ \sum\limits_{\ell=0}^{n-k} \binom{n-k}{\ell}^2 \lambda_{k,\ell}^2 \sum\limits_{W \subseteq K} \mathbb{V} \left[ \hat{I}_{K,\ell}^W \mid m_{K,\ell}^W \right] \\
    = & \ \sum\limits_{W \subseteq K} \sum\limits_{\ell \in \mathcal{L}_k^{|W|}} \binom{n-k}{\ell}^2 \lambda_{k,\ell}^2 \mathbb{V} \left[ \hat{I}_{K,\ell}^W \mid m_{K,\ell}^W \right] \\
    = & \ \sum\limits_{W \subseteq K} \sum\limits_{\ell \in \mathcal{L}_k^{|W|}} \binom{n-k}{\ell}^2 \lambda_{k,\ell}^2 \mathbb{V} \left[ \frac{1}{m_{K,\ell}^W}\sum\limits_{m=1}^{m_{K,\ell}^W} \nu(A_{K,\ell,m}^W) \mid m_{K,\ell}^W \right] \\
    = & \ \sum\limits_{W \subseteq K} \sum\limits_{\ell \in \mathcal{L}_k^{|W|}} \binom{n-k}{\ell}^2 \lambda_{k,\ell}^2  \cdot \frac{\sigma_{K,\ell,W}^2}{m_{K,\ell}^W}.
\end{align*}

The strata estimates $\hat{I}_{K,\ell}^W$ and $\hat{I}_{K,\ell'}^{W'}$ are independent for $W \neq W'$ or $\ell = \ell'$ because each sampled coalition $A_b$ can only be used to update one estimate.
Consequently, their covariance is zero.
Finally, the variances of the estimates for the explicitly calculated strata are zero and thus eliminated.
\end{proof}

\textbf{Theorem 4.2.}
\textit{
    For any $K \in \mathcal{N}_k$ the variance of the estimate $\hat{I}_K$ returned by SVARM-IQ is bounded by
    \begin{equation*}
        \mathbb{V} \left[ \hat{I}_K \right] \leq \frac{\gamma_k}{\tilde{B}} \sum\limits_{W \subseteq K} \sum\limits_{\ell \in \mathcal{L}_k^{|W|}} \binom{n-k}{\ell}^2 \lambda_{k,\ell}^2 \sigma_{K,\ell,W}^2.
    \end{equation*}
}

\begin{proof}
We combine the variance of each estimate variance conditioned on the sample numbers given by \cref{lem:VarianceGivenSampleNumbers} with the bound on the expected inverted total sample numbers given by \cref{lem:InvertedSamples}:
\begin{align*}
    \mathbb{V} \left[ \hat{I}_K \right] = & \ \mathbb{E}_{\left( m_{K,\ell}^W \right)_{\ell \in \mathcal{L}, W \subseteq K}} \left[ \mathbb{V} \left[ \hat{I}_K \mid \left( m_{K,\ell}^W \right)_{\ell \in \mathcal{L}, W \subseteq K} \right] \right] \\
    = & \ \mathbb{E}_{\left( m_{K,\ell}^W \right)_{\ell \in \mathcal{L}, W \subseteq K}} \left[ \sum\limits_{W \subseteq K} \sum\limits_{\ell \in \mathcal{L}_k^{|W|}} \binom{n-k}{\ell}^2 \lambda_{k,\ell}^2 \cdot \frac{\sigma_{K,\ell,W}^2}{m_{K,\ell}^W} \right] \\
    = & \ \sum\limits_{W \subseteq K} \sum\limits_{\ell \in \mathcal{L}_k^{|W|}} \binom{n-k}{\ell}^2 \lambda_{k,\ell}^2 \sigma_{K,\ell,W}^2 \cdot \mathbb{E} \left[ \frac{1}{m_{K,\ell}^W} \right] \\
    \leq & \ \frac{\gamma_k}{\tilde{B}} \sum\limits_{W \subseteq K} \sum\limits_{\ell \in \mathcal{L}_k^{|W|}} \binom{n-k}{\ell}^2 \lambda_{k,\ell}^2 \sigma_{K,\ell,W}^2.
\end{align*}
\end{proof}

\textbf{Corollary 4.3.}
\textit{
    For any $K \in \mathcal{N}_k$ the mean squared error of the estimate $\hat{I}_K$ returned by SVARM-IQ is bounded by
    \begin{equation*}
        \mathbb{E} \left[ \left( \hat{I}_K - I_K \right)^2 \right] \leq \frac{\gamma_k}{\tilde{B}} \sum\limits_{W \subseteq K} \sum\limits_{\ell \in \mathcal{L}_k^{|W|}} \binom{n-k}{\ell}^2 \lambda_{k,\ell}^2 \sigma_{K,\ell,W}^2.
    \end{equation*}
}

\begin{proof}
The bias-variance decomposition allows us to decompose the mean squared error into the bias of $\hat{I}_K$ and its variance.
Since we have shown the estimate's unbiasedness in Theorem 4.1
, we can reduce it to its variance bounded in Theorem 4.2
:
\begin{align*}
    \mathbb{E} \left[ \left( \hat{I}_K - I_K \right)^2 \right] = & \ \left( \mathbb{E} \left[ \hat{I}_K - I_K \right] \right)^2 + \mathbb{V} \left[ \hat{I}_K \right] \\
    = & \ \mathbb{V} \left[ \hat{I}_K \right] \\
    \leq & \ \frac{\gamma_k}{\tilde{B}} \sum\limits_{W \subseteq K} \sum\limits_{\ell \in \mathcal{L}_k^{|W|}} \binom{n-k}{\ell}^2 \lambda_{k,\ell}^2 \sigma_{K,\ell,W}^2.
\end{align*}
\end{proof}

\subsection{Threshold Exceedence Probability} \label{app:epsilon_delta}
\textbf{Corollary 4.4.}
\textit{
    For any $K \in \mathcal{N}_k$ and fixed $\varepsilon > 0$ the absolute error of the estimate $\hat{I}_K$ returned by SVARM-IQ exceeds $\varepsilon$ with a probability of at most
    \begin{equation*}
        \mathbb{P} \left( | \hat{I}_K - I_K | \geq \varepsilon \right) \leq \frac{\gamma_k}{\varepsilon^2 \tilde{B}} \sum\limits_{W \subseteq K} \sum\limits_{\ell \in \mathcal{L}_k^{|W|}} \binom{n-k}{\ell}^2 \lambda_{k,\ell}^2 \sigma_{K,\ell,W}^2. 
    \end{equation*}
}

\begin{proof}
We apply Chebychev's inequality and make use of the variance bound in Theorem 4.2: 
\begin{equation*}
    \mathbb{P} \left( | \hat{I}_K - I_K | \geq \varepsilon \right) \leq \frac{\mathbb{V} \left[\hat{I}_k\right]}{\varepsilon^2} \leq \frac{\gamma_k}{\varepsilon^2 \bar{B}_k} \sum\limits_{W \subseteq K} \sum\limits_{\ell \in \mathcal{L}_k^{|W|}} \binom{n-k}{\ell}^2 \lambda_{k,\ell}^2 \sigma_{K,\ell,W}^2.
\end{equation*}
\end{proof}

\begin{lemma} \label{lem:HoeffdingStrataGivenSamples}
    For the stratum estimate $\hat{I}_{K,\ell}^W$ of any $K \in \mathcal{N}_k$ with $W \subseteq K$, $\ell \in \mathcal{L}_k^{|W|}$, and some fixed $\varepsilon > 0$ holds
    \begin{equation*}
        \mathbb{P} \left( | \hat{I}_{K,\ell}^W - I_{K,\ell}^W | \geq \varepsilon \mid m_{K,\ell}^W \right) \leq 2 \exp \left( - \frac{2 m_{K,\ell}^W \varepsilon^2}{r_{K,\ell,W}^2} \right).
    \end{equation*}
\end{lemma}

\begin{proof}
We combine Hoeffding's inequality with the unbiasedness of the strata estimates shown in \cref{lem:UnbiasedStrata} and obtain:
    \begin{align*}
        & \ \mathbb{P} \left( | \hat{I}_{K,\ell}^W - I_{K,\ell}^W | \geq \varepsilon \mid m_{K,\ell}^W \right) \\
        = & \ \mathbb{P} \left( | \hat{I}_{K,\ell}^W - \mathbb{E}[\hat{I}_{K,\ell}^W] | \geq \varepsilon \mid m_{K,\ell}^W \right) \\
        = & \ \mathbb{P} \left( \Bigg| \sum\limits_{m=1}^{m_{K,\ell}^W} \nu(A_{S,\ell,m}^W) - \mathbb{E} \left[ \sum\limits_{m=1}^{m_{K,\ell}^W} \nu(A_{K,\ell,m}^W) \right] \Bigg| \geq m_{K,\ell}^W \varepsilon \mid m_{K,\ell}^W  \right) \\
        \leq & \ 2 \exp \left( - \frac{2 m_{K,\ell}^W \varepsilon^2}{r_{K,\ell,W}^2} \right) .
    \end{align*}
\end{proof}

\begin{lemma} \label{lem:HoeffdingStrata}
    For any $K \in \mathcal{N}_k$ with $W \subseteq K$, $\ell \in \mathcal{L}_k^{|W|}$ and some fixed $\varepsilon > 0$ holds
    \begin{equation*}
        \mathbb{P} \left( | \hat{I}_{K,\ell}^W - I_{K,\ell}^W | \geq \varepsilon \right) \leq \exp \left( - \frac{\tilde{B}}{2\gamma_k^2} \right) + 2 \frac{\exp \left( - \frac{2 \varepsilon^2}{r_{K,\ell,W}^2} \right)^{\left\lfloor \frac{\tilde{B}}{2\gamma_k} \right\rfloor}}{\exp \left( \frac{2 \varepsilon^2}{r_{K,\ell,W}^2} \right) - 1}.
    \end{equation*}
\end{lemma}

\begin{proof}
We start by deriving with Hoeffding's inequality and \cref{lem:ExpectedSamples} a bound on the probability that $\bar{m}_{K,\ell}^W$ falls below $\frac{\tilde{B}}{2\gamma_k}$:
\begin{align*}
    & \ \mathbb{P} \left( \bar{m}_{K,\ell}^W \leq \frac{\tilde{B}}{2\gamma_k} \right) \\
    = & \ \mathbb{P} \left( \mathbb{E} \left[ \bar{m}_{K,\ell}^W \right] - \bar{m}_{K,\ell}^W \geq \mathbb{E} \left[ \bar{m}_{K,\ell}^W \right] - \frac{\tilde{B}}{2\gamma_k} \right) \\
    \leq & \ \exp \left( - \frac{2 \left( \mathbb{E} \left[ \bar{m}_{K,\ell}^W \right] - \frac{\tilde{B}}{2\gamma_k} \right)^2}{\tilde{B}} \right) \\
    \leq & \ \exp \left( - \frac{\tilde{B}}{2\gamma_k^2} \right) .
\end{align*}

Further, we show with \cref{lem:HoeffdingStrataGivenSamples} another statement:
\begin{align*}
    & \ \sum\limits_{m = \left\lfloor \frac{\tilde{B}}{2\gamma_k} \right\rfloor +1}^{\tilde{B}+1} \mathbb{P} \left( | \hat{I}_{K,\ell}^W - I_{K,\ell}^W | \geq \varepsilon \mid m_{K,\ell}^W = m \right) \\
    \leq & \ 2 \sum\limits_{m = \left\lfloor \frac{\tilde{B}}{2\gamma_k} \right\rfloor + 1}^{\tilde{B}+1} \exp \left( - \frac{2 m \varepsilon^2}{r_{K,\ell,W}^2} \right) \\
    = & \ 2 \sum\limits_{m=0}^{\tilde{B}+1} \exp \left( - \frac{2 \varepsilon^2}{r_{K,\ell,W}^2} \right)^m - 2 \sum\limits_{m=0}^{\left\lfloor \frac{\tilde{B}}{2\gamma_k}\right\rfloor} \exp \left( - \frac{2 \varepsilon^2}{r_{K,\ell,W}^2} \right)^m \\
    = & \ 2 \frac{\exp \left( - \frac{2 \varepsilon^2}{r_{K,\ell,W}^2} \right)^{\left\lfloor \frac{\tilde{B}}{2\gamma_k} \right\rfloor} - \exp \left( - \frac{2 \varepsilon^2}{r_{K,\ell,W}^2} \right)^{\tilde{B}+1}}{\exp \left( \frac{2 \varepsilon^2}{r_{K,\ell,W}^2} \right) - 1} \\
    \leq & \ 2 \frac{\exp \left( - \frac{2 \varepsilon^2}{r_{K,\ell,W}^2} \right)^{\left\lfloor \frac{\tilde{B}}{2\gamma_k} \right\rfloor}}{\exp \left( \frac{2 \varepsilon^2}{r_{K,\ell,W}^2} \right) - 1} .
\end{align*}

\newpage

Finally, we combine both intermediate results and obtain:

\begin{align*}
    & \ \mathbb{P} \left( | \hat{I}_{K,\ell}^W - I_{K,\ell}^W | \geq \varepsilon \right) \\
    \leq & \ \sum\limits_{m=1}^{\tilde{B}+1} \mathbb{P} \left( | \hat{I}_{K,\ell}^W - I_{K,\ell}^W | \geq \varepsilon \mid m_{K,\ell}^W = m \right) \cdot \mathbb{P} \left( m_{K,\ell}^W = m \right) \\
    = & \ \sum\limits_{m=1}^{\left\lfloor \frac{\tilde{B}}{2\gamma_k} \right\rfloor} \mathbb{P} \left( | \hat{I}_{K,\ell}^W - I_{K,\ell}^W | \geq \varepsilon \mid m_{K,\ell}^W = m \right) \cdot \mathbb{P} \left( m_{K,\ell}^W = m \right) \\
    & \ + \sum\limits_{m=\left\lfloor \frac{\tilde{B}}{2\gamma_k} \right\rfloor+1}^{\bar{B}_k+1} \mathbb{P} \left( | \hat{I}_{K,\ell}^W - I_{K,\ell}^W | \geq \varepsilon \mid m_{K,\ell}^W = m \right) \cdot \mathbb{P} \left( m_{K,\ell}^W = m \right) \\
    \leq & \ \mathbb{P} \left( m_{K,\ell}^W \leq \left\lfloor \frac{\tilde{B}}{2\gamma_k} \right\rfloor \right) + \sum\limits_{m=\left\lfloor \frac{\tilde{B}}{2\gamma_k} \right\rfloor+1}^{\tilde{B}+1} \mathbb{P} \left( | \hat{I}_{K,\ell}^W - I_{S,\ell}^W | \geq \varepsilon \mid m_{K,\ell}^W = m \right)  \\
    \leq & \ \exp \left( - \frac{\tilde{B}}{2\gamma_k^2} \right) + 2 \frac{\exp \left( - \frac{2 \varepsilon^2}{r_{K,\ell,W}^2} \right)^{\left\lfloor \frac{\tilde{B}}{2\gamma_k} \right\rfloor}}{\exp \left( \frac{2 \varepsilon^2}{r_{K,\ell,W}^2} \right) - 1} .
\end{align*}
\end{proof}

\textbf{Theorem 4.5.}
\textit{
    For any $K \in \mathcal{N}_k$ and fixed $\varepsilon > 0$ the absolute error of the estimate $\hat{I}_K$ exceeds $\varepsilon$ with probability of at most
    \begin{equation*}
        \mathbb{P} \left( | \hat{I}_K - I_K | \geq \varepsilon \right) \leq \sum\limits_{W \subseteq K} \sum\limits_{\ell \in \mathcal{L}_k^{|W|}} \exp \left( - \frac{\tilde{B}}{2\gamma_k^2} \right) + 2 \frac{\exp \left( - \frac{2 \varepsilon^2}{\binom{n-k}{\ell}^2 \lambda_{k,\ell}^2 R_K^2} \right)^{\left\lfloor \frac{\tilde{B}}{2\gamma_k} \right\rfloor}}{\exp \left( \frac{2 \varepsilon^2}{\binom{n-k}{\ell}^2 \lambda_{k,\ell}^2 R_K^2} \right) - 1}.
    \end{equation*}
}

\newpage

\begin{proof}
We derive the result by applying \cref{lem:HoeffdingStrata} and utilizing the fact that for all explicitly computed strata $I_{K,\ell}^W \in \mathcal{I}_{\text{exp}}$ holds $\hat{I}_{K,\ell}^W = I_{K,\ell}^W$:
\begin{align*}
    & \ \mathbb{P} \left( | \hat{I}_K - I_K | \geq \varepsilon \right) \\
    = & \ \mathbb{P} \left( \Bigg| \sum\limits_{\ell=0}^{n-k} \sum\limits_{W \subseteq K} \binom{n-k}{\ell} \lambda_{k,\ell} (-1)^{k-|W|} \left( \hat{I}_{K,\ell}^W - I_{K,\ell}^W \right) \Bigg| \geq \varepsilon \right) \\
    \leq & \ \mathbb{P} \left( \sum\limits_{\ell=0}^{n-k} \sum\limits_{W \subseteq K} \binom{n-k}{\ell} \lambda_{k,\ell} \Bigg| \hat{I}_{K,\ell}^W - I_{K,\ell}^W \Bigg| \geq \varepsilon \right) \\
    = & \ \mathbb{P} \left( \sum\limits_{W \subseteq K} \sum\limits_{\ell \in \mathcal{L}_k^{|W|}} \binom{n-k}{\ell} \lambda_{k,\ell} \Bigg| \hat{I}_{K,\ell}^W - I_{K,\ell}^W \Bigg| \geq \varepsilon\right) \\
    \leq & \ \sum\limits_{W \subseteq K} \sum\limits_{\ell \in \mathcal{L}_k^{|W|}} \mathbb{P} \left( \binom{n-k}{\ell} \lambda_{k,\ell} \Bigg| \hat{I}_{K,\ell}^W - I_{K,\ell}^W\Bigg| \geq \frac{\varepsilon r_{K,\ell,W}}{R_K} \right) \\
    = & \ \sum\limits_{W \subseteq K} \sum\limits_{\ell \in \mathcal{L}_k^{|W|}} \mathbb{P} \left( \Bigg| \hat{I}_{K,\ell}^W - I_{K,\ell}^W\Bigg| \geq \frac{\varepsilon r_{K,\ell,W}}{\binom{n-k}{\ell} \lambda_{k,\ell} R_K} \right) \\
    \leq & \ \sum\limits_{W \subseteq K} \sum\limits_{\ell \in \mathcal{L}_k^{|W|}} \exp \left( - \frac{\tilde{B}}{2\gamma_k^2} \right) + 2 \frac{\exp \left( - \frac{2 \varepsilon^2}{\binom{n-k}{\ell}^2 \lambda_{k,\ell}^2 R_K^2} \right)^{\left\lfloor \frac{\tilde{B}}{2\gamma_k} \right\rfloor}}{\exp \left( \frac{2 \varepsilon^2}{\binom{n-k}{\ell}^2 \lambda_{k,\ell}^2 R_K^2} \right) - 1} .
\end{align*}
\end{proof}

\clearpage
\section{DESCRIPTION OF MODELS, DATASETS AND EXPLANATION TASKS}
\label{app:models_datasets}

We briefly sketched the datasets and models on which our cooperative games, used for the experiments, are built.
Hence, we provide further details and sources to allow for reproducibility.
Note that the LM, CNN, and SOUM are akin to \citep{Fumagalli.2023}.

\subsection{Language Model (LM)}
\label{app:models_LM}

We used a pretrained sentiment analysis model for movie reviews.
To be more specific, it s a variant of \texttt{DistilBert}, fine-tuned on the IMDB dataset, and its python version can be found in the \emph{transformers} API \citep{Wolf_Transformers_State-of-the-Art_Natural_2020} at \url{https://huggingface.co/lvwerra/distilbert-imdb}.
The explanation task is to explain the model's sentiment rating between $-1$ and $1$ for randomly selected instances, where positive model outputs indicate positive sentiment.
The features, which are words in this case, are removed on the token level, meaning that tokens of missing values are removed from the input sequence of words, shortening the sentence.
Thus, a coalition within a given sentence is given by the sequence containing only the words associated with each each player of that coalition.
The value function is given by the model's sentiment rating.


\subsection{Vision Transformer (ViT)}
\label{app:models_ViT}

The ViT is, similar to the LM, a transformer model. Unlike the LM, the ViT operates on image patches instead of words.
The python version of the underlying ViT model can be found in the \emph{transformers} API at \url{https://huggingface.co/google/vit-base-patch32-384}.
It originally consists of 144 32x32 pixel image patches, 12 patches for each column and row.
In order to calculate the ground truth values exhaustively via brute force, we cluster smaller input patches together into 3x3 images containing 9 patches in total or into 4x4 images containing 16 patches in total.
Patches of a cluster are jointly turned on and off depending on whether the cluster is part of the coalition or not.
Players, represented by image patches, that are not present in a coalition are removed on the token level and their token is set to the empty token.
The worth of a coalition is the model's predicted class probability for the class which has the highest probability for the grand coalition (the original image with no patches removed) and is therefore within $[0, 1]$.


\subsection{Convolutional Neural Network (CNN)}
\label{app:models_CNN}

The next local explanation scenario is based on a ResNet18\footnote{https://pytorch.org/vision/main/models/generated/torchvision.models.resnet18.html} model \citep{resnet18} trained on ImageNet \citep{ImageNet}.
The task is to explain the predicted class probability for randomly selected images from ImageNet \citep{ImageNet}.
In order to obtain a player set, we use SLIC \citep{achanta2012slic} to merge single pixels to 14 super-pixels.
Each super-pixel corresponds to a player in the resulting cooperative game, and a coalition of players entails the associated super-pixels.
Absent super-pixel players are removed by setting the contained pixels to grey (mean-imputation).
The worth of a coalition is given by the model's predicted class probability, using only the present super-pixels, for the predicted class of the full image with all super-pixels at hand.


\subsection{Sum Of Unanimity Models (SOUM)}
\label{app:models_SOUM}

We further consider synthetic cooperative games, for which the computation of the ground truth values is feasible within polynomial time.
For a given player set $\mathcal{N}$ with $n$ many players, we draw $D=50$ interaction subsets $S_1,\ldots,S_D \subseteq \mathcal{N}$ uniformly at random from the power set of $\mathcal{N}$.
Next, we draw for each interaction subset $S_d$ a coefficient $c_d \in [0,1]$ uniformly at random.
The value function is simply constructed by defining
\begin{equation*}
    \nu(S) = \sum\limits_{d=1}^D c_d \cdot  \llbracket S_d \subseteq S \rrbracket
\end{equation*}
for all coalitions $S \subseteq \mathcal{N}$.
We generate 50 instances of such synthetic games and average the approximation results.
To our advantage, this construction yields a polynomial closed-form solution of the underlying CII values \citep{Fumagalli.2023}, which allows us to use higher player numbers than in real-world explanation scenarios.
For details of the CII computation we refer the interested reader to \citep{Fumagalli.2023}.

\clearpage
\section{FURTHER EMPIRICAL RESULTS}
\label{app:empirical_results}

We conducted more experiments than shown in the main part but had to omit them due to space constraints.
Besides the approximation curves, comparing SVARM-IQ's approximation quality for the SII, STI, and FSI against current baselines measured by the MSE and Prec@10, we present another type of visualization to demonstrate how SVARM-IQ's performance advantage aids in enriching explanations by including interaction effects.

\subsection{Further Results on the Approximation Quality}
\label{app:empirical_results:approx_quality}

This section contains more detailed versions of the figures depicted in the main section.
We compare the approximation quality of SVARM-IQ against baselines for the SII on the LM and ViT in \cref{fig_SII}, for SII, STI, and FSI for CNN in \cref{fig_CNN}, and for SOUM in \cref{fig_SOUM}.

\begin{figure*}[h]
    \centering
    \includegraphics[width=0.32\textwidth]{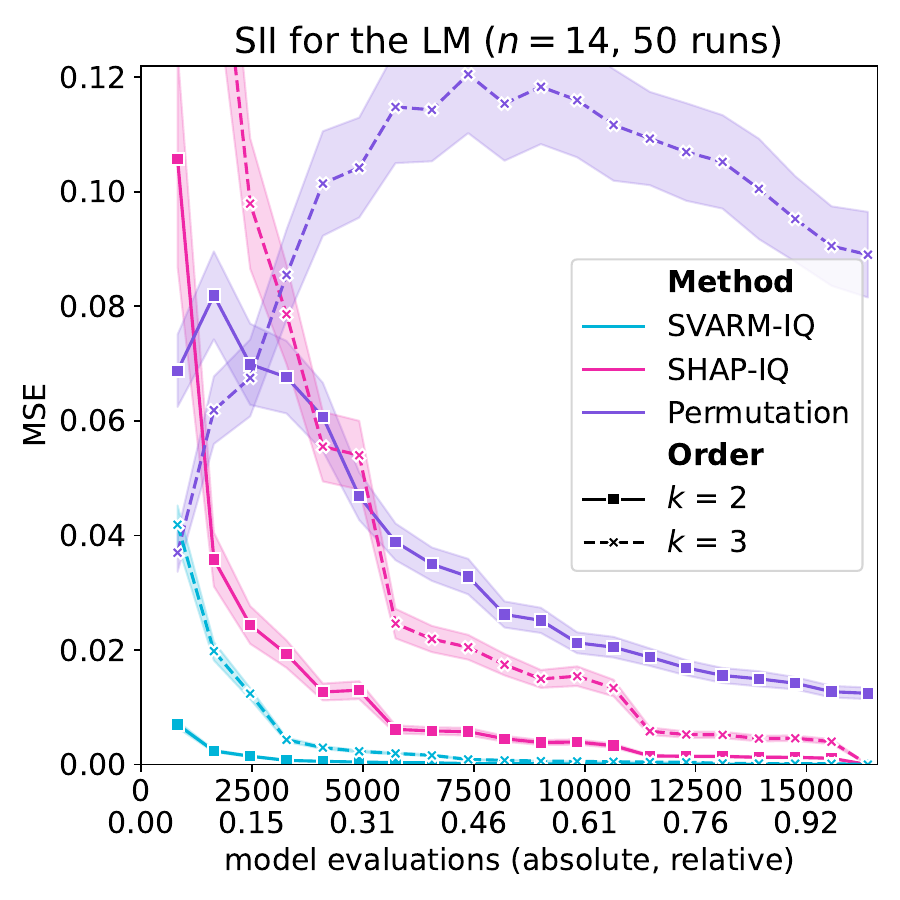}
    \includegraphics[width=0.32\textwidth]{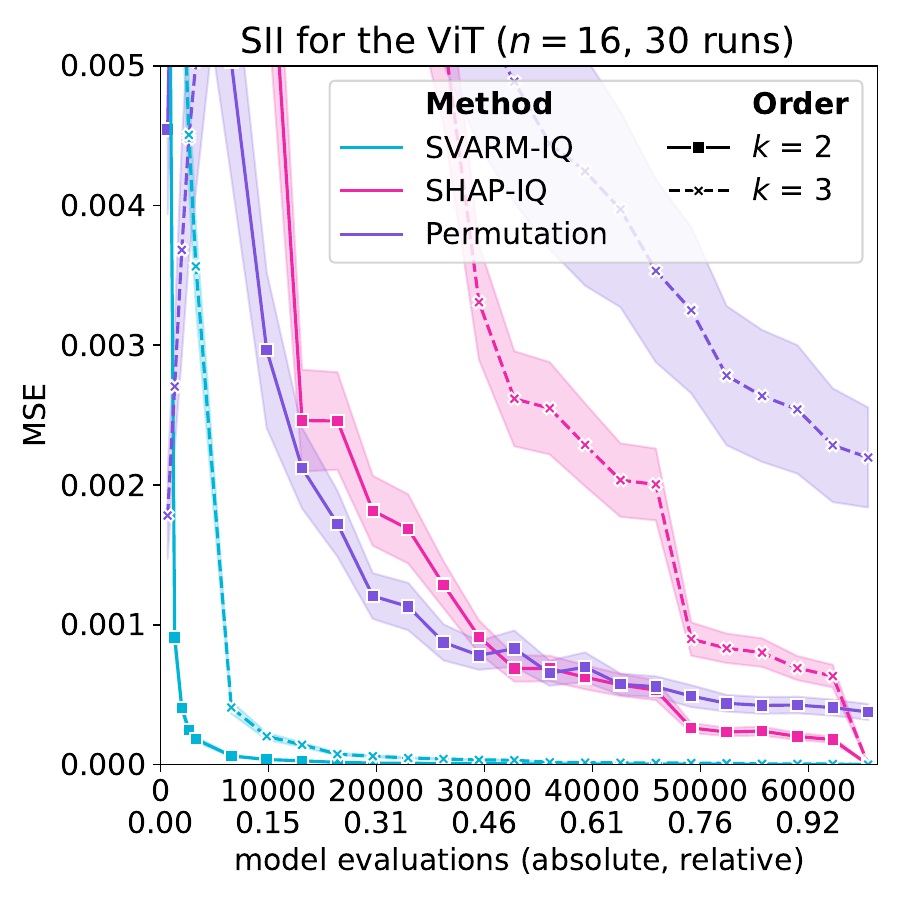}
    \includegraphics[width=0.32\textwidth]{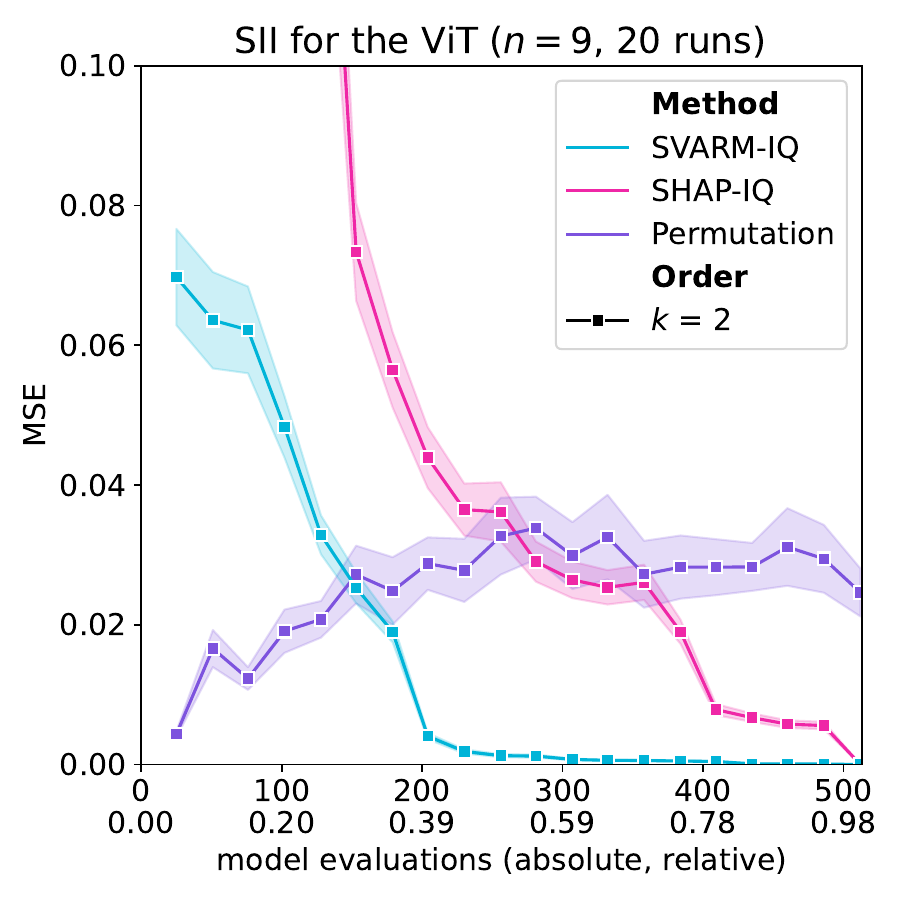} \\
    \includegraphics[width=0.32\textwidth]{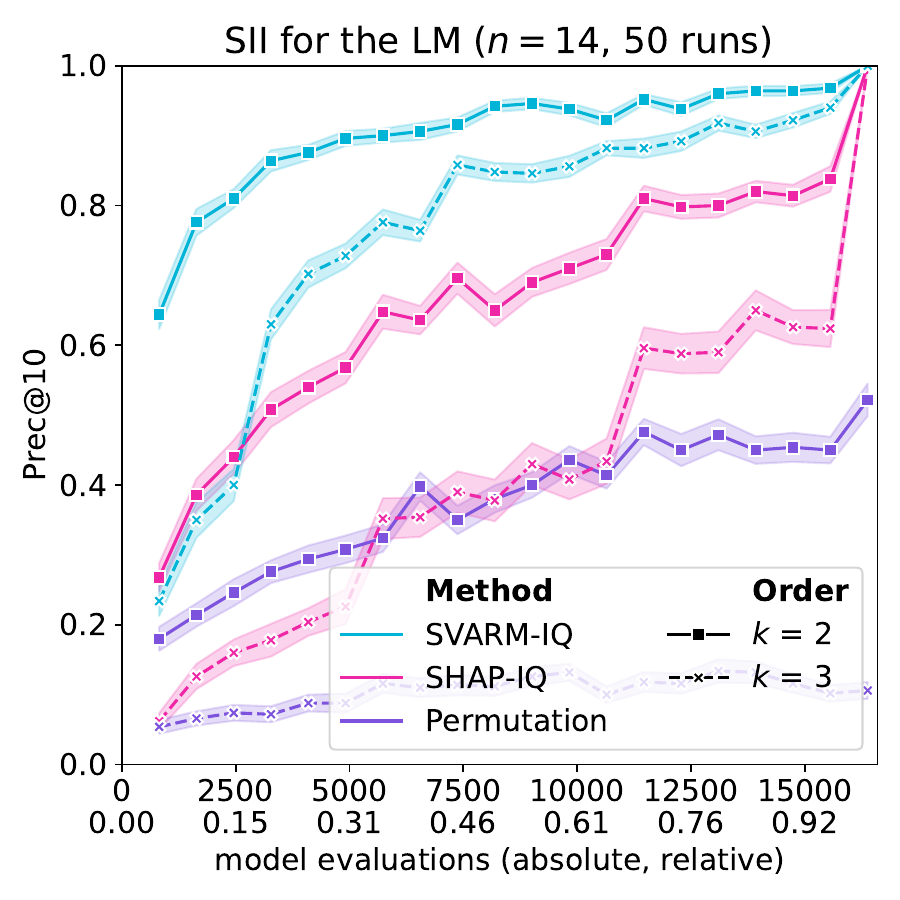}
    \includegraphics[width=0.32\textwidth]{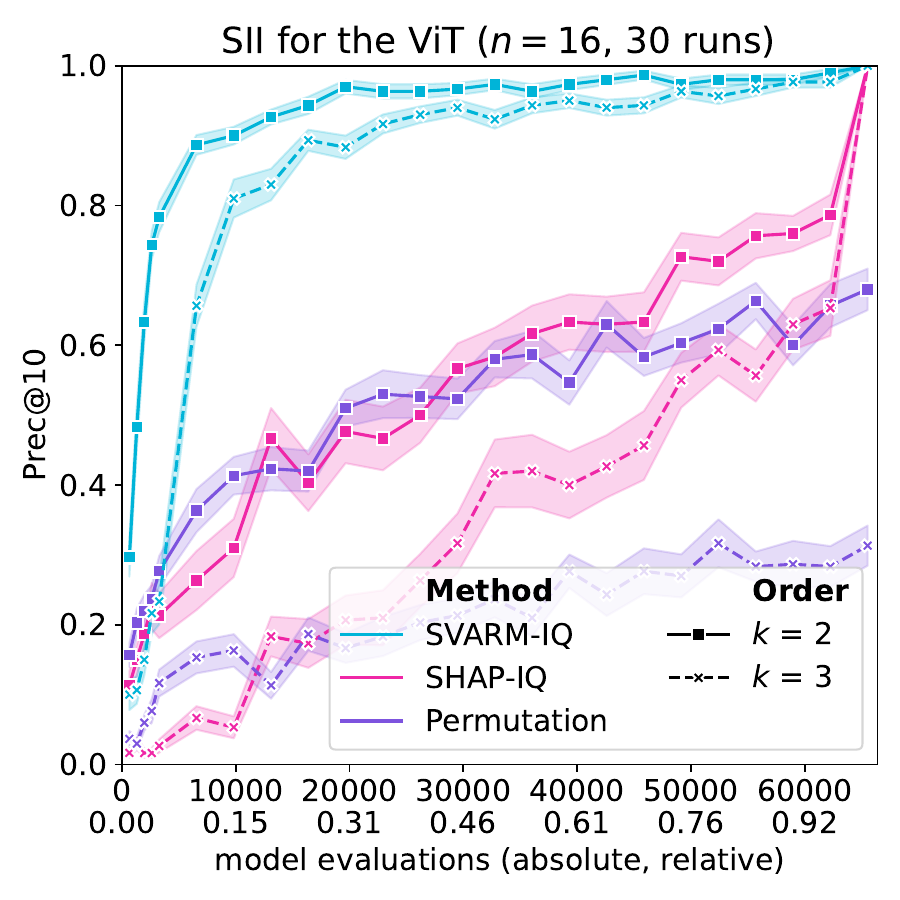}
    \includegraphics[width=0.32\textwidth]{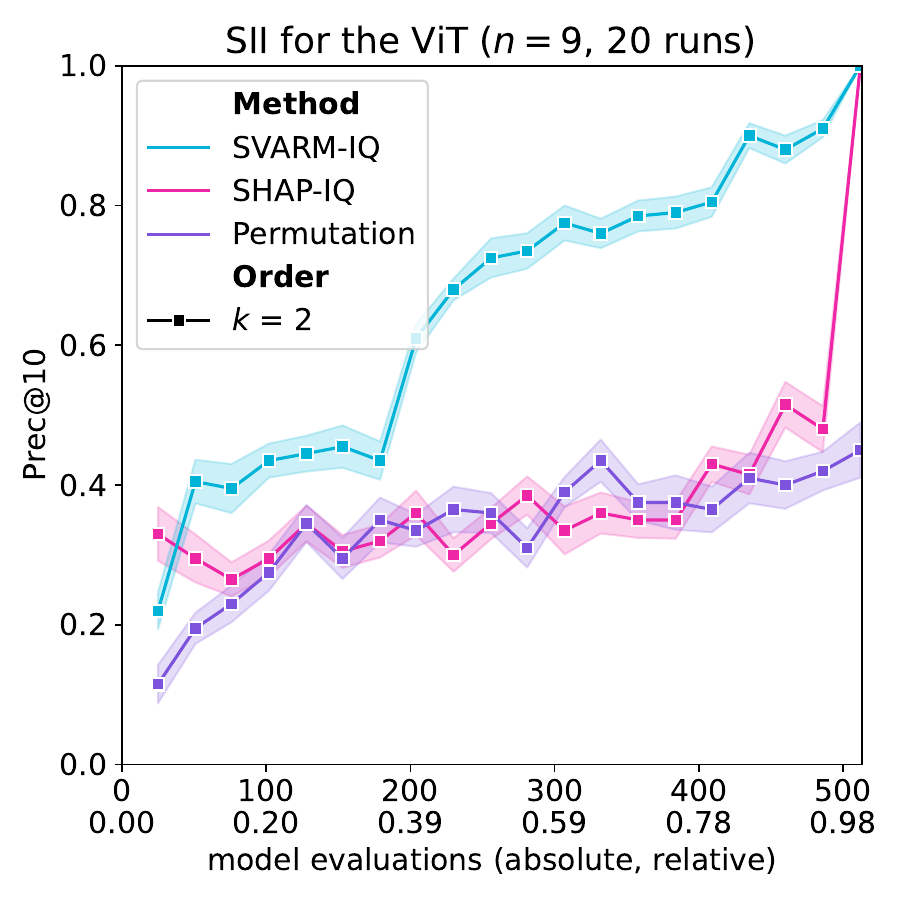}
    \caption{Approximation quality of SVARM-IQ (\textcolor{svarmiq}{blue}) compared to SHAP-IQ (\textcolor{shapiq}{pink}) and permutation sampling (\textcolor{baseline}{purple}) baselines averaged over multiple runs
    for estimating the SII of order $k=2,3$ on the LM (first column, $n=14$, 50 runs) and the ViT (second column, $n=16$, 30 runs; second column, $n=9$, 20 runs). The performance is measured by the MSE (first row) and Prec@10 (second row). The shaded bands represent the standard error over the number of performed runs.}
    \label{fig_SII}
\end{figure*}

\begin{figure*}[h]
    \centering
    \includegraphics[width=0.32\textwidth]{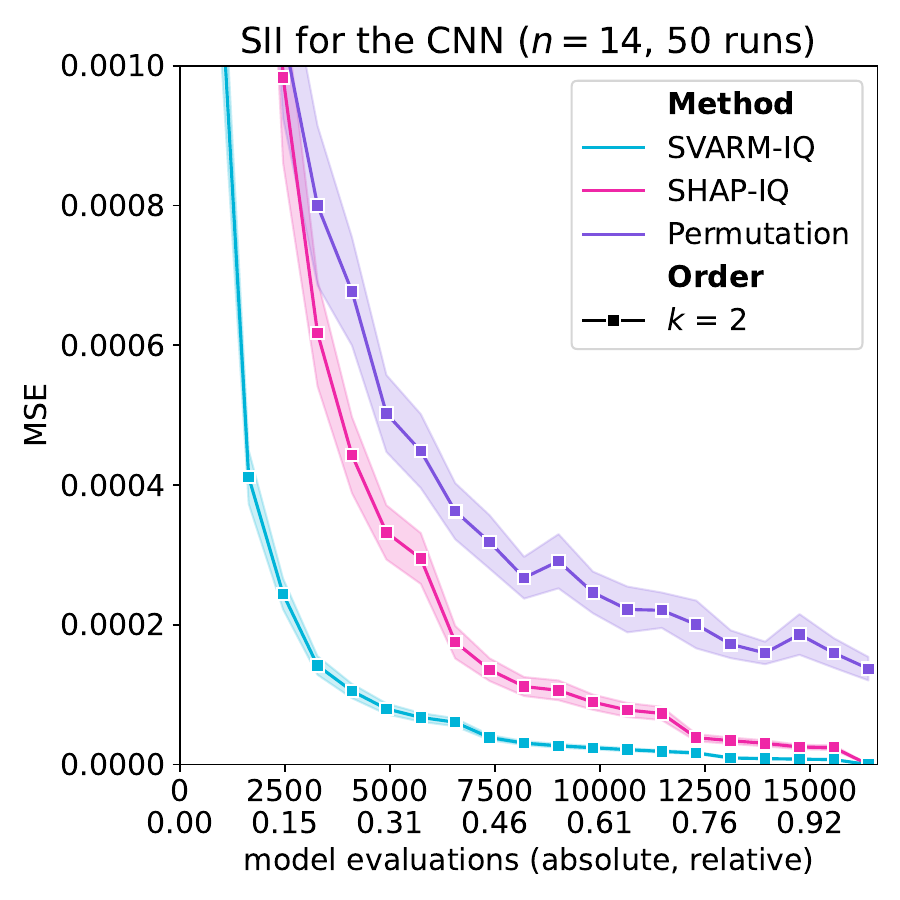}
    \includegraphics[width=0.32\textwidth]{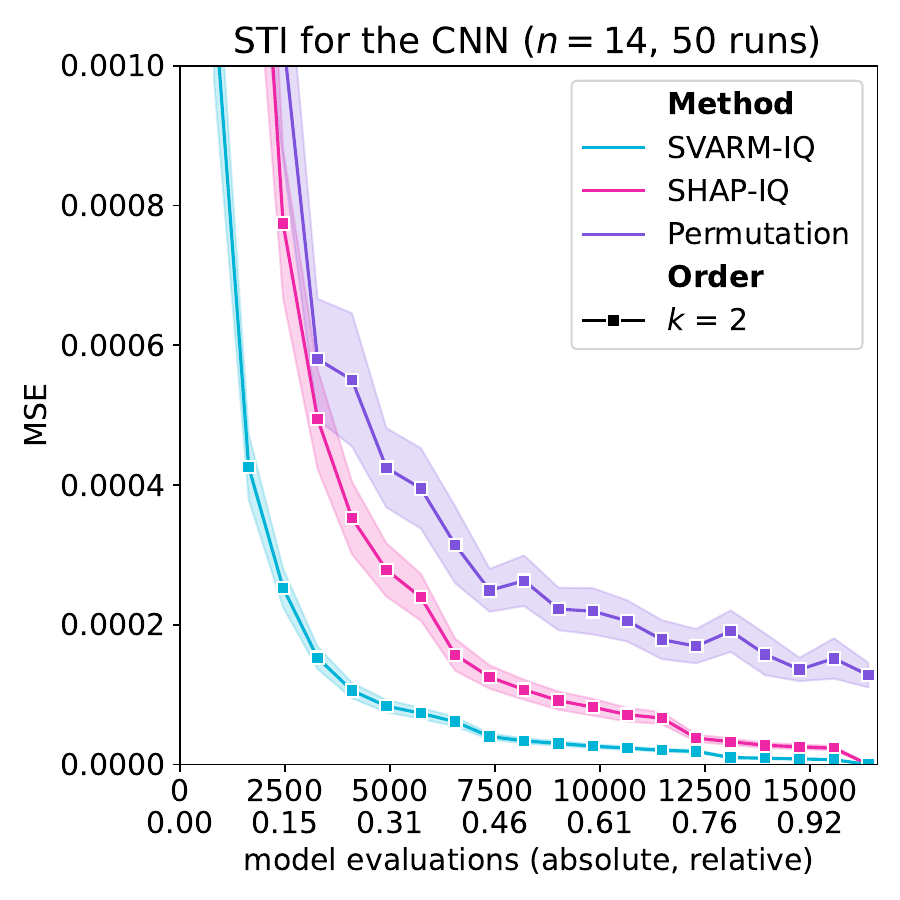}
    \includegraphics[width=0.32\textwidth]{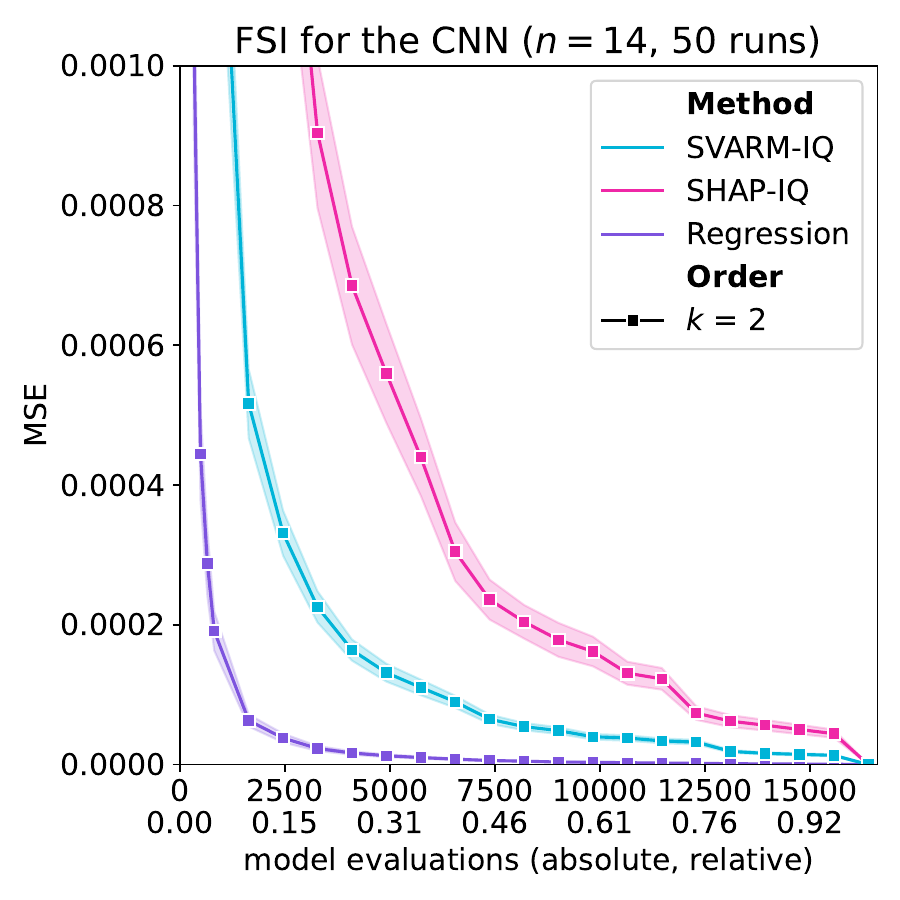} \\
    \includegraphics[width=0.32\textwidth]{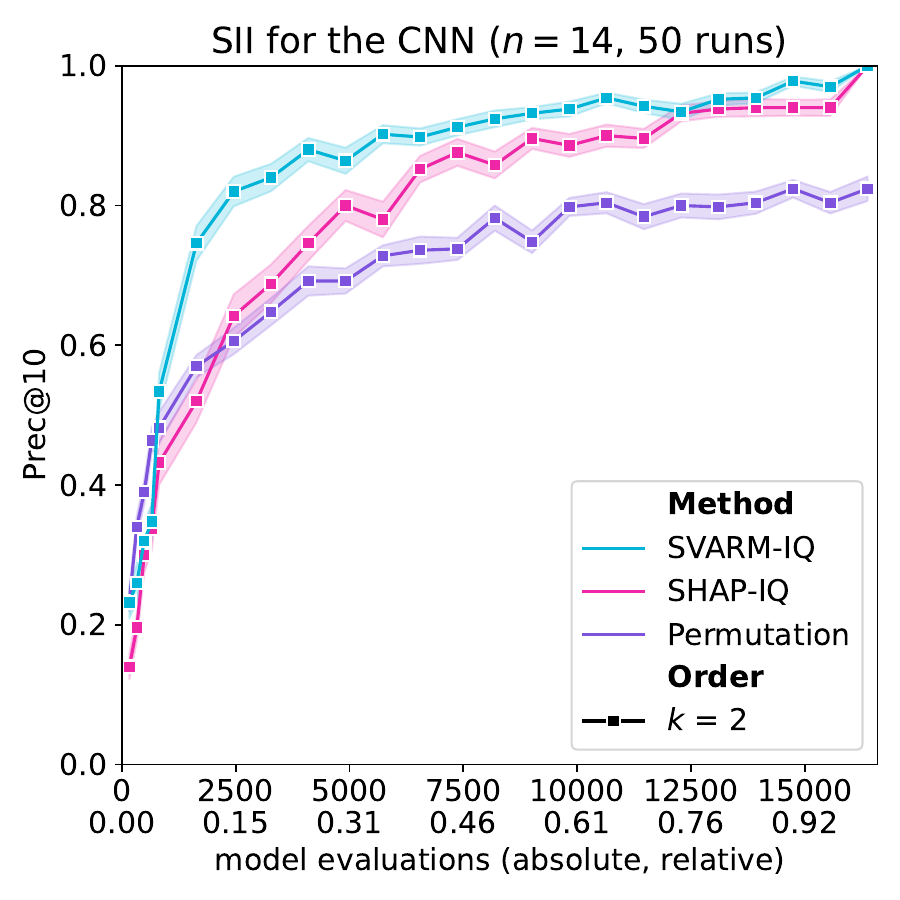}
    \includegraphics[width=0.32\textwidth]{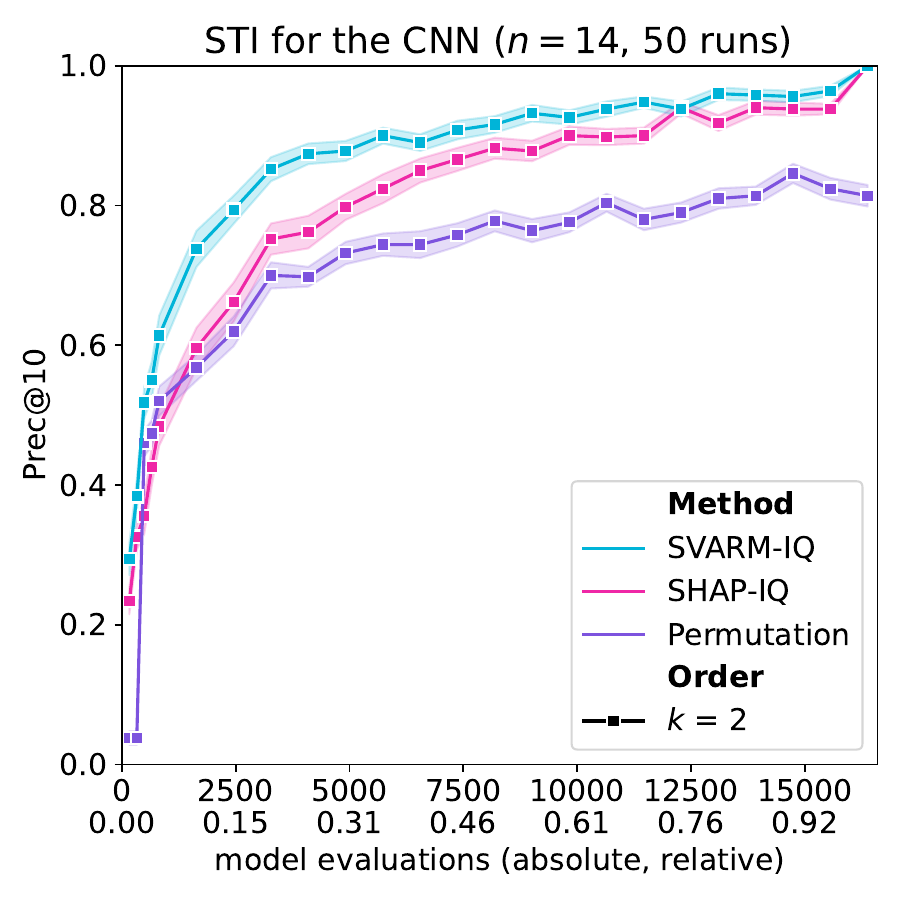}
    \includegraphics[width=0.32\textwidth]{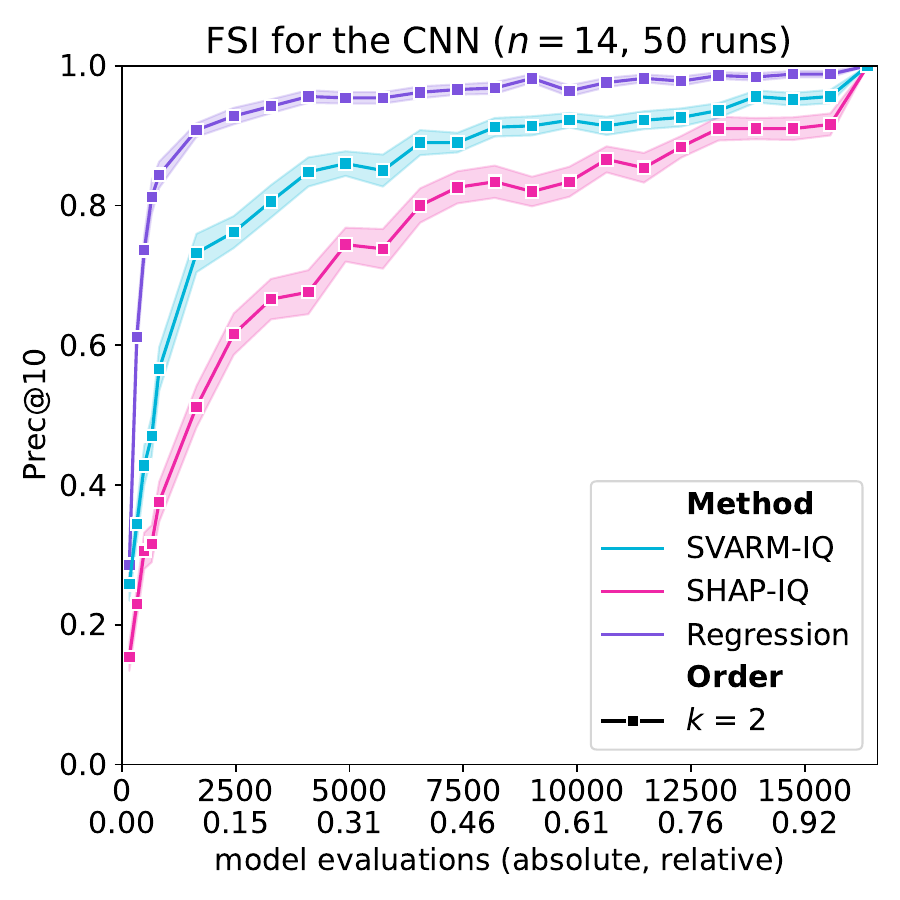}
    \caption{Approximation quality of SVARM-IQ (\textcolor{svarmiq}{blue}) compared to SHAP-IQ (\textcolor{shapiq}{pink}) and permutation sampling (\textcolor{baseline}{purple}) baselines averaged over 50 runs on the CNN for estimating the SII (first column), STI (second column), and FSI (third column) of order $k=2$ for $n=14$. The performance is measured by the MSE (first row) and Prec@10 (second row). The shaded bands represent the standard error over the number of performed runs.}
    \label{fig_CNN}
\end{figure*}

\begin{figure*}[h]
    \centering
    \includegraphics[width=0.32\textwidth]{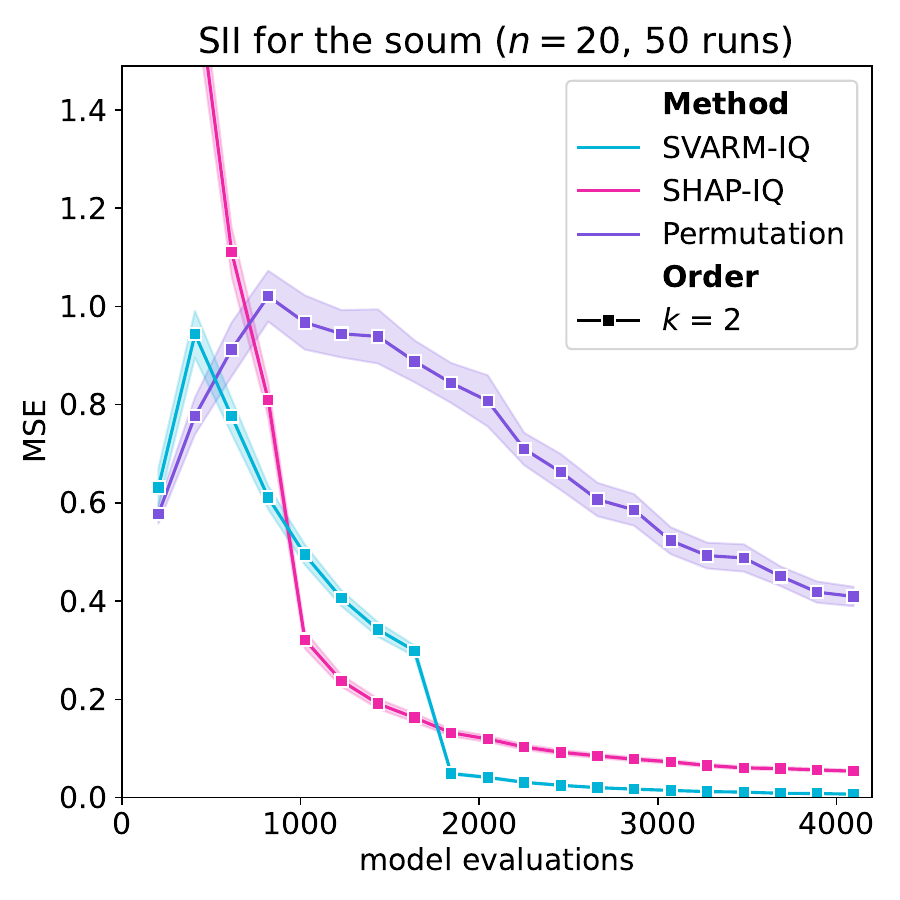}
    \includegraphics[width=0.32\textwidth]{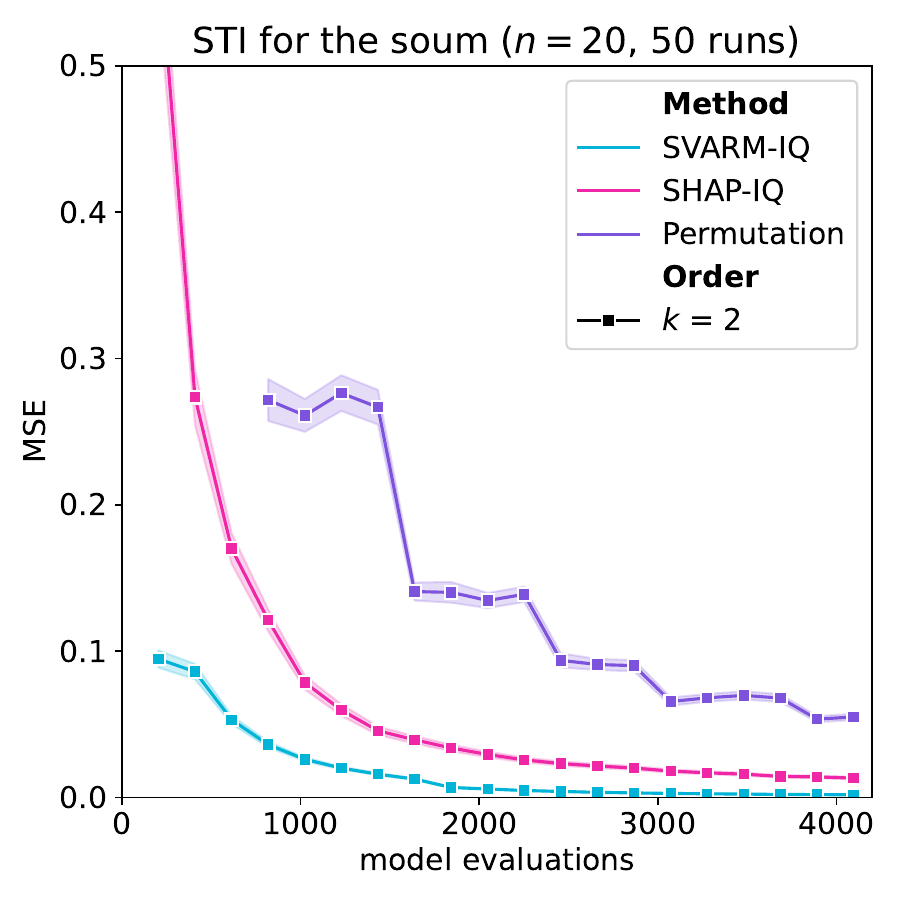}
    \includegraphics[width=0.32\textwidth]{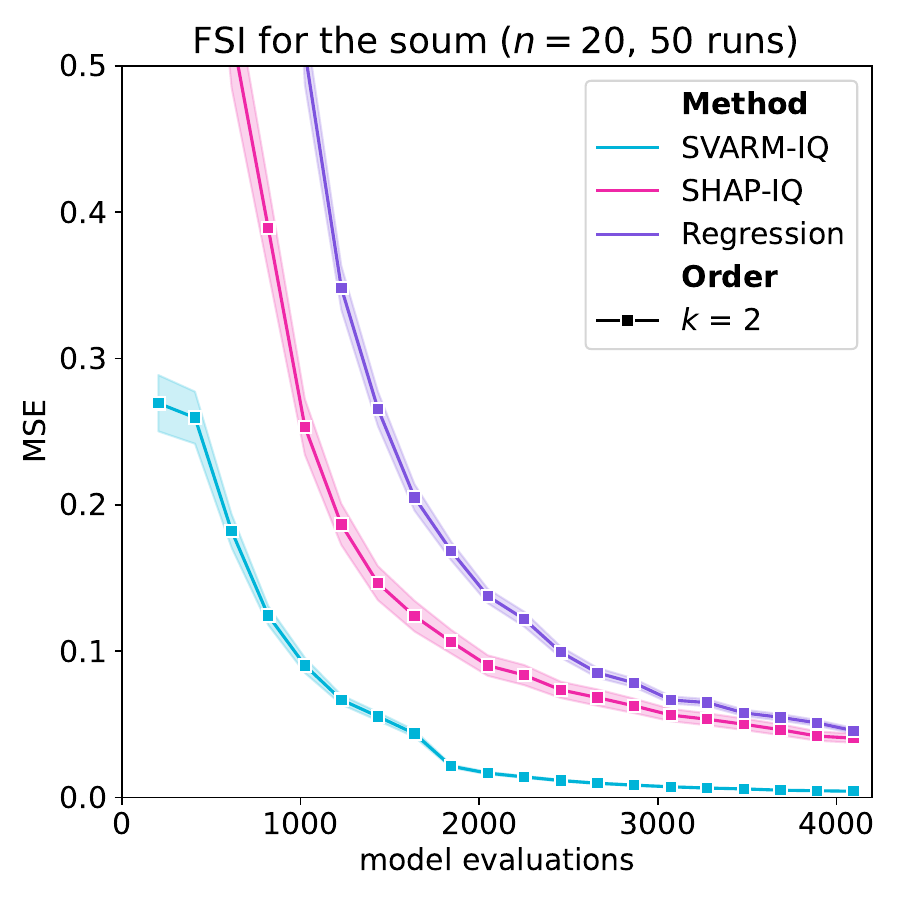} \\
    \includegraphics[width=0.32\textwidth]{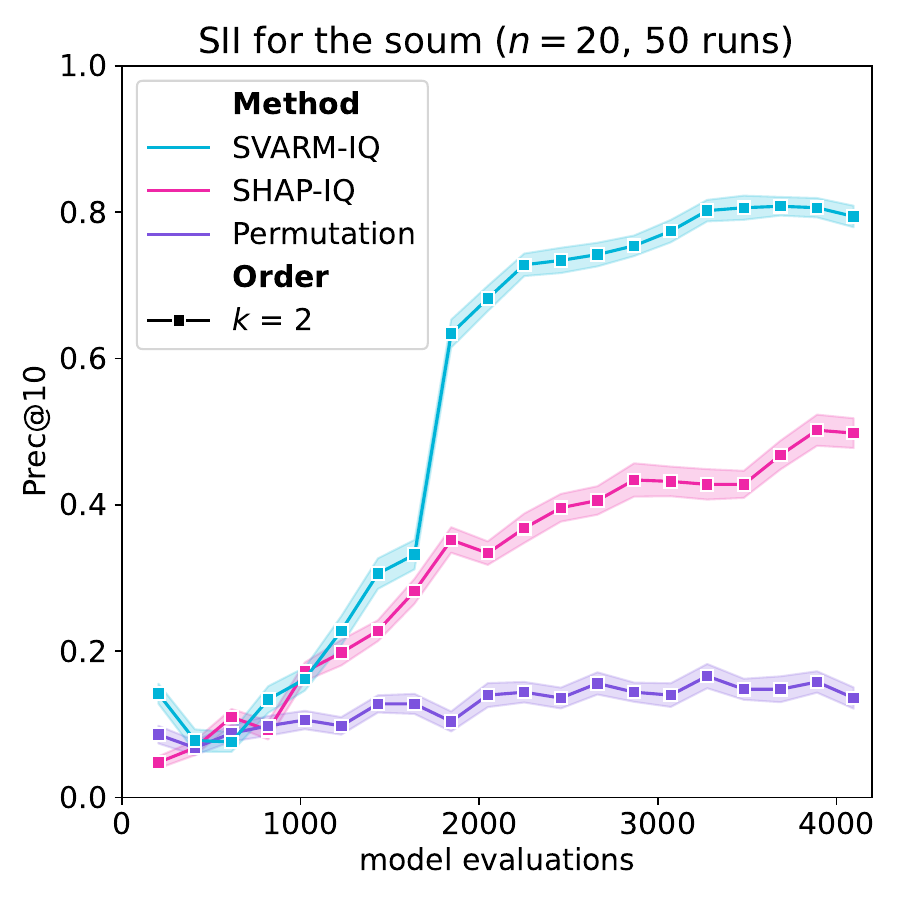}
    \includegraphics[width=0.32\textwidth]{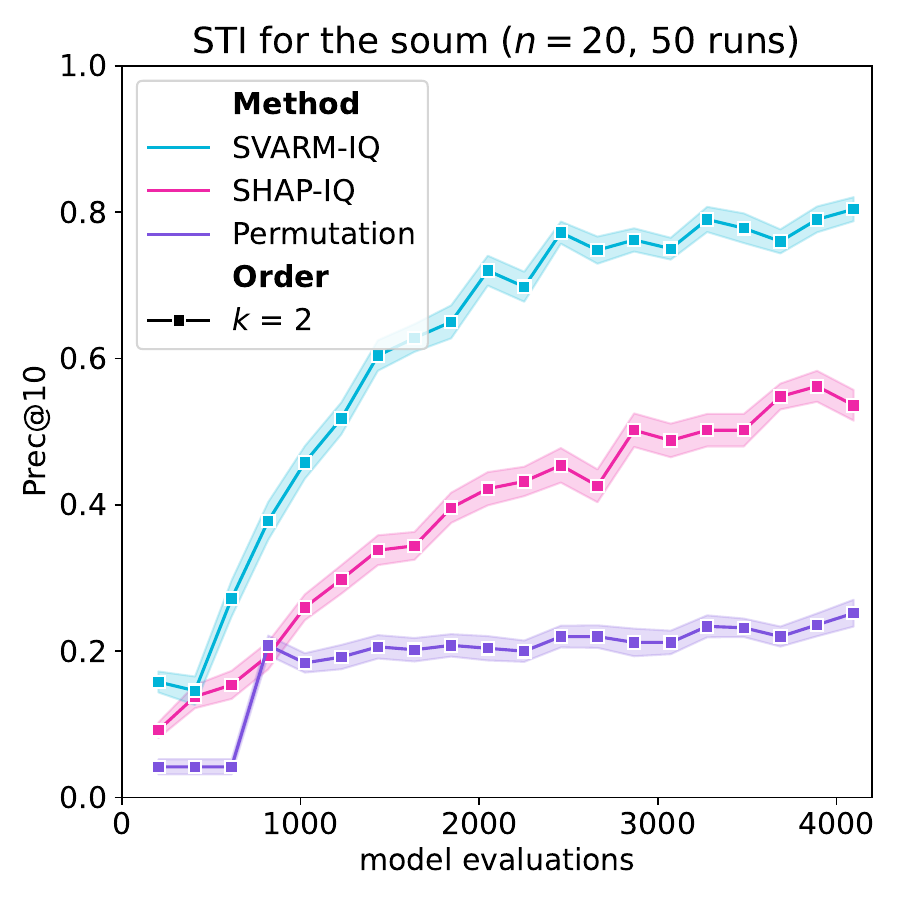}
    \includegraphics[width=0.32\textwidth]{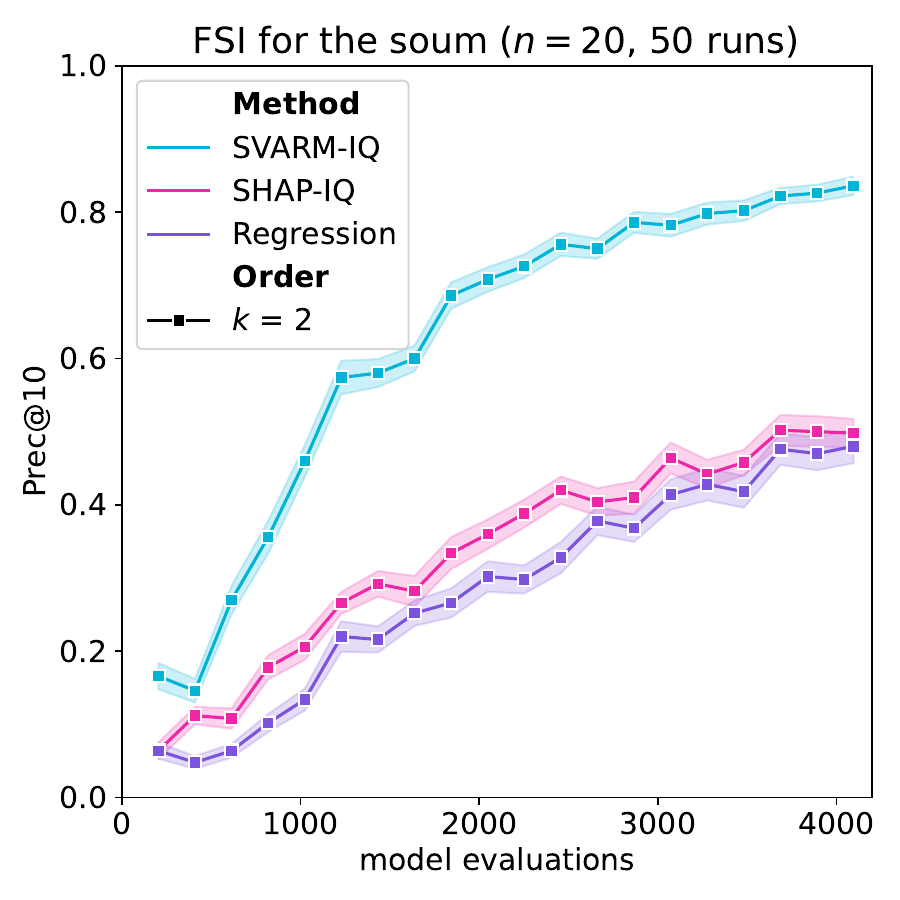}
    \caption{Approximation quality of SVARM-IQ (\textcolor{svarmiq}{blue}) compared to SHAP-IQ (\textcolor{shapiq}{pink}) and permutation sampling (\textcolor{baseline}{purple}) baselines averaged over 50 runs on the SOUM for estimating the SII (first column), STI (second column), and FSI (third column) of order $k=2$ for $n=20$. The performance is measured by the MSE (first row) and Prec@10 (second row). The shaded bands represent the standard error over the number of performed runs.}
    \label{fig_SOUM}
\end{figure*}

\clearpage
\subsection{Further Examples of the Vision Transformer Case Study}
\label{app:empirical_results:vit_examples}

In the following, we demonstrate how the inclusion of interaction besides attributions scores may enrich interpretability and how significantly SVARM-IQ contributes to more reliable explanations due to faster converging interaction estimates.
First, we present in \cref{fig:network_examples} SVARM-IQ's estimates for our ViT scenario, which quantify the importance and interaction of image patches, revealing the insufficiency of sole importance scores and emphasizing the contribution of interaction scores for explaining class predictions for images.
Second, we compare in \cref{fig_network_comp} attribution scores and interaction values estimated by SVARM-IQ and permutation sampling with the ground truth.
Our results showcase that even with a relatively low number of model evaluations SVARM-IQ mirrors the ground truth almost perfectly, while the inaccurate estimates of its competitor pose the visible risk of misleading explanations, thus harming interpretability.

The obtained estimates for the labrador picture in \cref{fig:network_examples} (upper left) allow for a plausible explanation of the model's reasoning.
The most important image patches, those which capture parts of the dogs' heads, share some interesting interaction.
The three patches which contain at least one full eye, might be of high importance, but also exhibit strongly negative pairwise interaction.
This gives us the insight that the addition of such a patch to an existing one contributes on average little to the predicted class probability in comparison to the increase that such a patch causes on its own, plausibly due to redundant information.
In other words, it suffices for the vision transformer to see one patch containing eyes and further patches do not make it much more certain about its predicted class.
On the other side, some patches containing different facial parts show highly positive interaction. For example, the teeth and the pair of eyes complement each other since each of them contains valuable information that is missing in the other patch.
Considering only the importance scores and their ranking would have not led to this interpretation.
Quite the opposite, practitioners would assume most patches to be of equal importance and overlook their insightful interplay.

The comparison of estimates with the ground truth in \cref{fig_network_comp} allows for a twofold conclusion.
The estimates obtained by SVARM-IQ show barely any visible difference to the human eye.
In fact, SVARM-IQ's approximation replicates the ground truth with only a fraction of the number of model evaluations that are necessary for its exact computation.
Hence, it significantly lowers the computational burden for precise explanations.
On the contrary, permutation sampling yields estimated importance and interaction scores which are afflicted with evident imprecision.
Both, the strength and sign of interaction values are estimated with quite severe deviation for the two considered orders.
Hence, the attempt to order the true interactions' strengths or identifying the most influential pairs becomes futile.
This lack in approximation quality has the potential to misguide those who seek for explanations on why the model has predicted a certain class.

\begin{figure*}[h]
    \centering
    \includegraphics[width=0.43\textwidth]{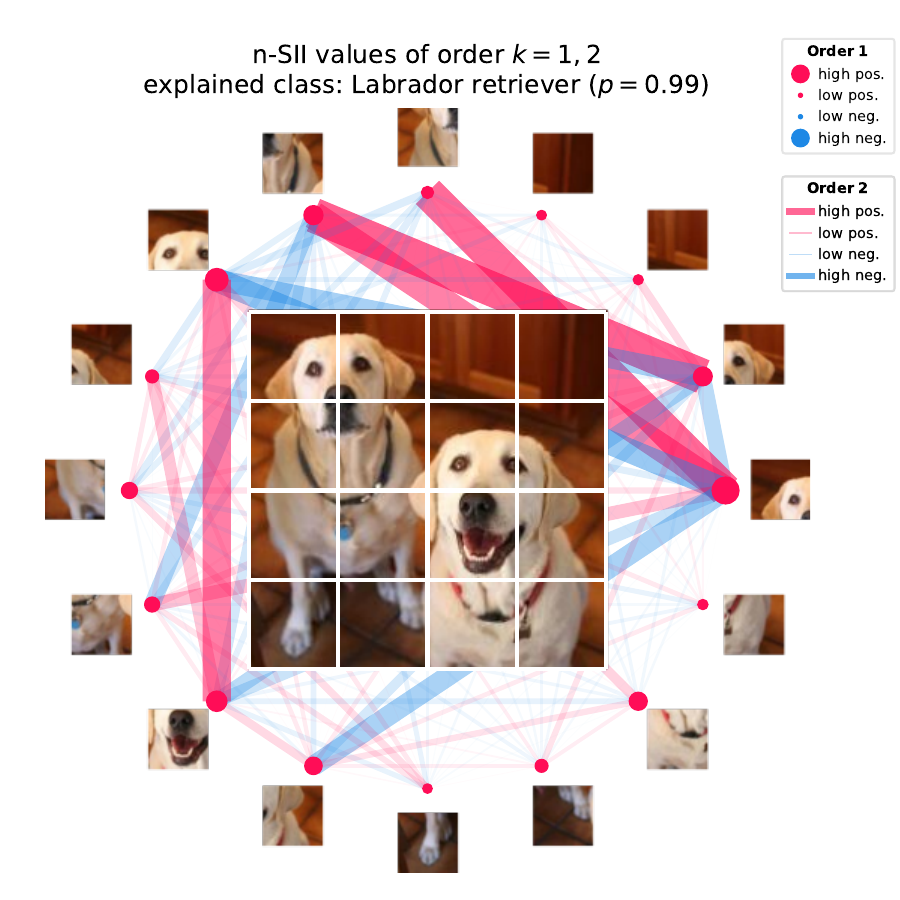}
    \includegraphics[width=0.43\textwidth]{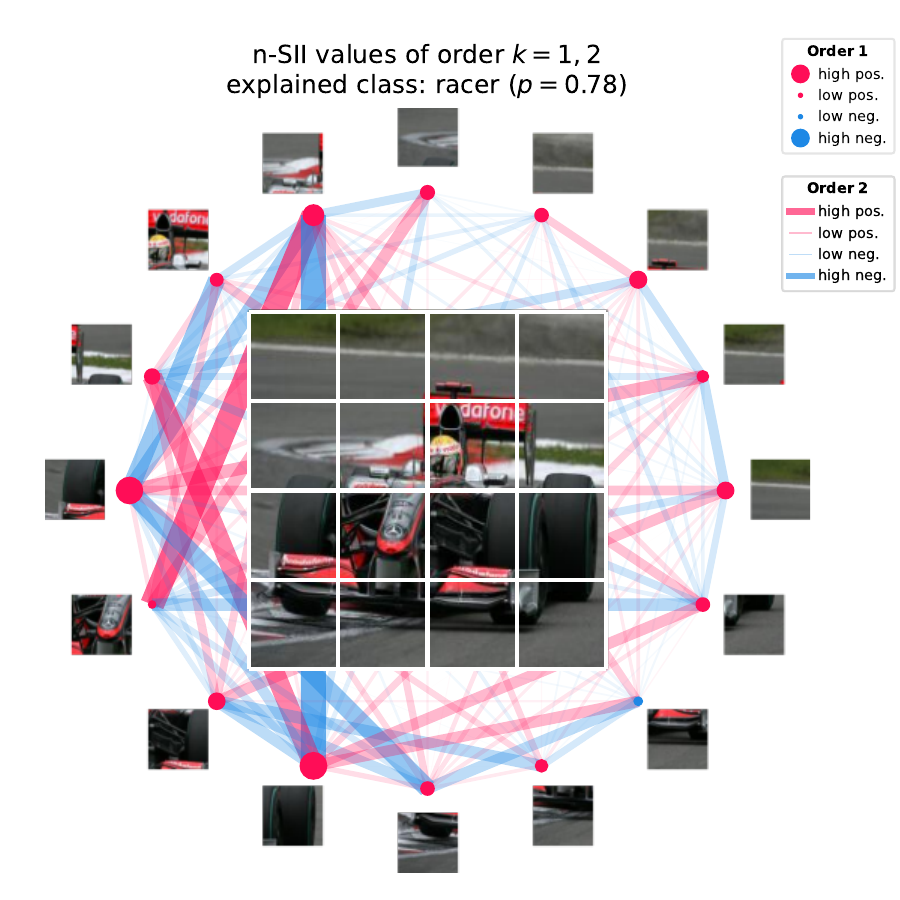}
    \\
    \includegraphics[width=0.43\textwidth]{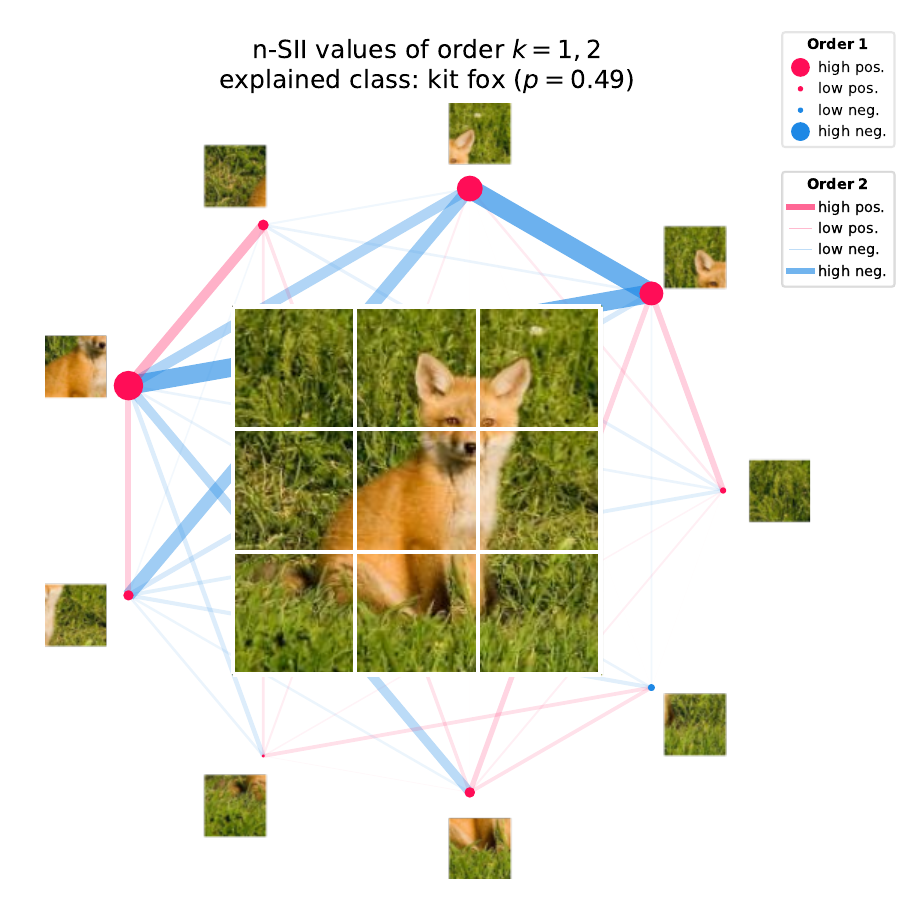}
    \includegraphics[width=0.43\textwidth]{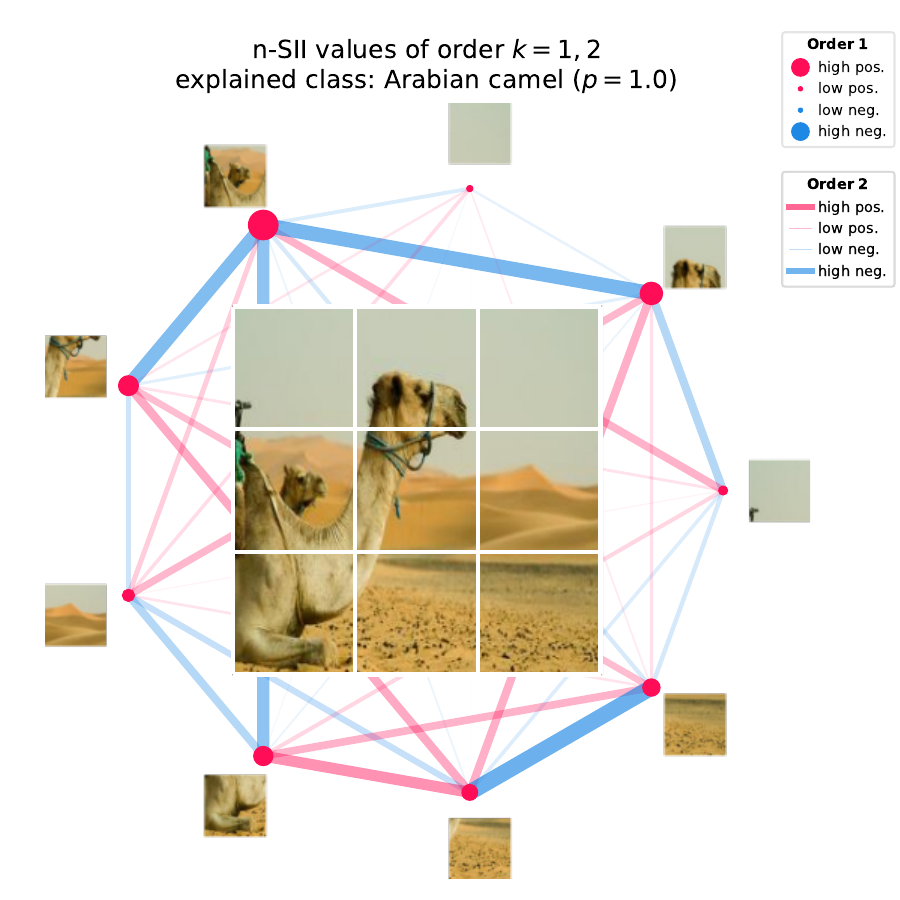}
    \caption{Computed n-SII values of order $k=1,2$ by SVARM-IQ for the predicted class probability of a ViT for selected images taken from ImageNet \citep{ImageNet}. The images are sliced into grids of multiple patches, $n=16$ in the first row and $n=9$ in the second row. The estimates are obtained after single computation runs given a budget of 10000 evaluations for $n=16$ patches and 512 (GTV) for $n=9$ patches.}
    \label{fig:network_examples}
\end{figure*}

\begin{figure*}[h]
    \centering
    \begin{minipage}[c]{0.32\textwidth}
    \centering
    \textbf{Ground Truth}\\model evaluations: $65\,536$
    \end{minipage}
    \begin{minipage}[c]{0.32\textwidth}
    \centering
    \textbf{SVARM-IQ}\\model evaluations: $5\,000$ (7.6\%)
    \end{minipage}
    \begin{minipage}[c]{0.32\textwidth}
    \centering
    \textbf{Permutation}\\model evaluations: $5\,000$ (7.6\%)
    \end{minipage}
    \\[1em]
    \includegraphics[width=0.32\textwidth]{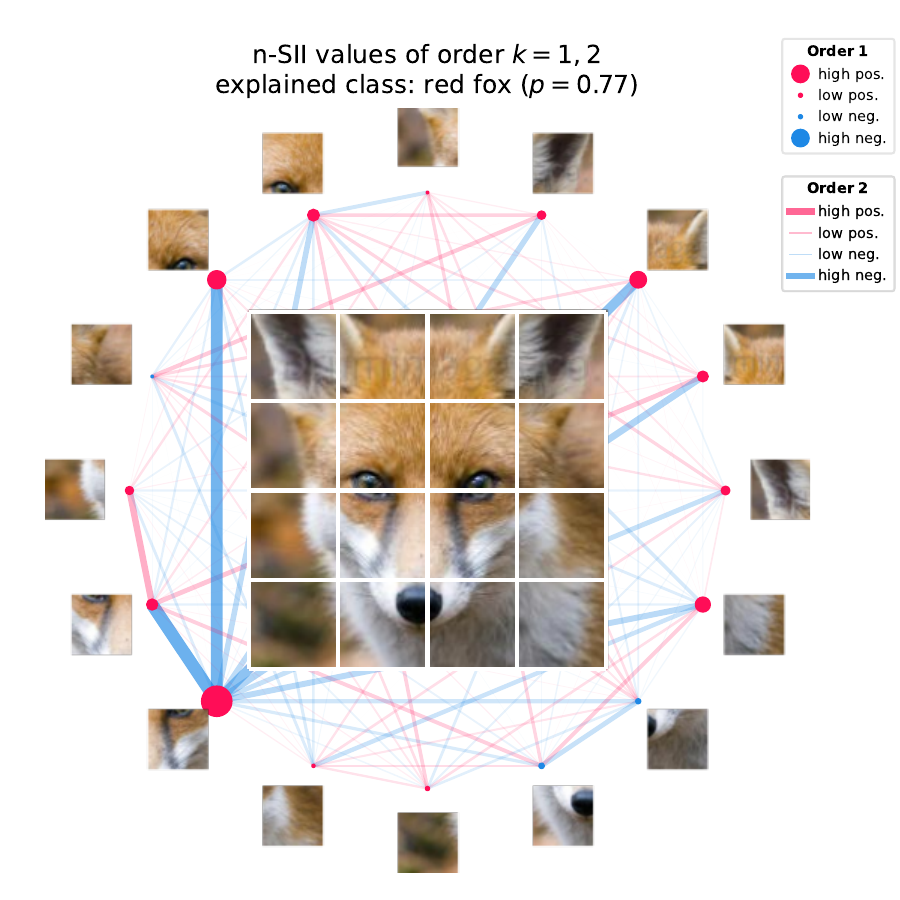} 
    \includegraphics[width=0.32\textwidth]{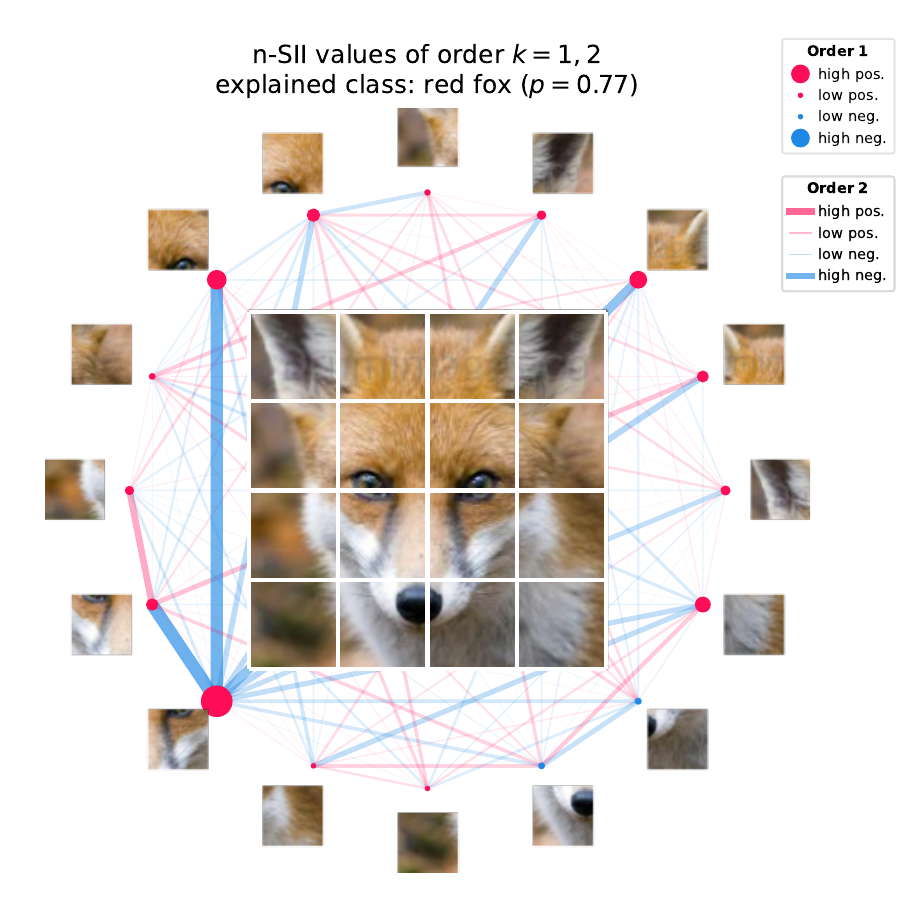}
    \includegraphics[width=0.32\textwidth]{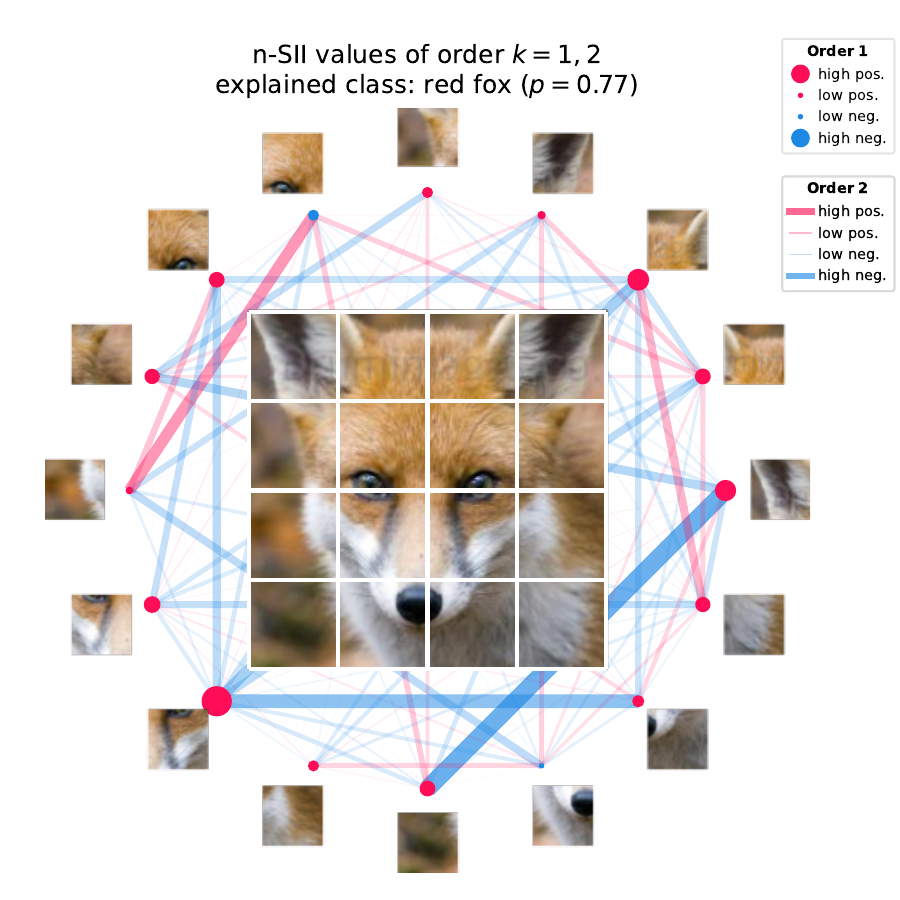} \\
     \includegraphics[width=0.32\textwidth]{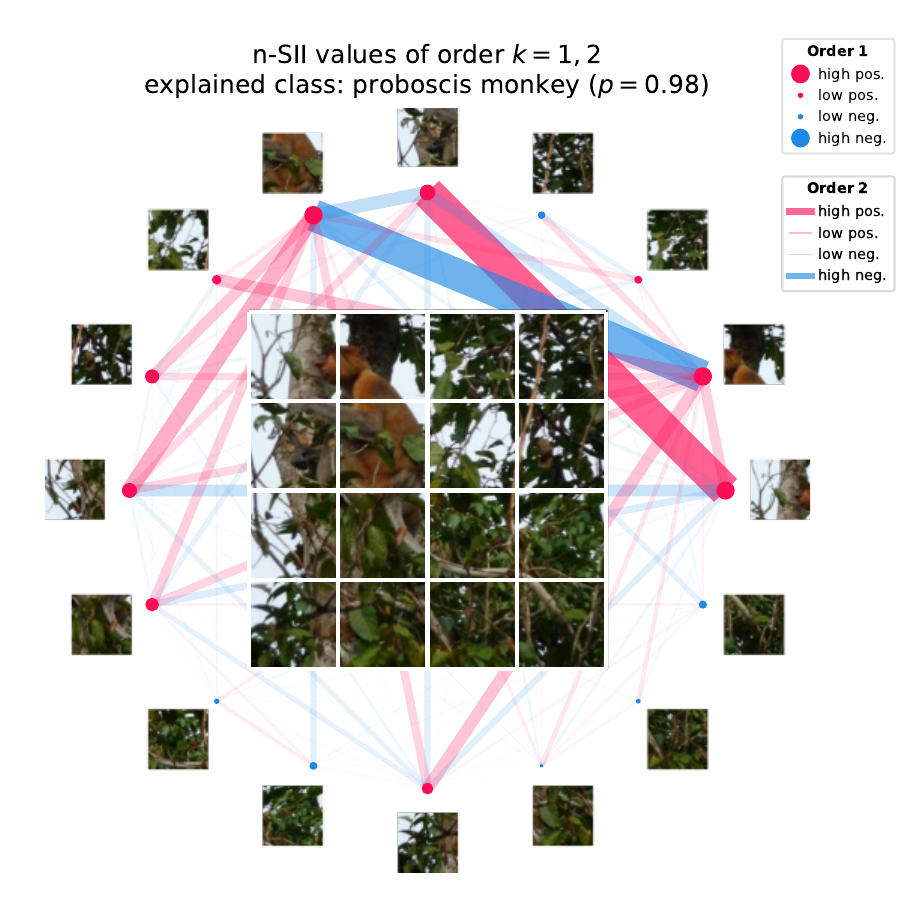} 
     \includegraphics[width=0.32\textwidth]{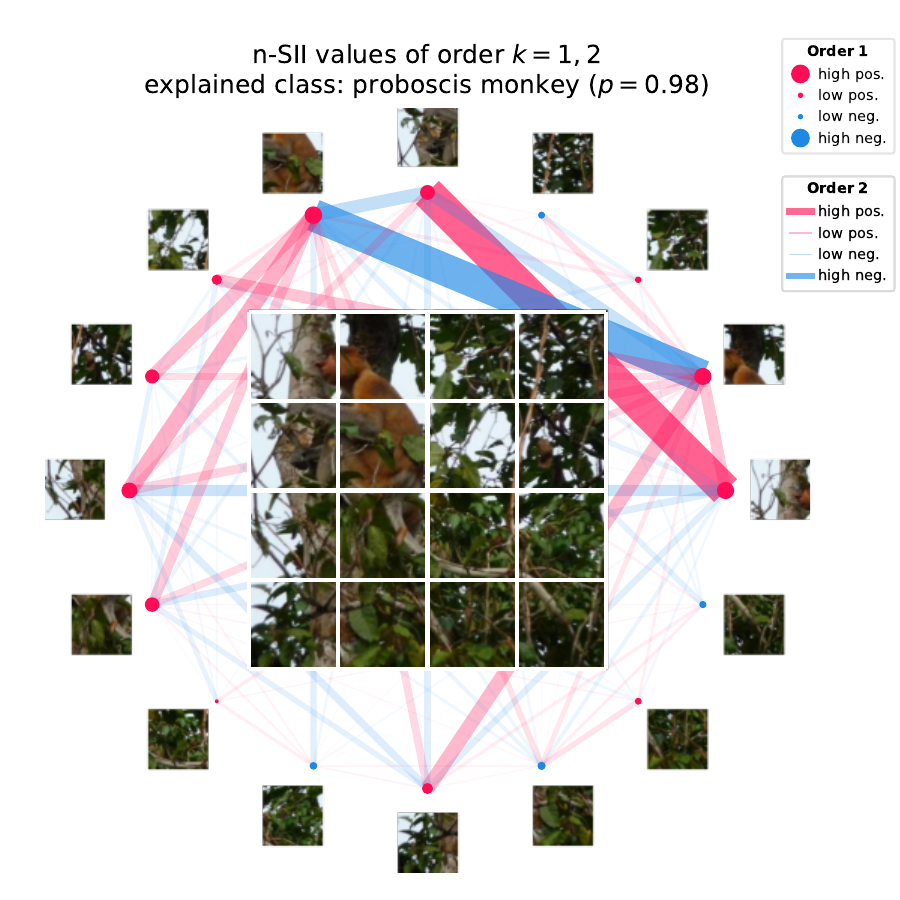}
     \includegraphics[width=0.32\textwidth]{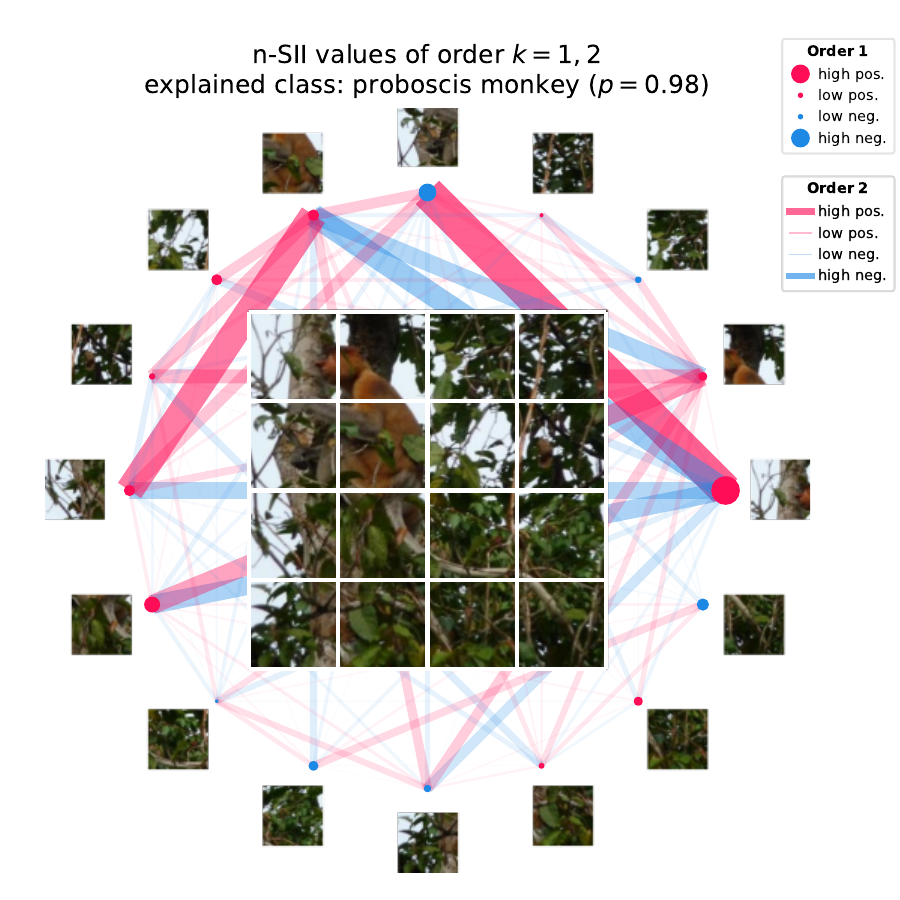}
     \\
     \includegraphics[width=0.32\textwidth]{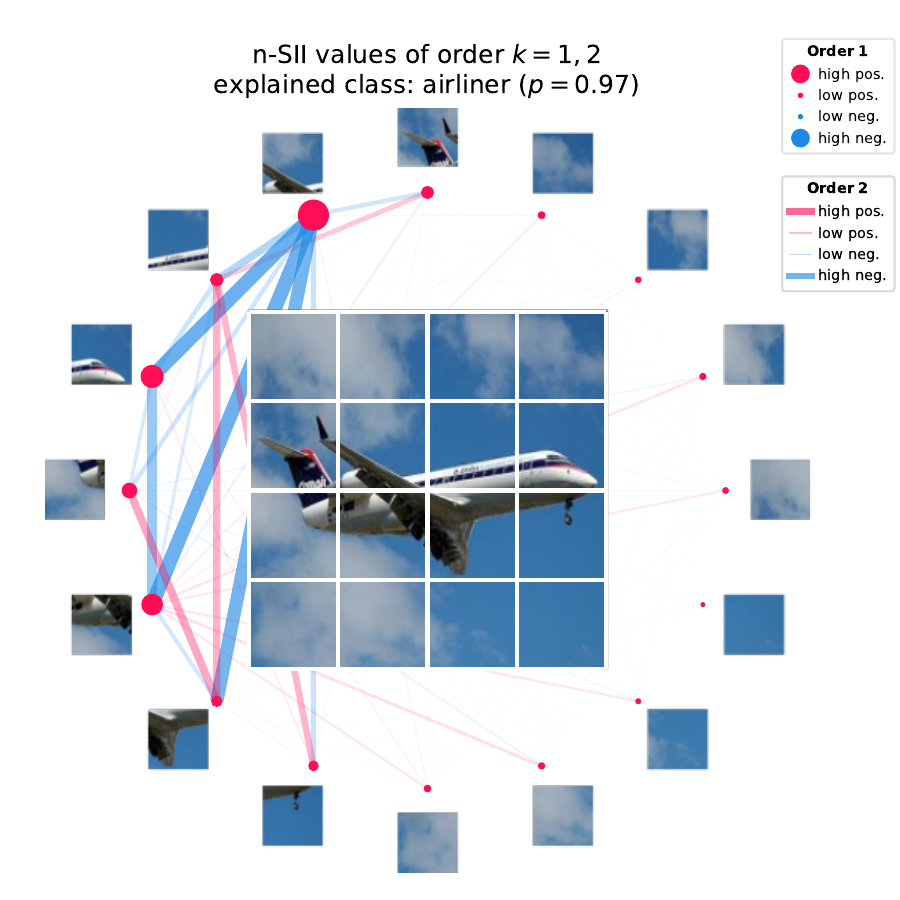}
     \includegraphics[width=0.32\textwidth]{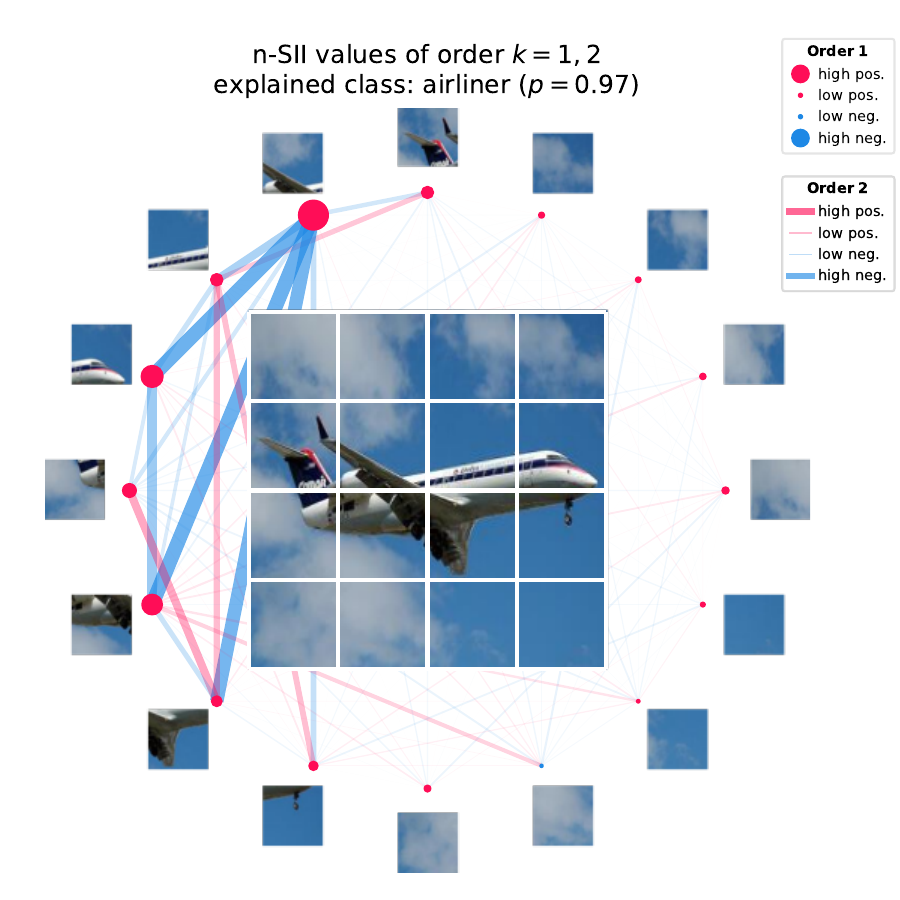}
     \includegraphics[width=0.32\textwidth]{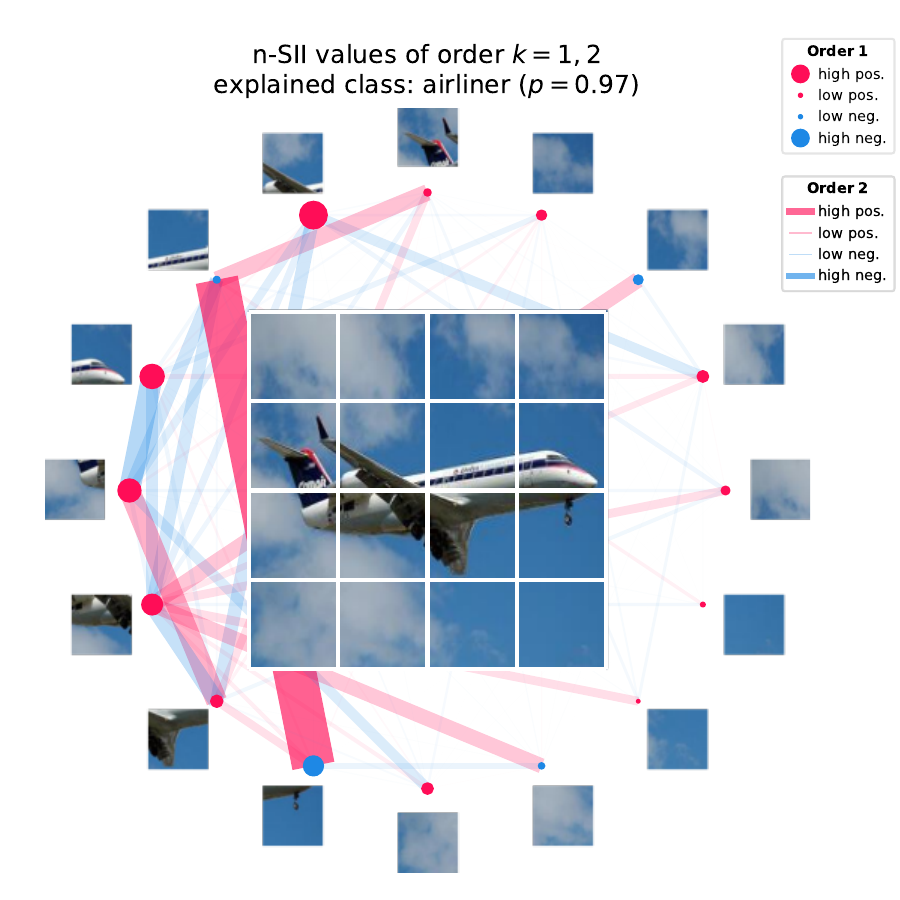}
    \caption{Row-wise comparison of ground-truth n-SII values of order $k=1,2$ for the predicted class probability of a ViT (first row) against n-SII values estimated by SVARM-IQ (second column) and permutation sampling (third row) with 5000 model evaluations. The pictures are taken from ImageNet \citep{ImageNet} and sliced into a grid of 16 patches ($n=16$).}
    \label{fig_network_comp}
\end{figure*}

\clearpage
\section{HARDWARE DETAILS}
\label{app:hardware_details}

This section contains the hardware details required to run and evaluate the empirical results.
All experiments where developed and run on a single DELL XPS 15 9510 notebook with Windows 10 Education installed as the operating system.
This laptop contains one 11th Gen Intel(R) Core(TM) i7-11800H clocking at 2.30GHz base frequency, 16.0 GB (15.7 GB usable) of RAM, and a NVIDIA GeForce RTX 3050 Ti Laptop GPU.

The model-function calls were pre-computed in around 10 hours on the graphics card.
The evaluation of the approximation quality required around 50 hours of work on the CPU.
In total, running the experiments took around 50 hours on a single core (no parallelization) and 10 hours on the graphics card.


\end{document}